\documentclass[11pt]{article}

\title{Minimum Degree Spanning Tree:\\  $(1+\epsilon,1)$-Approximation in Near-Linear Time}

\author{Sayan Bhattacharya\thanks{University of Warwick, UK.} \and Ermiya Farokhnejad\thanks{University of Warwick, UK.} \and Thatchaphol Saranurak\thanks{University of Michigan, USA. \texttt{thsa@umich.edu}.
        Supported by NSF Grant CCF-2238138 and a Sloan Fellowship.} \and Haoze Wang\thanks{Peking University, China.}}
\date{}

\usepackage[linesnumbered,ruled,vlined]{algorithm2e}

\usepackage{fancyhdr}
\usepackage{blindtext}
\usepackage{amsthm}
\usepackage{amsmath}
\usepackage{amssymb}
\usepackage{thmtools}
\usepackage{thm-restate}
\usepackage{anyfontsize}
\usepackage{xstring}
\usepackage{xargs}
\usepackage{makeidx}
\usepackage{cancel}
\usepackage{chngcntr}
\usepackage{caption}
\usepackage{subcaption}
\usepackage{float}
\usepackage{pgf,tikz,pgfplots}
\usepackage{latexsym}
\usepackage{verbatim}
\usepackage{bookmark}
\usepackage{multirow}
\usepackage[toc]{appendix}
\usepackage[letterpaper, left=1in, right=1in, top=0.9in, bottom=0.9in]{geometry}

\usepackage{cleveref}

\usepackage{listings}
\usepackage{hyperref}

\hypersetup{
	colorlinks=true,
	linkcolor=blue,
	filecolor=magenta,      
	urlcolor=cyan,
}

\usepackage{algorithmicx}
\usepackage{algpseudocode}
\usepackage[dvipsnames]{xcolor}

\usepackage{bidicontour}

\newtheorem{theorem}{Theorem}[section]
\newtheorem{lemma}[theorem]{Lemma}

\newtheorem{corollary}[theorem]{Corollary}

\newtheorem{definition}[theorem]{Definition}
\newtheorem{invariant}[theorem]{Invariant}
\newtheorem{remark}[theorem]{Remark}
\newtheorem{assumption}[theorem]{Assumption}
\newtheorem{observation}[theorem]{Observation}
\newtheorem{claim}[theorem]{Claim}

\hypersetup{ colorlinks, citecolor=blue}

\PassOptionsToPackage{unicode}{hyperref}

\usepackage{framed}
\usepackage[framemethod=tikz]{mdframed}
\usepackage{tikz-cd}
\newenvironment{wrapper}[1]
{
	\begin{center}
		\begin{minipage}{\linewidth}
			\begin{mdframed}[hidealllines=true, backgroundcolor=gray!20, leftmargin=0cm,innerleftmargin=0.4cm,innerrightmargin=0.4cm,innertopmargin=0.4cm,innerbottommargin=0.4cm,roundcorner=0pt]
				#1}
			{\end{mdframed}
		\end{minipage}
	\end{center}
}

\newcommand{\calP}[0]{\mathcal{P}}

\newcommand{\calS}[0]{\mathcal{S}}
\newcommand{\calF}[0]{\mathcal{F}}

\newcommand{\calC}[0]{\mathcal{C}}
\newcommand{\calA}[0]{\mathcal{A}}

\let\epsilon\varepsilon

\renewcommand{\th}[0]{\text{th}}

\newcommand{\old}[0]{\text{old}}
\newcommand{\new}[0]{\text{new}}
\newcommand{\add}[0]{\text{add}}
\newcommand{\mult}[0]{\text{mult}}

\newcommand{\lvl}[0]{\texttt{L}}

\newcommand{\inn}[0]{\text{in}}
\newcommand{\out}[0]{\text{out}}

\newcommand{\init}[0]{\text{init}}

\newcommand{\final}[0]{\text{final}}
\newcommand{\findChain}[0]{\textsc{find-chain}}
\newcommand{\backSearch}[0]{\textsc{backward-search}}

\newcommand{\bound}[0]{\text{bound}}
\newcommand{\ds}[0]{\texttt{DS}}
\newcommand{\deficit}[0]{\text{deficit}}
\newcommand{\conf}[2]{\left(\calF^{#1}_{#2}, B^{#1}_{#2}, \calC^{#1}_{#2}, D^{#1}_{#2}\right)}

\usepackage{ifthen}
\newboolean{debug}
\setboolean{debug}{false} %

\begin{document}

\maketitle

\pagenumbering{gobble}

\begin{abstract}
The minimum degree spanning tree problem is a classic NP-hard problem whose optimal approximation guarantee was established since the early 1990s: F\"urer and Raghavachari \cite{FR92} gave an $\tilde O(mn)$-time algorithm that computes a spanning tree with maximum degree $\Delta^\star+1$, where $\Delta^\star$ denotes the optimum value. Whether similarly strong guarantees can be achieved in near-linear time has remained open for over three decades.

We give the first near-linear-time algorithm that computes a spanning tree with maximum degree $\lceil (1+\epsilon)\Delta^\star\rceil+1$ in $\tilde O(m/\epsilon^2)$ time. Prior near-linear-time algorithms either achieved the weaker bound $\lceil (1+\epsilon)\Delta^\star\rceil + O(\log n/\epsilon^2)$ \cite{DHZ20} or required dense graphs with $m\ge n^{7/4}$ \cite{CQT21,BFW26}.

Using the same framework, our algorithm can also compute a spanning tree with maximum degree $\Delta^\star+1$ in $\tilde O(mn^{2/3})$ time, improving upon the recent $\tilde O(mn^{3/4})$-time algorithm of \cite{BFW26}. These two results strictly improve all previous construction algorithms for the minimum degree spanning tree problem.
\end{abstract}

\newpage

\tableofcontents
\addtocontents{toc}{\protect\setcounter{tocdepth}{2}}

\newpage
\pagenumbering{arabic}
\section{Introduction}

In the \emph{minimum degree spanning tree} (MDST) problem, we are given an undirected connected graph $G=(V,E)$ and seek a spanning tree $T$ minimizing its maximum degree, i.e., $\max_{u\in V}\deg_T(u)$. Let
\[
\Delta^{\star}:=\min_{T:\text{ spanning tree of }G}\max_{u\in V}\deg_T(u)
\]
denote the optimum value. The problem is NP-hard: deciding whether $\Delta^{\star}=2$ is equivalent to deciding whether $G$ contains a Hamiltonian path.

MDST is one of the most prominent examples of problems whose \emph{approximation guarantee} was essentially understood long ago, while its \emph{running time} remained open for decades. In a landmark result, Fürer and Raghavachari \cite{FR92,furer1994approximating} gave an algorithm that returns a spanning tree of maximum degree at most $\Delta^{\star}+1$ in $\tilde{O}(mn)$ time. This additive-one guarantee is best possible. Thus, from the approximation viewpoint, the problem was already nearly settled in the early 1990s. The main remaining challenge has been to obtain similarly strong degree guarantees in near-linear time.

Recent progress narrowed this gap only partially. Duan, He, and Zhang~\cite{DHZ20} gave the first near-linear-time algorithm, but with the weaker degree bound \((1+\varepsilon)\Delta^{\star}+O(\varepsilon^{-2}\log n)\). Chekuri, Quanrud, and Torres~\cite{CQT21} obtained a near-linear-time algorithm for estimating the optimum value and constructing a sparsifier, but not for constructing the spanning tree itself. Bhattacharya, Farokhnejad, and Wang~\cite{BFW26} later gave a \(\Delta^{\star}+1\) approximation in \(\tilde{O}(mn^{3/4})\) time, breaking the classical \(O(mn)\) barrier. Running this algorithm on the sparsifier of \cite{CQT21} also yields a spanning tree of degree \(\lceil(1+\varepsilon)\Delta^{\star}\rceil+2\) in \(\tilde{O}((m+n^{7/4})/\varepsilon^{2})\) time, which is near-linear only when \(m\ge n^{7/4}\). Hence, before our work, no near-linear-time algorithm was known to construct a spanning tree with maximum degree \(O(\Delta^{\star})\); indeed, none achieved even an \(o(\Delta^{\star}\log n)\) guarantee. 

\paragraph{Our Results.}
We show that strong degree guarantees are achievable in near-linear time by giving a unified framework that strictly improves \emph{all} prior construction algorithms for MDST.

\begin{theorem}
\label{thm:main1}
There exists a deterministic algorithm that, given an undirected graph and $\varepsilon\in(0,1]$, computes a spanning tree with maximum degree at most $\left\lceil (1+\varepsilon)\Delta^{\star}\right\rceil +1$ in $\tilde{O}(m/\varepsilon^{2})$ time.
\end{theorem}

This is the first near-linear-time algorithm that constructs a spanning tree with maximum degree $O(\Delta^{\star})$, resolving an open question raised by \cite{DuanP20,P16,S24}. 
In fact, it gives the first near-linear-time algorithm with an $o(\Delta^{\star}\log n)$ guarantee.
Note that for $\Delta^\star = \mathrm{polylog}(n)$, \Cref{thm:main1} gives a spanning tree with degree $\Delta^\star +2$ in near-linear time by setting $\epsilon = 1/\Delta^\star$.

The same framework also yields a faster additive-one approximation.
\begin{theorem}
\label{thm:main2}
There exists a deterministic algorithm that, given an undirected graph, computes a spanning tree with maximum degree at most $\Delta^{\star}+1$ in $\tilde{O}(mn^{2/3})$ time.
\end{theorem}

This improves both the classical $\tilde{O}(mn)$-time algorithm of \cite{furer1994approximating} and the recent $\tilde{O}(mn^{3/4})$ bound of \cite{BFW26}. See \Cref{tab:literature} for a comparison with previous results.

\setlength{\arrayrulewidth}{0.1mm}
\setlength{\tabcolsep}{5pt}
\renewcommand{\arraystretch}{1.3}
\begin{table}[h]
\footnotesize
\centering
\begin{tabular}{|c|c|c|c|}
\hline
 & Maximum Degree & Time & Note\tabularnewline
\hline
\hline
\cite{furer1990nc} & $O(\Delta^{\star}\log n)$ & $\text{poly}(n)$ & \tabularnewline
\hline
\cite{FR92} & $O(\Delta^{\star}+\log n)$ & $\text{poly}(n)$ & \tabularnewline
\hline
\cite{FR92,furer1994approximating} & $\Delta^{\star}+1$ & $O(mn)$ & \tabularnewline
\hline
\cite{DHZ20} & $\left\lceil (1+\epsilon)\Delta^{\star}\right\rceil +O(\frac{\log n}{\epsilon^{2}})$ & $\tilde{O}(m/\epsilon^{7})$ & \tabularnewline
\hline
\multirow{2}{*}{\cite{CQT21}} & $\left\lceil (1+\epsilon)\Delta^{\star}\right\rceil +1$ & $\tilde{O}(m/\epsilon^{2})$ & Compute only a value or a sparsifier\tabularnewline
\cline{2-4}
 & $\left\lceil (1+\epsilon)\Delta^{\star}\right\rceil +2$ & $\tilde{O}(n^{2}/\epsilon^{2})$ & Run \cite{furer1994approximating} on \cite{CQT21} sparsifier\tabularnewline
\hline
\multirow{2}{*}{\cite{BFW26}} & $\Delta^{\star}+1$ & $\tilde{O}(mn^{3/4})$ & \tabularnewline
\cline{2-4}
 & $\left\lceil (1+\epsilon)\Delta^{\star}\right\rceil +2$ & $\tilde{O}((m+n^{7/4})/\epsilon^{2})$ & Run \cite{BFW26} on \cite{CQT21} sparsifier\tabularnewline
\hline
\multirow{2}{*}{\textbf{Ours}} & $\left\lceil (1+\epsilon)\Delta^{\star}\right\rceil +1$ & $\tilde{O}(m/\epsilon^{2})$ & \tabularnewline
\cline{2-4}
 & $\Delta^{\star}+1$ & $\tilde{O}(mn^{2/3})$ & \tabularnewline
\hline
\end{tabular}
\caption{Minimum degree spanning tree algorithms in the literature.\label{tab:literature}}
\end{table}

\paragraph{Broader Context.}

MDST is arguably the most basic problem in degree-constrained network design. The minimum-cost variant of MDST \cite{RMRRH93,KR00,CRRT05,Goemans06,singh2015approximating} and its generalizations---including bounded-degree Steiner tree, $k$-connected subgraphs \cite{lau2013additive,fukunaga2015iterative}, arborescences \cite{bansal2008additive}, and dijoins \cite{kiraly2012degree}---have played a central role in the development of iterative rounding techniques \cite{LauRS11}, which usually take large polynomial time. In contrast, the fast-algorithm side of this area remains far less developed. Our work suggests that strong degree guarantees in degree-constrained network design may be compatible with near-linear-time algorithms.

\paragraph{Our Unified Framework.}

Both \Cref{thm:main1,thm:main2} follow from the same core subroutine. Suppose we are given a forest $\mathcal{F}$ of maximum degree at most $k$ with $f$ connected components. For a parameter $H\ge 1$, the goal is to reduce $f$ quickly by repeatedly using \emph{augmenting chains} of length at most $H$, where an augmenting chain is analogous to an augmenting path in the maximum flow problem.
We show, formally in \Cref{lem:main}, how to reduce the number of components by $\Omega(f/H)$ in $\tilde{O}(mH)$ time.

This improves the recent subroutine of \cite{BFW26}, which in $\tilde{O}(mH)$ time reduces the number of components by only $\Omega(f/(\theta H))$, where $\theta=n/f$ is the average component size. At a high level, both algorithms rely on two ingredients: a progress step that merges forest components via an augmenting chain, and an auxiliary step that \emph{merges reducible regions} to help expose such chains. Informally, a region is reducible if the forest can be locally rearranged inside the region so as to make any prescribed node in the region slack, without violating the degree bound.

The algorithm of \cite{BFW26} performs auxiliary merging aggressively as long as reducible regions have size at most $O(\theta)$, and only then extracts augmenting chains of length at most $H$. Each augmenting chain renders $O(H\theta)$ vertices unusable, so their algorithm reduces the number of components by only $\Omega(f/(\theta H))$.

\begin{remark}
    In Appendix \ref{appendix-previous-work-barriers}, we provide a comprehensive explanation of the barriers of the previous works.
    Specifically, in Appendix \ref{appendix:hard-instance}, we show a concrete example that \cite{BFW26} inherently takes super-linear time, even if we allow $O(\Delta^{\star})$ degree bound.
    The intuition of our main object (augmenting chain) and how it improves \cite{BFW26}, is explained in Appendix \ref{sec:techniques-overview}.
\end{remark}

In contrast, our algorithm prioritizes progress: for each threshold $t=1,2,\dots,H$, it first searches for augmenting chains of length at most $t$, and invokes merging only when this search fails. Although chains found in round $t$ may render some regions unusable, we can restore all regions at the start of round $t+1$ in $\tilde{O}(m)$ time. Our new analysis, based on a relaxed notion of augmenting chains called \emph{pseudo-chains}, shows that after $H$ rounds we reduce the number of forest components by $\Omega(f/H)$ in $\tilde{O}(mH)$ total time.

\paragraph{Roadmap for the rest of the paper.}

The paper is split into two parts. \Cref{part:EA} presents an extended abstract of the main ideas behind our unified framework. Its goal is to explain the proof of \Cref{lem:main} -- the key ingredient behind both of our main results -- by introducing the main objects and explaining the potential-based analysis. \Cref{part:full} provides the full technical development, including the precise definitions, the detailed implementation of the augmenting-chain subroutines, and the complete proofs of correctness and running time.

\newpage
\part{Extended Abstract}\label{part:EA}

\section{Notations and Preliminaries}
\label{new:prelim}
Let $G = (V, E)$ be the connected input graph with $|V| = n$ nodes and $|E| = m$ edges.  For any subgraph $K \subseteq G$, we let $V(K)$ and $E(K)$ respectively denote the set of nodes and the set of edges of $K$. 
We denote by $\Delta^\star$ and $k$, the maximum degree of the optimal spanning tree of $G$ and the spanning tree of $G$ returned by our algorithm, respectively. Thus, either $k := \lceil (1+\epsilon) \cdot \Delta^\star \rceil + 1$ for some arbitrarily small constant $\epsilon \in (0,1)$, or $k := \Delta^\star+1$ (see \Cref{thm:main1} and \Cref{thm:main2}). Throughout this extended abstract, whenever we refer to a forest $\calF$ of $G$, we assume that $V(\calF) = V$. We say that a node $v \in V$ is {\bf $k$-saturated} w.r.t.~a forest $\calF$ of $G$ if $\deg_{\calF}(v) = k$, and {\bf $k$-slack} if $\deg_{\calF}(v) < k$, where $\deg_{\calF}(v)$ denotes the degree of $v$ in $\calF$. Furthermore, we say that a forest $\calF$ of $G$ is {\bf $k$-feasible} iff $\deg_{\calF}(v) \leq k$ for all $v \in V$, i.e., every node is either saturated or slack w.r.t.~$\calF$. Throughout this paper, \(k\) will be fixed by the algorithm, and hence we will usually omit \(k\) from the terminology. We want to compute a feasible spanning tree of $G$.

\begin{remark}
    \label{rm:binary-search} Throughout this extended abstract, we assume that the value of $\Delta^\star$ is known to the algorithm. We can get rid of the assumption by performing a simple binary search, which incurs an extra $O(\log \Delta^\star)$ factor in the total runtime (see Appendix~\ref{appendix:unknown-Delta} for details).
\end{remark}

Consider any forest $\calF$ of $G$, and any pair of nodes $u, v \in V$ that are in the same connected component of $\calF$. Then, we let $P_{u, v}^{\calF}$ denote the unique path in $\calF$ connecting $u$ and $v$. %
For any subset of nodes $B \subseteq V$, we refer to each connected component of $\calF[V- B]$ as an {\bf $(\calF, B)$-region}, where $\calF[U]$ denotes the subgraph of $\calF$ induced by a subset of nodes $U \subseteq V$. 
Now, consider any node $v \in V$. Let $N_{\calF}(v) \subseteq V$ denote the set of neighbors of $v$ in $\calF$ (excluding the node $v$ itself). We define the {\bf boundary} of the node $v$ w.r.t.~$(\calF, B)$, denoted by $\bound_{(\calF, B)}(v)$, to be the collection of all $(\calF, B)$-regions $T$ which satisfy $V(T) \cap (N_{\calF}(v) \cup \{v\}) \neq \emptyset$. Finally, for any simple path $P$ in $\calF$, we define the {\bf boundary} of $P$ w.r.t.~$(\calF, B)$ as $\bound_{(\calF, B)}(P) := \bigcup_{v \in V(P)} \bound_{(F, B)}(v)$.

\paragraph{Reducible Nodes.}

Let $T$ be any (connected) subtree of a feasible forest $\calF$ of $G$.
Consider the following procedure, which is essentially the algorithm of \cite{FR92} running on $T$.
We initialize the set $B \subseteq V(T)$ as the collection of all saturated nodes (w.r.t.~$\calF$) in $T$. Subsequently, as long as there exists a non-forest edge $(x, y) \in E - E(\calF)$ with $x, y \in V(T) - B$, we set $B \leftarrow B - V\left(P_{x,y}^{\calF}\right)$. %

We say that $T$ is {\bf reducible} iff $B = \emptyset$ at the end of the above procedure.  Morally, if $T$ is reducible, then the following condition holds for every node $v \in V(T)$:  We can locally change $\calF$ by inserting/deleting some edges inside the subgraph of $G$ induced by $V(T)$, to achieve another feasible forest $\calF'$ where $v$ is slack.
This key property is summarized in \Cref{new:lem:degree-reduction}.

\iffalse
At the end of the above procedure, we refer to every $(\calF, )$

mark all of the saturated nodes w.r.t.~$T$  as  \textit{bad nodes}, and mark all $(T,B)$-regions as \textit{slack components} where $B$ is the set of bad nodes.
Then, as long as there exists a non-forest edge $(x,y) \in E(G[T]) - E(\calF)$ between two nodes $x$ and $y$ that are contained in two different slack components, we consider the path $P_{x,y}^T$ between $x$ and $y$ in $T$, and {\em merge} all of the slack components {\em hitting} or {\em adjacent to} $P_{x,y}^T$ together with all bad nodes on $P_{x,y}^T$, to form a new larger slack component.
Throughout this process, each slack component remains a sub-tree of $T$, and different slack components remain mutually node-disjoint separated by bad nodes.

\begin{definition}
    At the end of the above procedure, we refer to every node $v \in V(T)$ inside a  \textbf{reducible} node w.r.t.~$T$, and refer to every remaining bad node as a \textbf{non-reducible} node w.r.t.~$T$.
\end{definition}
\fi

\begin{lemma}[Degree Reduction: \cite{FR92}, and Lemma 2.1 in the arXiv version of \cite{BFW26}]\label{new:lem:degree-reduction}     
    There is a {\bf degree-reduction} subroutine which works as follows.
    Consider any feasible forest $\calF$ of $G$, any reducible subtree $C$ of $\calF$, and any node $u \in V(C)$.
    The subroutine takes $(C, u)$ as input, and modifies $\calF$ by inserting/deleting some edges $e \in E$ whose {\em both} endpoints lie in $V(C)$. 
    The subroutine runs in $\tilde{O}\left(\sum_{x \in V(C)} \deg_G(x)\right)$ time,\footnote{More specifically, we need $\tilde{O}\left(\sum_{x \in V(C)} \deg_G(x)\right)$ time to identify the edge set $E_G(C) := \{ (x, y) \in E : x, y \in V(C) \}$ consisting of the edges in $G$ with both endpoints in $V(C)$.
    After having $E_G(C)$ explicitly, the running time of the subroutine is $\tilde{O}(|E_G(C)|)$.} and returns an updated forest $\calF^+$ that guarantees:
    \begin{enumerate}
        \item $\deg_{\calF^+}(u) \leq k-1$.
        \item $\calF^+$ remains a feasible forest. 
    \end{enumerate}
\end{lemma}

\section{Our Framework}
\label{sec:building-blocks}

The main contribution of our paper is a unified algorithm for \Cref{lem:main}.
This theorem holds for both set of the following parameters.
\begin{equation}\label{eq:param}
    (k,H) := \left(\lceil (1+\varepsilon) \cdot \Delta^\star \rceil + 1, \Theta\left(\log_{1+\varepsilon} n\right)\right) \text{ or } (k,H) := \left(\Delta^\star + 1, \Theta(n/f)\right).
\end{equation}

\begin{theorem}\label{lem:main}
    There exists a deterministic algorithm that, given a feasible forest $\calF$ of $G$ with $f = \Omega(1)$ components, returns another feasible forest with at least $\Omega(f/H)$ many fewer components. The algorithm takes $\tilde{O}(m  H)$ time.
    Here, $k$ and $H$ are defined as in \Cref{eq:param}.
\end{theorem}

We will combine the theorem above with a basic subroutine that reduces the number of forest components by an additive one in near-linear time, as summarized in \Cref{lem:reduce-one}.

\begin{lemma}[\cite{FR92}, and Lemma 3.1  \cite{BFW26}]\label{lem:reduce-one} For any $k \ge \Delta^\star + 1$, there exists a deterministic algorithm that, given a feasible forest $\calF$ of $G$ with at least two components, returns another feasible forest with one fewer component. The algorithm runs in $\tilde{O}(m)$ time.
\end{lemma}

\noindent
\textbf{Proof of  \Cref{thm:main1}.}
Fix  $k :=  \lceil (1+\varepsilon)\Delta^\star \rceil + 1$ and $H := \Theta(\log_{1+\varepsilon} n)$.
We initialize an empty forest $\calF$, with $V(\calF) = V$ and $E(\calF) = \emptyset$.
As long as $\calF$ has $f = \Omega(1)$ components, we repeatedly call \Cref{lem:main} to reduce its number of components to 
$f - \Omega(f/H) = f(1-\Omega(\varepsilon/ \log n))$.
Each call takes $\tilde{O}(m \cdot H) = \tilde{O}(m\cdot \varepsilon^{-1})$ time.
After  $O(\varepsilon^{-1}\log^2 n)$ iterations, the number of components of  $\calF$ eventually reduces to $O(1)$.
Then, by repeatedly calling \Cref{lem:reduce-one} $O(1)$ times, we get a feasible spanning tree of $G$.
The total runtime of this algorithm is $\tilde{O}(m\cdot \varepsilon^{-2})$.

\medskip
\noindent
\textbf{Proof of \Cref{thm:main2}.}
Fix  $k := \Delta^\star + 1$ and $H :=  \Theta(n/f)$.
We initialize an empty forest $\calF$, with $V(\calF) = V$ and $E(\calF) = \emptyset$.
As long as $\calF$ has $f \geq n^{2/3}$ components, we repeatedly call \Cref{lem:main} to reduce its number of components to
$f - \Omega(f/H) = f(1-\Omega(n^{-1/3}))$. Thus, each call reduces the number of components by a $(1-\Omega(n^{-1/3}))$ factor, and takes $\tilde{O}(m \cdot H) = \tilde{O}(m n^{1/3})$ time.
After  $O(n^{1/3}\log n)$ iterations, the number of components of $\calF$ eventually reduces to less than $n^{2/3}$.
Then, by repeatedly calling \Cref{lem:reduce-one} at most $n^{2/3}$ times, we obtain a feasible spanning tree of $G$.
The total runtime of the algorithm is $\tilde{O}(m n^{2/3})$.

\medskip
For the rest of the extended abstract, we focus on explaining the proof of \Cref{lem:main}.

\section{Basic Building Blocks}
\label{new:sec:building:block}
In this section, we present an overview of two key objects that underpin our algorithm: a {\bf configuration} (see \Cref{new:def:config}) and  an {\bf augmenting chain} (see \Cref{new:def:simple:aug}).

\begin{definition}
\label{new:def:config}
    A {\bf configuration} of the input graph $G$ is a triple $(\calF, B, \calC)$, such that:
\begin{enumerate}
\item \label{config:1} $\calF$ is a feasible forest of $G$.
\item \label{config:2} $B \subseteq V$ is a subset of nodes. We refer to the nodes in $B$ as {\bf junctions}.
\item \label{config:3} $\calC$ is a subset of the collection of all $(\calF, B)$-regions, satisfying the property  that every $C \in \calC$ is reducible. We refer to the $(\calF, B)$-regions in $\calC$ as {\bf untouched}, and the remaining $(\calF, B)$-regions as {\bf touched}. For every node $v \in V$ that belongs to a touched (resp.~untouched) $(\calF, B)$-region, we refer to $v$ itself as {\bf touched} (resp.~{\bf untouched}). 
\end{enumerate}
\end{definition} 

According to this definition, the set of nodes of the graph is partitioned into junctions, touched nodes, and untouched nodes. 

We first introduce a useful notion before defining the augmenting chain.

\begin{definition}
\label{new:def:effective}
An \emph{ordered} pair of nodes $(u,v)$ is an {\bf effective} edge w.r.t.~a configuration $(\calF, B, \calC)$ iff $(u,v) \in E - E(\calF)$ is a non-forest edge connecting two different $(\calF, B)$-regions, $u$ is untouched, and $v$ is either [untouched] or [touched and slack].
\end{definition}

\input{figures/config}

Below, we say that a node $v \in V$ \textbf{separates} two mutually disjoint node-sets $X, Y \subseteq V - \{v\}$ in a forest $\calF$ iff the following two conditions hold: (i) All the nodes in $X \cup Y \cup \{v\}$ are in the same connected component of $\calF$. (ii) There are two distinct connected components of $\calF - \{v\}$, one containing all the nodes in $X$ and the other containing all the nodes in $Y$.  Here, $\calF - \{v\}$ is the subgraph of $\calF$ obtained by deleting $v$ and all its incident edges from $\calF$.

\begin{definition}
\label{new:def:simple:aug}
Fix any configuration $(\calF,B, \calC)$ and any integer $h\ge0$. An {\bf augmenting chain} of length $h+1$ w.r.t.\ $(\calF,B, \calC)$ is a tuple of nodes $(w_{0},z_{1},w_{1},z_{2},\ldots,z_{h},w_{h},z_{h+1})$ such that: 
\begin{enumerate}
\item \label{new:simple:aug:1} $(w_{h},z_{h+1})$ is an effective edge w.r.t. $(\calF,B,\calC)$.
\item \label{new:simple:aug:2} For every $i\in\{h,\ldots,1\}$, the nodes $z_i, w_{i-1}$ satisfy the following properties:
\begin{enumerate}
\item \label{new:simple:aug:3} $z_{i}\in B$ is a junction that separates $\{w_i\}$ and $\{z_{i+1}, w_{i+1}, \ldots, z_h, w_h, z_{h+1}\}$ in $\calF$.
\item \label{new:simple:aug:4} $(w_{i-1},z_{i}) \in E - E(\calF)$ is a non-forest edge and $w_{i-1}$ is untouched.
\end{enumerate}
\item \label{new:simple:aug:5} All the nodes  $\{z_1, w_1, \ldots, z_h, w_h, z_{h+1}\}$ lie in the same connected component of $F$, whereas $w_0$ lie in a different connected component of $\calF$.
\end{enumerate}
\end{definition}

While the above definition might look daunting at first, there happens to be a relatively intuitive algorithm for generating an augmenting chain w.r.t.~a given configuration, which works as follows.

\begin{wrapper}
We initialize a sequence $\Gamma = \left(w^\old, z^\old \right)$, where the ordered pair $\left(w^\old, z^\old\right)$ is any arbitrary effective edge.
This initial edge corresponds to $(w_h,z_{h+1})$ in the final augmenting chain (see Property \ref{new:simple:aug:1} of \Cref{new:def:simple:aug}), and the construction is downwards from $i = h$ to $i=1$.
Subsequently, we proceed in iterations.

At the start of an iteration, we check if $w^\old$ and $z^\old$ are in different connected components of $\calF$. If yes, then we terminate and return the sequence $\Gamma$ (see Property \ref{new:simple:aug:5} of \Cref{new:def:simple:aug}). Otherwise, we try to find two nodes $w^\new, z^\new$ that satisfy the following conditions: 
\begin{itemize}
\item $z^\new \in B$ and it separates $\{w^\old\}$ and $\Gamma - \{ w^\old\}$ in $\calF$ (see Property \ref{new:simple:aug:3} of \Cref{new:def:simple:aug}).\footnote{Here, we slightly abuse notation and treat $\Gamma$ as a set as opposed to a sequence.}
\item   $(w^\new, z^\new) \in E - E(\calF)$ and $w^\new$ is untouched (see Property \ref{new:simple:aug:4} of \Cref{new:def:simple:aug}).
\end{itemize}
We add $(u^\new, v^\new)$ to the front of  $\Gamma$. The next iteration starts with $u^\old \leftarrow u^\new$, $v^\old \leftarrow v^\new$.
\end{wrapper}

It is easy to see that if the above algorithm terminates, then it returns an augmenting chain $\Gamma$.

Our algorithm will start with an initial configuration, and repeatedly {\bf apply} an augmenting chain. This, in turn, will change the underlying  configuration,  as specified in the lemma below.

\begin{lemma}\label{new:lem:apply-chain}
    Let $\Gamma$ be any augmenting chain of length $h+1$ w.r.t.~a configuration $(\calF, B, \calC)$.
    When we {\bf apply} $\Gamma$ on $(\calF, B, \calC)$, we get  another configuration $(\calF', B', \calC')$ with the following guarantees.
    \begin{enumerate}
        \item\label{new:lem:apply-chain-P1} $\calF'$ has exactly one fewer component than $\calF$.
        
        \item\label{new:lem:apply-chain-P2} $B = B'$. Furthermore, there are at most $2h$ edges $(u,v) \in E(\calF') \oplus E(\calF)$ with $\{u,v\} \cap B \neq \emptyset$. {\em (We refer to such edges as \textbf{critical} edges of $\Gamma$.)}

        \item\label{new:lem:apply-chain-P4} $\calC' \subseteq \calC$ and $|\calC - \calC'| \leq h+2$. 
    \end{enumerate}
\end{lemma}

\paragraph{Applying an Augmenting Chain.} Let  $\Gamma = (w_0, z_1, w_1, \ldots, z_h, w_h, z_{h+1})$ be any augmenting chain w.r.t.~a configuration $(\calF, B, \calC)$. Consider any index $i \in [0, h]$.
Note that $w_i$ is untouched (see Properties \ref{new:simple:aug:1} and \ref{new:simple:aug:4} of \Cref{new:def:simple:aug}). 
Let $W_i$ be the unique untouched $(\calF, B)$-region containing $w_i$. Note that $W_i$ is reducible, by Property \ref{config:3} of \Cref{new:def:config}.
Similarly, if $z_{h+1}$ is untouched, then let $Z_{h+1}$ denote the unique untouched $(\calF, B)$-region containing $z_{h+1}$. Again, Property \ref{config:3} of \Cref{new:def:config} guarantees that $Z_{h+1}$ is reducible. Otherwise, let $Z_{h+1} := \{z_{h+1}\}$,  which is reducible as well since $z_{h+1}$ must be slack (due to Property \ref{new:simple:aug:1} of \Cref{new:def:simple:aug}).
Using Property \ref{new:simple:aug:3}  of \Cref{new:def:simple:aug}, we can show that the subtrees $\{W_0, \ldots, W_h, Z_{h+1}\}$ are mutually node-disjoint.
This motivates the following natural procedure for {\bf applying} $\Gamma$.
\begin{wrapper}
\begin{enumerate}
\item \label{new:apply:1} Apply \Cref{new:lem:degree-reduction} on $(Z_{h+1}, z_{h+1})$, so that $z_{h+1}$ becomes slack and $\calF$ remains a feasible forest. If $z_{h+1}$ was untouched before we performed this step, then set $\calC \leftarrow \calC -\{Z_{h+1}\}$, i.e., the $(\calF, B)$-region $Z_{h+1}$ becomes touched from this point onward.
\item \label{new:apply:2} For each $i \in \{h, \ldots, 0\}$: 
Apply \Cref{new:lem:degree-reduction} on $(W_i, w_i)$, so that $w_i$ becomes slack and $\calF$ remains a feasible forest. Set $\calC \leftarrow \calC - \{W_i\}$, i.e., the $(\calF, B)$-region $W_i$ becomes touched from now on.
\item \label{new:apply:4} For each $i \in \{h, \ldots, 1\}$: 
Set $E(\calF) \leftarrow E(\calF) - \{(z_i, y_i)\} + \{(w_i, z_{i+1})\}$, where $(z_i, y_i)$ is the unique edge incident on $z_i$ on the path $P_{z_i, w_i}^{\calF}$.
\item \label{new:apply:6} Set $E(\calF) \leftarrow E(\calF) + \{(w_0, z_1)\}$.
\end{enumerate}
\end{wrapper}

\paragraph{Proof (Sketch) of \Cref{new:lem:apply-chain}.} Since the subtrees $\{W_0, \ldots, W_h, Z_{h+1}\}$ are reducible and node-disjoint, we can consistently perform the Steps~\ref{new:apply:1} and~\ref{new:apply:2} in the procedure described above. 
At this point, we have a feasible forest such that all the nodes in $\{w_0, w_1, \ldots, w_h, z_{h+1}\}$ are slack.

Next, we focus on Steps~\ref{new:apply:4} and~\ref{new:apply:6} of the above procedure. Using Property \ref{new:simple:aug:3}  of \Cref{new:def:simple:aug}, we can show that these steps do not create any cycle in $\calF$.
They, however, increase the degree (w.r.t.~$\calF$) of each of the nodes in $\{w_0, w_1, \ldots, w_h, z_{h+1}\}$ by one. But, all these nodes are slack at the end of Step~\ref{new:apply:2}.
Thus, we infer that $\calF$  remains a feasible forest by the end of Step~\ref{new:apply:6}. 

Steps~\ref{new:apply:1} and~\ref{new:apply:2} do not change the number of components of $\calF$.
In Steps~\ref{new:apply:4} and~\ref{new:apply:6}, we inserted $h+1$ edges into the forest and removed $h$ edges from it.
Hence, by Property \ref{new:simple:aug:5} of \Cref{new:def:simple:aug}, the number of components of $\calF$ reduces by one as we apply $\Gamma$.
This implies guarantee \ref{new:lem:apply-chain-P1} of \Cref{new:lem:apply-chain}.

Next, the set of edges incident on $B$, which get inserted into/deleted from $\calF$ while applying  $\Gamma$, is given by: $\bigcup_{i=1}^h \{ (w_{i-1}, z_i), (z_i, y_i)\}$. This implies guarantee  \ref{new:lem:apply-chain-P2} of \Cref{new:lem:apply-chain}.

Finally, observe that Step~\ref{new:apply:1} reduces the number of untouched regions by at most one, whereas Step~\ref{new:apply:2} reduces the number of touched regions by $h+1$. In other words, as we apply $\Gamma$, the number of untouched regions decreases by at most an additive $h+2$. This implies guarantee \ref{new:lem:apply-chain-P4} of \Cref{new:lem:apply-chain}.

\begin{figure}[ht]
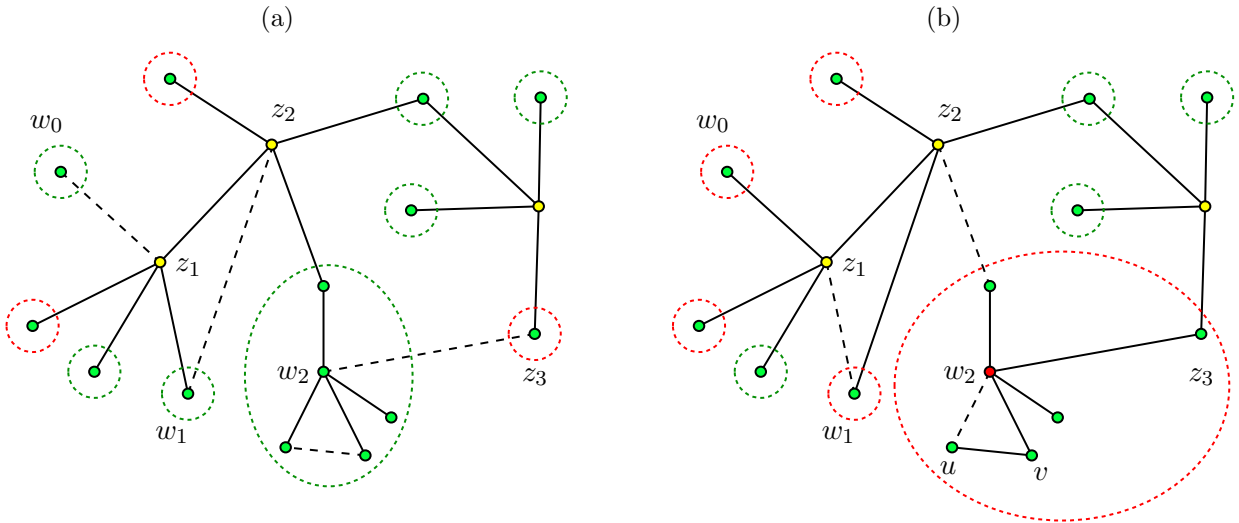

\caption{
An example of the augmenting chain. Solid lines are forest edges and dashed lines are non-forest edges.
The value of $k$ is $4$.
Untouched and touched regions are depicted with green and red ovals, respectively.
Yellow nodes correspond to junctions.
Green nodes correspond to either [untouched] or [touched and slack] nodes.
(a) The status of the forest before applying the augmenting chain $(w_0, z_1, w_1, z_2, w_2, z_3)$.
(b) The status of the forest after applying the augmenting chain. Regions containing $w_0$, $w_1$, $w_2$ become touched. The degree reduction process on $w_2$ inserted the edge $(u,v)$ and removed $(w_2,u)$. $w_2$ is touched and saturated after applying the chain.}
  \centering
  \subfloat[]{\input{figures/augchain-before}\label{fig:augchain-before}}
  \hfill
  \subfloat[]{\input{figures/augchain-after}\label{fig:augchain-after}}
  \label{fig:atoms-deg-reduction}
\end{figure}

\section{Our Algorithmic Template}
\label{new:sec:template}

\subsection{A Key Subroutine: Applying an Effective Edge}
\label{new:sec:effective}

Consider an effective edge $(u,v)$ w.r.t.~a configuration $(\calF, B, \calC)$ (see \Cref{new:def:effective}), and any integer $t \geq 1$. We use the phrase {\bf applying the effective edge $(u, v)$ up to length $t$} to denote the key subroutine used by our algorithm, which works as follows.

First, we call an abstract subroutine $\findChain(t, (u, v))$, which will be described in \Cref{new:sec:findchain}.
If the call succeeds, then it returns an augmenting chain $\Gamma = (w_0, z_1, w_1, \ldots, z_h, w_h, z_{h+1})$ of length $h+1 \leq t$ such that $w_h = u$ and $z_{h+1} = v$.
In this case, we refer to $(u, v)$ as a {\bf improving edge} w.r.t.~$(\calF, B, \calC)$. 
Otherwise, if the call to $\findChain(t, (u, v))$ fails, then we are guaranteed that $u$ and $v$ lie in the same connected component of $\calF$.\footnote{Since $\Gamma = (u,v)$ itself is an augmenting chain if $u$ and $v$ belong to different connected components of $\calF$, and $\findChain(t,(u,v))$ will not miss it.}
In this latter event, we consider the path $P_{u, v}^{\calF}$, and fork into two sub-cases. If $\bound_{(\calF,B)}(P^\calF_{u,v})$ does \emph{not} contain any touched $(\calF, B)$-regions, then we refer to $(u,v)$ as a \textbf{merging edge}; else we refer to $(u, v)$ as a \textbf{touching edge}.

Moreover,  we define the set
 \begin{equation}\label{new:eq:W}
X := \left( V\left(P^{\calF}_{u,v}\right) \right) \bigcup \left( \bigcup_{C \in \mathcal{Y}} V(C) \right), \text{where } \mathcal{Y} = \bound_{(\calF, B)}\left(P^{\calF}_{u,v}\right)
\end{equation}
which consists of all the nodes that belong either to the path $P_{u, v}^{\calF}$, or to an $(\calF, B)$-region on the boundary of $P_{u, v}^{\calF}$. We now perform the following operations.

\begin{itemize}

\item []
\textbf{If $(u, v)$ is a improving edge}, then
we have successfully found an augmenting chain $\Gamma$ w.r.t.~$(\calF, B, \calC)$.
Accordingly, we apply the augmenting chain $\Gamma$, as described in \Cref{new:sec:building:block}.

\item []
\textbf{If $(u, v)$ is a merging edge}, then we set $\calC \leftarrow \left(\calC - \bound_{(\calF, B)}\left(P^{\calF}_{u,v}\right)\right) + \{\calF\left[X\right]\}$ and $B \leftarrow B - V\!\left(P^{\calF}_{u,v}\right)$, where $X$ is defined as in \Cref{new:eq:W}.

\item []
\textbf{If $(u, v)$ is a touching edge}, then we set $\calC \leftarrow \calC - \bound_{(\calF, B)}\left(P^{\calF}_{u,v}\right)$ and $B \leftarrow B - V\!\left(P_{u, v}^{\calF} \right)$. So, every node that is either on $P_{u, v}^{\calF}$ or belongs to a region in the boundary of $P_{u, v}^{\calF}$ is touched.

\end{itemize}
\begin{figure}[ht]
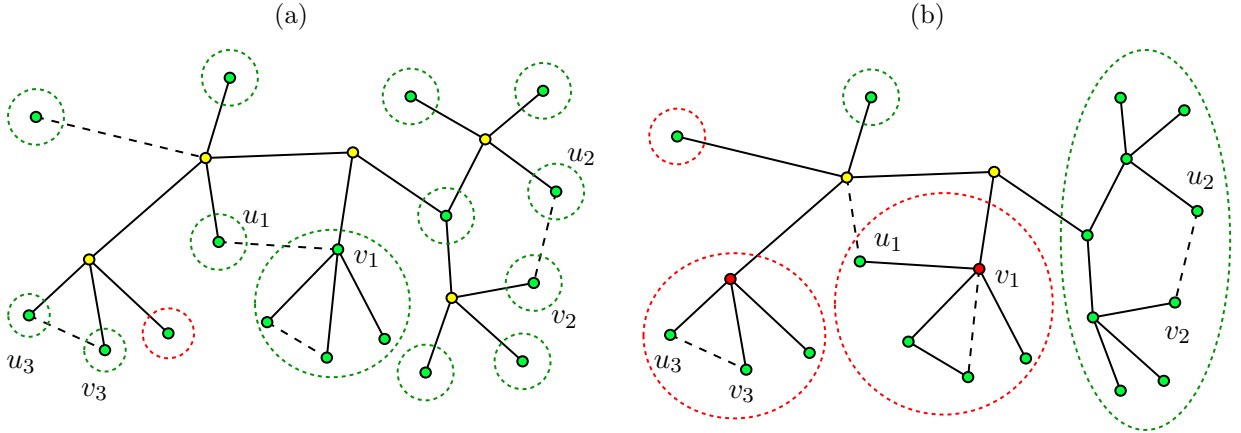

\caption{
An example of all types of effective edges.
Solid lines are forest edges and dashed lines are non-forest edges.
The value of $k$ is $4$.
Untouched and touched regions are depicted with green and red ovals, respectively.
Yellow nodes correspond to junctions.
Green nodes correspond to either [untouched] or [touched and slack] nodes.
(a) $(u_1,v_1)$ is an improving edge, $(u_2,v_2)$ is a merging edge, and $(u_3,v_3)$ is a touching edge.
(b) The status of the configuration after applying all of these effective edges. Red nodes here correspond to [touched and saturated] nodes.
}
  \centering
  \subfloat[]{\input{figures/effective-edges-before}\label{fig:effective-edges-before}}
  \hfill
  \subfloat[]{\input{figures/effective-edges-after}\label{fig:effective-edges-after}}
  \label{fig:atoms-deg-reduction}
\end{figure}

\begin{claim}\label{new:claim:apply-effective-remain-configuration}
    After applying an effective edge $(u, v)$ up to length $t$ w.r.t.~a configuration $(\calF, B, \calC)$, the updated triple $(\calF, B, \calC)$ continues to remain a valid configuration of the input graph $G = (V, E)$.
\end{claim}

\begin{proof}(Sketch)
    If $(u,v)$ is a improving edge, then the claim follows from the guarantee of applying an augmenting chain to a configuration (see \Cref{new:lem:apply-chain}).
    If $(u,v)$ is a merging edge, then it is easy to verify that the subtree $\calF\left[X\right]$ is reducible (see \Cref{new:prelim}), since every region in $\bound_{(\calF,B)}(P_{u, v}^{\calF})$ is untouched by definition.
    Finally, if $(u, v)$ is a touching edge, then we remove all the junctions on the path $P_{u, v}^{\calF}$ from the set $B$, and $\calF\left[X\right]$ becomes a touched $(\calF, B)$-region.
    The claim follows.
\end{proof}

\subsection{High-Level Description of Our Algorithm}

To describe our algorithmic template at a high-level, we need to introduce the concept of a {\bf $D$-avoiding augmenting chain}, as defined below.

\begin{definition}
\label{new:def:avoiding} Let  $\Gamma = (w_0, z_1, w_1, \ldots, z_h, w_h, z_{h+1})$ be any augmenting chain w.r.t.~a configuration $(\calF, B, \calC)$, and $D \subseteq V$ be any subset of nodes. We say that  $\Gamma$ is {\bf $D$-avoiding} iff $z_{h+1} \notin D$.
\end{definition}

Our algorithm proceeds in a sequence of {\bf rounds}. 
Each round breaks into multiple {\bf iterations}, where an iteration consists of applying an effective edge.
Throughout each round, we maintain a set $D$ of {\bf dirty nodes} satisfying Invariant \ref{inv:dirty}.
The reason for introducing $D$ is the time constraints in our algorithm.

\begin{invariant}\label{inv:dirty}
    Throughout each round, the set $D$ of dirty nodes consists of all touched nodes that have been saturated at the time of getting touched.
    Hence, $D = \emptyset$ at the beginning of each round and remains monotonically growing until the end of the round.
\end{invariant}

We only consider effective edges $(u,v)$ satisfying $v \notin D$, which concludes that our subroutine $\findChain(\cdot, \cdot)$ only searches for $D$-avoiding augmenting chains.
Upon applying such an effective edge in an iteration, some new touched regions might emerge.
In that case, we add those nodes $x$ to the set $D$ that satisfy {\em both} these two conditions: (i) $x$ becomes touched during the iteration and (ii) $x$ is saturated (i.e., $\deg_{\calF}(v) = k$) at the end of the iteration.
It is very well possible that some dirty nodes become slack in the future iterations, but we \textbf{keep them dirty} since we do not want to spend time on those nodes again.

\medskip
\noindent
\textbf{Algorithm Description.}
The algorithm runs in $H$ rounds, and maintains a configuration $(\calF, B, \calC)$ and a set $D$ of dirty nodes. 
We use the subscript $t$ to denote the status of any object during round $t \in [1, H]$.
We start with a configuration $(\calF, B, \calC)$, where $\calF$ is a feasible forest $\calF$ of $G$ that is part of the input, $B = \{ v \in V : \deg_{\calF}(v) = k\}$ is the set of saturated nodes in $\calF$, $\calC$ is the collection of all $(\calF, B)$-regions, and $D := \emptyset$.
Next, for consistency of notations, we assume that $\calF_{0}^\final := \calF$, $B_{0}^\final := B$, $\calC_{0}^\final := \calC$, and $D_0^\final := \emptyset$.

\medskip
\noindent
\textbf{Description of Round $t$.}
Below, we explain how to implement a round $t \in [1, H]$.

\medskip
\noindent
\textbf{Initialization of Round $t$.}
At the beginning of round $t$, we let $\calF^\init_t := \calF^\final_{t-1}$, $B^\init_t := B_{t-1}^\final \cup D^\final_{t-1}$, $D_t^\init := \emptyset$, and $\calC^\init_t$ be the collection of \emph{all} $(\calF^\init_t, B^\init_t)$-regions.

\medskip
\noindent
\textbf{Main Body of Round $t$.}
We start with $(\calF_t, B_t, \calC_t) := (\calF^\init_t, B^\init_t, \calC^\init_t)$ and $D_t := D_t^\init$.
Next, we repeatedly apply either a $D_t$-avoiding augmenting chain of length $\leq t$, a merging edge, or a touching edge, via the following {\bf while} loop.

\begin{wrapper}
While there is an effective edge $(u, v)$ w.r.t.~$(\calF_t, B_t, \calC_t)$ such that $v \notin D_t$:
\begin{itemize}
\item  Apply the effective edge $(u, v)$ up to length $t$, and update the configuration $(\calF_t,B_t,\calC_t)$ as per the procedure in \Cref{new:sec:effective}.
\item For every node $x \in V$ that becomes touched due to the previous step:
\begin{itemize}
\item If $x$ is saturated (i.e., $\deg_{\calF_t}(x) = k$),  then set  $D_t \leftarrow D_t \cup \{x\}$.
\end{itemize}
\end{itemize}
\end{wrapper}

When the above {\bf while} loop terminates, we finish round $t$, and we set   $\left(\calF^\final_t, B^\final_t, \calC^\final_t\right) := (\calF_t, B_t, \calC_t)$ and $D_t^\final := D_t$. In other words,  $(\calF^\final_t, B^\final_t, \calC^\final_t)$ and $D_t^\final$ respectively denotes the final configuration and the final state of the set of dirty nodes at the end of round $t$.

\begin{observation}\label{new:obs:cond-end-round}
    At the end of any round $t \in [1, H]$, there does not exist any effective edge $(u, v)$ w.r.t.~the final configuration $\left(\calF_t^\final, B_t^\final, \calC_t^\final\right)$ with $v \notin D_t^\final$.
\end{observation}

\noindent
\textbf{The Output.}
We return $\calF_{H}^\final$ as the final output of our algorithm for \Cref{lem:main}.

\section{Searching for Augmenting Chains}\label{sec:find-chain}

\textbf{Level of Nodes and Edges.}
Throughout our algorithm, for each node $x \in V$, we maintain an integer $\lvl_x \in [1, H+1]$.
Similarly, for each edge $(u,v) \in E$, we maintain an integer $\lvl_{(u,v)} \in [0, H+1]$.
We call $\lvl_x$ and $\lvl_{(u,v)}$ the \textbf{level} of $x$ and $(u,v)$, respectively.
At the start of round $1$, the levels of all nodes and edges are set to $1$ and $0$, respectively.
These levels keep increasing throughout the entire algorithm until the end of round $H$.

\subsection{Description of the Subroutine {\sc Find-Chain}$(h, (u, v))$}
\label{new:sec:findchain}
This subroutine takes as input: an integer $h \in [1, H]$, and an effective edge $(u,v)$ (the input being an effective edge is aligned with Property \ref{new:simple:aug:1} of \Cref{new:def:simple:aug}).
Its goal is to find an augmenting chain $A = (w_0, z_1, \ldots, z_r, u, v)$ of length at most $h$. It works as follows.

\begin{wrapper}
\begin{itemize}
\item [] {\bf While} $\lvl_{(u,v)} \leq h-1$:
\begin{itemize}
\item We call the subroutine $\backSearch$ on the edge $(u,v)$ (see \Cref{new:back:search:edge}).
\item If $\backSearch$ on $(u,v)$ is successful, then it returns an augmenting chain $A = (w_0, z_1, \ldots, z_r, u, v)$ of length at most $h$.
In this case, we terminate and return the same sequence $A$ as the output of $\findChain(h,(u,v))$.
\end{itemize}
\end{itemize}
\end{wrapper}

\subsection{Backward-Search}\label{sec:backward-search}

In this section, we provide a DFS-based procedure called $\backSearch$.
We have two types of $\backSearch$: (i) on an edge $(u,v) \in E - E(\calF)$,  or (ii) on a junction $x \in B$.
The goal of both types of $\backSearch$ are to find a partial sequence of nodes that can potentially be extended to an augmenting chain.

\subsubsection{Description of Backward-Search on an Edge $(y, x) \in E - E(\calF)$}
\label{new:back:search:edge}

The $\backSearch$ on $(y, x)$ tries to find $(w_0,z_1,\ldots,w_r,z_{r+1})$ such that $w_r = y$ and $z_{r+1} = x$.

\medskip
\noindent
\textbf{Description.}
First, we check if $x$ and $y$ are in different connected components of $\calF$ (aligned with Property \ref{new:simple:aug:5} of \Cref{new:def:simple:aug});
If yes, then we return $(y, x)$ and say that  $\backSearch$ on $(y,x)$ is successful.
From now on, we assume that $y$ and $x$ are in the same connected component of $\calF$. We next perform the following operations.

\begin{wrapper}
\begin{itemize}
\item {\bf While} $\min\limits_{x' \in B \cap P^\calF_{y,x}} \lvl_{x'} \leq \lvl_{(y, x)}$:
\begin{itemize}
    \item Let $x^\star \in B \cap P^\calF_{y,x}$ be a junction on the path $P^\calF_{y, x}$ with minimum level $\lvl_{x^\star}$ (These conditions are for achieving Property \ref{new:simple:aug:3} of \Cref{new:def:simple:aug}).
    \item We recursively call $\backSearch$ on  $x^\star$, as per \Cref{new:back:search:node}.
    \item If the call to $\backSearch$ on the junction $x^\star$ is successful, it must return a sequence $(w_0,z_1,w_1,\ldots,w_{r-1}, x^\star)$.
    In this case, we  terminate, and return $(w_0,z_1,w_1,\ldots, w_{r-1}, x^\star, y, x)$ as the output of $\backSearch$ on $(y, x)$. We also say that $\backSearch$ on $(y, x)$ is successful.
\end{itemize}
\item We increment $\lvl_{(x, y)}$ by one, and terminate. We say $\backSearch$ on $(y, x)$ fails.
\end{itemize}
\end{wrapper}

\subsubsection{Description of Backward-Search on a Junction $x \in B$}
\label{new:back:search:node}
The $\backSearch$ on a junction $x$ tries to find $(w_0,z_1,\ldots,w_r,z_{r+1})$ such that $z_{r+1} = x$.

\begin{definition}
\label{new:def:back:edge}
We say that an ordered pair $(y,x)$ is a {\bf backward edge} iff:
\begin{enumerate}\label{new:cond:neighbor}
    \item $(y, x) \in E - E(\calF)$ and $y$ is untouched (aligned with Property \ref{new:simple:aug:4} of \Cref{new:def:simple:aug}), and

    \item $\lvl_{(y,x)} \leq \lvl_x - 1$ (this condition helps ensure Property \ref{new:simple:aug:3} of \Cref{new:def:simple:aug}).
\end{enumerate}
\end{definition}

\noindent Now, the subroutine  $\backSearch$ on $x$ works as follows.

\begin{wrapper}
\begin{itemize}
\item {\bf While} there exists at least one backward edge $(y, x)$:
\begin{itemize}
\item Let $(y^\star, x)$ be a backward edge ending at $x$ with the minimum possible level $\lvl_{(y^\star, x)}$.
\item We recursively call $\backSearch$ on $(y^\star, x)$ (see \Cref{new:back:search:edge}).
\item If $\backSearch$ on $(y^\star, x)$ succeeds, then let $(w_0,z_1,\ldots,w_{r-1},z_r, y^\star, x)$ be the sequence returned by it.
    In this case, we say that $\backSearch$ on $x$ is also successful, we return the same sequence as our output, and terminate.
\end{itemize}
\item  We increment $\lvl_{x}$ by one, and terminate. We say that $\backSearch$ on $x$ fails.
\end{itemize}
\end{wrapper}

\iffalse
We refer to each such edge $(y, x)$ as a \textbf{backward edge}.
If there exists at least one backward edge, we find the backward edge $(y^\star, x)$ with the \textbf{minimum} possible level and make a recursive call to the $\backSearch$ on $(y^\star, x)$.
There are two cases for this call;
\begin{itemize}
    \item It successfully returns $(w_0,z_1,\ldots,w_{r-1},z_r, y^\star, x)$.
    In this case, we say that $\backSearch$ on $x$ is successful, and return th same sequence as the output.

    \item It fails.
    In this case, the level of ${(y^\star, x)}$ has increased upon failure.
    We then consider the next backward edge $(y^\star_\new, x)$ with minimum level.
\end{itemize}
As long as there is at least one backward edge, we repeatedly call $\backSearch$ on such an edge with minimum level.
If at least one recursive call for a backward edge is successful, we say that the $\backSearch$ on $x$ is successful as well.
If we end up in a situation where there is no more backward edge, we say that $\backSearch$ on $x$ has failed, increment the level of $x$, and terminate the $\backSearch$.
\fi

\subsection{Correctness of the Find-Chain Subroutine}\label{new:sec:correctness-find-chain}

\begin{claim}\label{claim:return-chain}
    If $\findChain(h,(u,v))$ returns a sequence, it must be an augmenting chain.
\end{claim}

\begin{proof}
Assume that $\findChain(h,(u,v))$ has returned $A = (w_0,z_1,\ldots, z_h, w_h, z_{h+1})$, where $w_h = u$ and $z_{h+1} = v$.
While describing $\backSearch$ and $\findChain$, we explicitly mentioned the correspondence between each property of \Cref{new:def:simple:aug} and each condition in the subroutines.
So, all of these properties are straightforward from our descriptions of $\backSearch$ and $\findChain$ except Property \ref{new:simple:aug:3}.
For the rest of the proof, we show Property \ref{new:simple:aug:3}.

According to the description of the subroutines, it is obvious that 
\begin{equation}\label{eq:zi-in-path}
    \text{For all $i \in [1, h]$: $z_i \in P^\calF_{w_i,z_{i+1}}$.}
\end{equation}
Denote by $T_{x \leftarrow y}^\calF$ (for this proof only) the sub-tree of $\calF$ containing $x$ that is derived by removing the unique edge incident to $y$ on the path $P^\calF_{x,y}$ from $\calF$.
Since we already have (\ref{eq:zi-in-path}), an equivalent way of writing Property \ref{new:simple:aug:3} is as follows:
\begin{equation}\label{eq:new-2d}
    \text{For all $i \in [1, h]$: all of the nodes $\{z_{i+1}, w_{i+1}, z_{i+2}, \ldots, w_h, z_{h + 1}\}$ are contained in $T^\calF_{z_{i+1}\leftarrow z_i}$.}
\end{equation}
To see the equivalence, (\ref{eq:zi-in-path}) concludes that $w_i$ must be outside of $T^\calF_{z_{i+1}\leftarrow z_i}$.
Combining this fact and (\ref{eq:new-2d}) automatically implies that $z_i$ separates $\{w_i\}$ and $\{z_{i+1}, w_{i+1}, z_{i+2}, \ldots, w_h, z_{h + 1}\}$.
It remains to show (\ref{eq:new-2d}).

\medskip
\noindent
\textbf{Proof of Property (\ref{eq:new-2d}).}
Assume that (\ref{eq:new-2d}) for $A$ is not satisfied for an index $i \in [1, h]$.
Consider the first node $x$ in the following order $ z_{i+1}, w_{i+1}, z_{i+2}, \ldots, w_h, z_{h + 1}$ that is not contained in $T^\calF_{z_{i+1}\leftarrow z_i}$.
Obviously, $z_{i+1} \in T^\calF_{z_{i+1}\leftarrow z_i}$.
We have the following two cases for $x$.

\medskip
\noindent
\textbf{Case I:} $x = w_j$ for some $j \in [i+1,h]$, which means $w_j \notin V(T^\calF_{z_{i+1} \leftarrow z_i})$.
According to the definition of $x$, we have $z_j \in V(T^\calF_{z_{i+1} \leftarrow z_i})$.
The two conditions $w_j \notin V(T^\calF_{z_{i+1} \leftarrow z_i})$ and  $z_j \in V(T^\calF_{z_{i+1} \leftarrow z_i})$ conclude $z_i \in V(P^\calF_{w_j, z_j})$.
According to (\ref{eq:zi-in-path}), we have that $z_j \in V(P^\calF_{w_j, z_{j+1}})$.
The two conditions $z_i \in V(P^\calF_{w_j, z_j})$ and $z_j \in V(P^\calF_{w_j, z_{j+1}})$ imply that $z_i$ is also contained on $V(P^\calF_{w_j, z_{j+1}})$.
Moreover, the descriptions of $\findChain$ and $\backSearch$ conclude $\lvl_{z_i} \leq \lvl_{(w_i, z_{i+1})} \leq \lvl_{z_{i+1}} - 1 \leq \lvl_{z_j}-1$.

This is in contradiction with $\backSearch$ since it always make a recursive call on a node with \textbf{minimum} possible level.
In other words, $z_i \in V(P^\calF_{w_j, z_{j+1}})$ has priority over $z_j$, and a recursive $\backSearch$ call must be made on $z_i$ instead of $z_j$.

\medskip
\noindent
\textbf{Case II:}
$x = z_j$ for some $j\in[i+2,h+1]$, which means $z_j \notin V(T^\calF_{z_{i+1} \leftarrow z_i})$.
According to the definition of $x$, we have that $w_{j-1} \in V(T^\calF_{z_{i+1} \leftarrow z_i})$.
The two conditions $z_j \notin V(T^\calF_{z_{i+1} \leftarrow z_i})$ and $w_{j-1} \in V(T^\calF_{z_{i+1} \leftarrow z_i})$ conclude $z_i \in V(P^\calF_{w_{j-1}, z_j})$.
Similarly, we get a contradiction since $\lvl_{z_i} < \lvl_{z_j}$.
\end{proof}

\section{Analysis}

We begin by noting that at the start of each round $t \in [1, H]$, $\left(\calF_t^\init, B_t^\init, \calC_t^\init \right)$ is a configuration.
This clearly holds at the start of round $t = 1$.
Subsequently, every time the algorithm updates $(\calF_t, B_t, \calC_t)$ by applying an effective edge, it remains a configuration (see \Cref{new:claim:apply-effective-remain-configuration}).
Hence, the execution of the algorithm is correct and $\left(\calF_{H}^\final, B_{H}^\final, \calC_{H}^\final\right)$ at the end of round $H$ is a configuration as well.
So, the output $\calF_{H}^\final$ is a feasible forest of $G$. 
Accordingly, the rest of this section is devoted towards proving \Cref{new:key:lemma}, which implies \Cref{lem:main}.

\begin{lemma}
    \label{new:key:lemma} The number of connected components of $\calF^\final_H$ is at most the number of connected components of $\calF^\init_1$ minus $\Omega(f/H)$. Further, it takes $\tilde{O}(mH)$ time to implement rounds $\{1, \ldots, H\}$.
\end{lemma}

For ease of exposition and to highlight our key conceptual insights, in this extended abstract we only provide the proof of \Cref{new:key:lemma} for the following values of the parameters.
\begin{equation}
\label{new:eq:param}
(k, H) := (2\Delta^\star + 1,  \lceil 10\log n \rceil).
\end{equation} 

In the full version, we show that \Cref{new:key:lemma} holds even if the values of $(k, H)$ are given by \Cref{eq:param}. 
We will show that $\calF^\final_H$ has $f/(100 H)$ many fewer components than $\calF_1^\init$.
Towards this end, recall that each time we apply an augmenting chain, the number of components of the underlying configuration decreases by one (see \Cref{new:lem:apply-chain}). So, we work under \Cref{new:assumption:small-improvement}. We will show that this assumption leads to a contradiction.

\begin{assumption}\label{new:assumption:small-improvement}
    Our algorithm applies only $\tau < f/(100 \cdot H)$ augmenting chains.
\end{assumption}

Throughout the algorithm, when an augmenting chain $A$ is applied, the degree of some junctions in $B$, or some dirty nodes in $D$ can be reduced.
The only scenario where this might happen is when a critical edge of $A$ is removed from the forest $\calF$ (see guarantee \ref{new:lem:apply-chain-P2} of \Cref{new:lem:apply-chain}).
Whenever a node $v$ becomes part of the set $B \cup D$ in our algorithm, it has $\deg_{\calF}(v) = k$  at that time-instant (see \Cref{new:sec:template}).
Since the total number of augmenting chains applied by our algorithm is at most $f/(100\cdot H)$ (see \Cref{new:assumption:small-improvement}), and each of these chains has at most $2H$ critical edges (follows from guarantee \ref{new:lem:apply-chain-P2} of \Cref{new:lem:apply-chain}, and the fact that each such augmenting chain is of length at most $H$),  we conclude that the sum $\sum_{x \in B \cup D} \deg_\calF(x)$ remains very close to $k \cdot |B \cup D|$. In particular, we have $\sum_{x \in B \cup D} \deg_\calF(x) \geq k \cdot |B \cup D| - f/c$ for some large constant $c > 1$.
For ease of exposition, in this extended abstract, we ignore this extra $f/c$ term, and work under \Cref{new:assumption:full-degree}.  The full version of the paper contains a comprehensive analysis without this extra assumption.

\begin{assumption}\label{new:assumption:full-degree}
    The configuration $(\calF, B, \calC)$ and the set of dirty nodes $D$ maintained by the algorithm always satisfy the following condition:  $\deg_\calF(x) = k = 2\Delta^\star+1$ for all $x \in B \cup D$.
\end{assumption}

\subsection{A Potential Function: Tail Independent Stable $t$-Chains}

In this section, we introduce a potential function that underpins the analysis of our algorithm. We start by defining the notion of a {\bf pseudo chain}. 

\begin{definition}\label{new:def:p:aug} Fix any configuration $(\calF, B, \calC)$ and any integer $h \geq 0$. A {\bf pseudo chain} of length $h$ w.r.t.~$(\calF, B, \calC)$ is a tuple of nodes $A = (w_0, z_1, w_1, z_2, w_2, \ldots, z_h, w_h)$ such that:
\begin{enumerate}
\item \label{p:aug:1} 
$w_h$ is untouched, $z_h \in B$ is a junction, and $w_h,z_h$ lie in the same connected component of $\calF$.
\item For every $i \in \{h-1, \ldots, 1\}$:
\begin{enumerate}
\item \label{p:aug:2} $(w_{i}, z_{i+1}) \in E - E(\calF)$ is a non-forest edge, and $w_{i}$ is untouched.
\item \label{p:aug:3} $w_i$ and $z_{i+1}$ are in different connected components of $\calF - z_i$.
\end{enumerate}
\item \label{p:aug:3} $w_0$ and $z_{1}$ are in two different connected components of $\calF$,
and $w_{0}$ is untouched.
\end{enumerate}
Consider any $i \in [1, h]$. Let $(z_i, y_i) \in P_{z_i, w_i}^{\calF}$ be the unique edge incident on $z_i$ in the path $P_{z_i, w_i}^{\calF}$.
Let $A(i)$ be the connected component of $\calF - \{(z_i, y_i)\}$ containing $y_i$, where $\calF - \{(z_i, y_i)\}$ is the forest  obtained by deleting the edge $(z_i, y_i)$ from $\calF$.
Further, let $A(0)$ denote the entire connected component of $\calF$ containing $w_0$.
For $i \in [0, h]$, we refer to $A(i)$ as the {\bf $i^{th}$ block of $A$}.  In particular, we refer to $A(h)$ as the {\bf tail} of $A$. 
Finally, for every $i \in [1, h]$, we refer to the nodes $z_i$ and $y_i$ respectively as the {\bf root} and {\bf pivot} of $A(i)$. In contrast, the root and pivot of $A(0)$ are undefined. 
\end{definition}

Note that a pseudo chain is very similar to an augmenting chain (see \Cref{new:def:simple:aug}), the only differences are: (i) The pseudo chain ends at $w_h$, as opposed to ending at $z_{h+1}$. (ii) We no longer enforce that the nodes $z_1,w_1,\ldots,w_h,z_{h+1}$ are in the same connected component of $\calF$, and the pseudo chain can hit multiple connected components of $\calF$.
(iii) It might very well be the case that $z_i = z_{j}$ for some $1 \leq i < j \leq h$ in a pseudo chain.

It will be of particular interest to us to keep track of those pseudo chains that are {\em stable}, in the sense that they continue to remain a valid pseudo chain at all future time-steps during the execution of our algorithm. Below, we formally introduce this concept.

\medskip
\begin{definition}
\label{new:def:stable:chain}
Consider any $t \in [1, H]$. Let $A = (w_0,z_1,w_1,\ldots, z_h, w_h)$ be a pseudo chain of length $h \in [0, t-1]$ w.r.t.~$\left(\calF^{\final}_t, B_t^\final, \calC_t^\final\right)$.
We say that $A$ is a  \textbf{stable $t$-chain} iff the following  conditions hold at all future time-steps (from the start of round $t+1$ until the end of round $H$):
\begin{enumerate}
    \item\label{P1:cond:stable-1}  For each $i \in [1, h]$, the node $z_i$ is never removed from the set $B$ of junctions, and the  edge $(z_i, y_i) \in P_{z_i, w_i}^{\calF}$, where $y_i$ is the $i^{th}$ pivot of $A$, is never removed from the forest $\calF$.

    \item\label{P1:cond:stable-2} For all $i \in [0, h]$, the $\left(\calF^\final_t,B^\final_{t}\right)$-regions in $A(i)$ are all untouched in $\left(\calF_t^\final, B_t^\final, \calC^\final_t\right)$, and they continue to remain untouched at all subsequent time-steps. Thus, the structure of the junctions and untouched regions inside a block of $A$  does not change after round $t$.
\end{enumerate}
\end{definition}

\begin{definition}
\label{new:def:tail}
Consider any round $t \in [1, H]$, and  a collection $\calA = \{A_1,\ldots, A_q\}$ of $q$ stable $t$-chains, for some integer $q \geq 1$.
We say that  $\calA$ is a \textbf{tail-independent} set of stable $t$-chains iff the tails of $A_1, \ldots, A_q$ are mutually node-disjoint.
\end{definition}

We are now ready to define the potential function we will use to analyze our algorithm. Our key technical insight is summarized in the lemma below.

\begin{lemma}\label{new:lem:potential}
For every round $t \in [1, H]$, there exists a tail-independent set of stable $t$-chains, denoted by $M(t)$. Under \Cref{new:assumption:small-improvement}, these sets satisfy the following properties.
\begin{enumerate}
\item \label{new:lem:potential-1} $|M(1)| \geq 4f/5$.
\item \label{new:lem:potential-2} $|M(t+1)| \geq (5/4) \cdot |M(t)|$ for each $t \in [1, H - 1]$.
\end{enumerate}
\end{lemma}

Assuming \Cref{new:lem:potential}, we have 
$ |M(H)| \geq (5/4)^{H-1} \cdot (4f/5) > n$, since $H = \lceil 10\log n \rceil$ as per \Cref{new:eq:param}.
This is clearly a contradiction since we can not have $|M(H)| > n = |V|$ stable $t$-chains whose tails are mutually node-disjoint.
In other words, \Cref{new:assumption:small-improvement} cannot be true, and this implies the first part of \Cref{new:key:lemma}; namely, that $\calF^\final_H$ contains $\Omega(f/H)$ fewer connected components than $\calF^\init_1$. Accordingly, we now focus on proving \Cref{new:lem:potential}.

\medskip
\noindent
\textbf{Analysis of $M(1)$.}
When we apply an augmenting chain, it affects only two connected components of $\calF$ (see \Cref{new:def:simple:aug}).
Hence, after applying $\tau$ augmenting chains, there are at least $f - 2\tau \geq f - f/(50 \cdot H) \geq 4f/5$ many connected components (see \Cref{new:assumption:small-improvement} and \Cref{new:eq:param}) of the input forest $\calF_1^\init$ that remain unaffected, i.e., have empty intersection with all the $\tau$ augmenting chains applied throughout all rounds.
Let $T$ be any such unaffected connected component of $\calF$.
It is easy to verify that all $\left(\calF_1^\final, B_1^\final\right)$-regions in $T$ remain untouched throughout the execution of the algorithm from the start of round $2$ forward, and each such component $T$ gives rise to exactly one stable $1$-chain of length $0$.
Clearly, these stable $1$-chains form a tail-independent set of stable $1$-chains (the tail of each of these stable $1$-chains is an entire connected component of $\calF^\final_1$).
This implies Property \ref{new:lem:potential-1} of \Cref{new:lem:potential}.

\medskip
It now remains to prove Property \ref{new:lem:potential-2} of \Cref{new:lem:potential}. We devote \Cref{new:sec:proof} towards this proof.

\subsection{Geometric Expansion of the Potential: Proving Property \ref{new:lem:potential-2} of \Cref{new:lem:potential}}
\label{new:sec:proof}

Throughout this section, we fix any round $t \in [1, H]$, and assume that $M(t) = \{A_1,A_2, \ldots, A_q\}$ is a tail-independent set of stable $t$-chains, for some integer $q \geq 1$.
Our goal is to construct  $M(t+1)$, a  tail-independent set of stable $(t+1)$-chains satisfying $|M(t+1)| \geq (5/4) \cdot q $.
Throughout this section, we will denote the tail of $A_i$ by $T_i$, for each $i \in [1, H]$.

\subsubsection{Tail-Expanding Edges}\label{sec:tail-exapnding}

Fix some optimal spanning tree $T^\star$ of $G$ with maximum degree $\Delta^\star$.
Consider any index $i \in [1, q]$.
We refer to an ordered pair $(u^\star,v^\star)$ as $T_i$-\textbf{expanding} iff $(u^\star, v^\star) \in E(T^\star)$,  $u^\star$ belongs to a $\left( \calF_t^\final, \calC_t^\final\right)$-region that is untouched w.r.t.~$\left(\calF_t^\final, B_t^\final, \calC_t^\final \right)$ and  contained  within $T_i$, and $v^\star$ lies outside of $T_i$.
We refer to  $(u^\star, v^\star)$ as \textbf{expanding edge} iff it is $T_i$-expanding for some $i \in [1, q]$.
Let $ B^\out_{t} := \left\{v^\star \mid (u^\star, v^\star) \text{ is an expanding edge}\right\} $ be the set of all endpoints $v^\star$ of expanding edges.

\begin{observation}\label{new:obs:expanding-types}
For each $T_i$-expanding ordered pair $(u^\star, v^\star)$, we have  $(u^\star, v^\star) \notin E\left(\calF^\final_t\right)$ unless $v^\star$ is the root of $T_i$ and $u^\star$ is the pivot of $T_i$. 
\end{observation}

\begin{claim}\label{new:claim:exist-expanding-edge}
    For each $i \in [1, q]$, there exists at least one $T_i$-expanding pair.
\end{claim}

\begin{proof}
    Let $B^{(i)} := V(T_i) \cap B_t^\final$, and $\calC^{(i)}$ be the collection of $(\calF^\final_t, B_t^\final)$-regions inside $T_i$ (that are all untouched since $A_i$ is a stable $t$-chain).
    According to Assumption \ref{new:assumption:full-degree}, with a simple counting argument (by considering $T_i$ as a rooted tree, removing the nodes in $B^{(i)}$ iteratively, and counting the number of regions), we can show that $ |\calC^{(i)}| \geq 1 + (2\Delta^\star - 1) \cdot |B^{(i)}|$.
    
    $T^\star$ must connect the regions in $\calC^{(i)}$ together and to the rest of the graph with a set $E^\star_i$ of at least $|\calC^{(i)}| - 1$ edges.
    As a result, if there does not exist any $T_i$-expanding pair, then all of the edges in $E^\star_i$ must be incident on $B^{(i)}$ (there is no edge connecting two regions of $\calC^{(i)}$ according to \Cref{new:obs:cond-end-round}).
    This concludes $|\calC^{(i)}| - 1 \leq \Delta^\star \cdot |B^{(i)}|$ since $\deg_{T^\star}(x) \leq \Delta^\star$ for all $x \in B^{(i)}$.
    The two inequalities on $|\calC^{(i)}|$ clearly contradict each other, which concludes the claim.
\end{proof}

\begin{claim}\label{new:claim:size-B-outtt}
    We have that $ |B_t^\out| \geq q / \Delta^\star $.
\end{claim}

\begin{proof}
    According to \Cref{new:claim:exist-expanding-edge}, we have at least $q$ different expanding edges.
    These edges are in the optimal spanning tree $T^\star$ and have an endpoint in $B_t^\out$. This implies that $q \leq \Delta^\star \cdot |B_t^\out|$ since the degree of each $x \in B_t^\out$ in $T^\star$ is bounded by $\Delta^\star$.
\end{proof}

\begin{claim}\label{new:claim:B-out-in-B-init}
    We have $B_{t}^\out \subseteq B_{t+1}^\init$.
\end{claim}

\begin{proof}
Consider any node $v^\star \in B_t^\out$.
By definition, there is an expanding edge of the form $(u^\star, v^\star)$.
If $(u^\star, v^\star) \in E(\calF_{t}^\final)$, then   $v^\star$ is the root of the tail of some stable $t$-chain (see \Cref{new:obs:expanding-types}).
In this case, we have $v^\star \in B_t^\final$ (see \Cref{new:def:p:aug}).

For the rest of the proof, assume that $(u^\star, v^\star) \notin E(\calF_{t}^\final)$.
If $v^\star \notin B^\final_t \cup D_t^\final$, then  $(u^\star,v^\star)$ is an effective edge w.r.t.~$\left(\calF^\final_t, B^\final_t, \calC^\final_t \right)$ with $v^\star \notin D_t^\final$, which contradicts \Cref{new:obs:cond-end-round}.
Hence, it must be the case that $v^\star \in B^\final_t \cup D_t^\final = B_{t+1}^\init$.

In other words, every node in $B_t^\out$ also belongs to $B_{t+1}^\init$. This concludes the proof.
\end{proof}

\subsubsection{Singly-Attached Regions}\label{sec:stable-t+1-chains}

Recall \Cref{new:claim:B-out-in-B-init}. Consider any $(\calF^\init_{t+1}, B^\out_{t})$-region $R$ (say).
We say that $R$ is a \textbf{singly-attached} $(\calF^\init_{t+1}, B^\out_{t})$-region iff it  is adjacent to \textbf{exactly} one node (say) $r$ in $B^\out_{t}$  in the forest $\calF^\init_{t+1}$.
If $R$ is singly-attached, then we refer to the node $r$ as the \textbf{root} of $R$.
Let $\calS_{t+1}$ denote the collection of all singly-attached $(\calF^\init_{t+1}, B^\out_{t})$-regions.

\begin{claim}\label{new:claim:size-singly-attached}
    We have that $|\calS_{t+1}| \geq (2\Delta^\star-1) \cdot |B^\out_{t}| $.
\end{claim}

\begin{proof}
    Let $X$ be a connected component of $\calF_{t+1}^\init$ such that $V(X) \cap B_t^\out \neq \emptyset$.
     Using Assumption \ref{new:assumption:full-degree}, with a simple counting argument (by considering $X$ as a rooted tree, removing the nodes in $B_t^\out \cap V(X)$ iteratively, and counting the number of resulting regions), we can show that the number of singly-attached $(\calF^\init_{t+1}, B^\out_{t})$-regions in $X$ is at least $1 + (2\Delta^\star- 1 ) \cdot |B_t^\out \cap V(X)|$.
    The claim follows by adding up these inequalities over all such connected components $X$ of $\calF_{t+1}^\init$.
\end{proof}

\subsubsection{Constructing Potential Stable $(t+1)$-Chains}\label{sec:defining-potential-chains}

Assume that $U \in \calS_{t+1}$ is an arbitrary singly-attached $(\calF^\init_{t+1}, B^\out_{t})$-region. Denote the root of $U$ by $x$. 
By definition, there exists a stable $t$-chain $A_i = (w_0,z_1,w_1,\ldots, z_{h}, w_h) \in M(t)$ of length $h \leq t-1$ whose tail is $T_i$, and there is a $T_i$-expanding pair $(u^\star,v^\star)$ with $ u^\star \in T_i$ and $v^\star = x \notin T_i$.

\medskip
We now define a pseudo chain $A_U$ (whose tail is $U$), which will be referred to as a {\bf potential stable $(t+1)$-chain}. Specifically, we consider two cases.

\begin{itemize}
\item 
If $(u^\star, v^\star) \notin E(\calF^\init_{t+1})$, then we consider any arbitrary leaf node $y$ of $\calF^\init_{t+1}$ inside $U$, and define $A_U := (w_0,z_1,w_1,\ldots, z_h, u^\star, x, y)$.
In this case, we refer to $A_U$ as an \textbf{extended} chain.
\item
Else if $(u^\star, v^\star) \in E(\calF^\init_{t+1})$, then $v^\star$ is the root of $T_i$ according to \Cref{new:obs:expanding-types}, and we define $A_U := A_i$.
In this case, we refer to $A_U$ as a \textbf{repetitive} chain.
\end{itemize}

\begin{observation}
\label{new:ob:pseudo}
    $A_U$ is a pseudo chain w.r.t.~$\left(\calF^\init_{t+1}, B^\init_{t+1}, \calC^\init_{t+1} \right)$, for all $U \in \calS_{t+1}$.
\end{observation}

Recall that the definition of a stable $(t+1)$-chain is w.r.t.~the configuration at the end of round $t+1$.
In contrast, we have defined $A_U$ w.r.t.~the configuration at the beginning of round $t+1$.
In subsequent sections, we will show that a large fraction of these potential stable $(t+1)$-chains actually happen to be stable $(t+1)$-chains.

\subsubsection{Definition of $M(t+1)$}\label{sec:def-M(t+1)}

Consider the potential stable $(t+1)$-chain $A_U = (w_0, z_1, w_1, \ldots, z_{h}, w_{h})$, for some $U \in \calS_{t+1}$.
We aim to see under what circumstances  $A_U$ might {\em not} be a stable $(t+1)$-chain.
Let $(u,v)$ be an effective edge that gets applied at some time-step after the beginning of round $t+1$ and before the end of round $H$.
We say that $(u,v)$ \textbf{affects} $A_U$ iff one of the following occurs while applying $(u, v)$.
\begin{enumerate}
    \item\label{destroy-P1} Either $z_{h}$ is removed from $B$, or the edge connecting $z_h$ to the pivot of the $h^{th}$ block of $A_U$ is removed from the forest $\calF$.

    \item\label{destroy-P2} Either an edge crossing $U$ (i.e., with one endpoint inside and the other endpoint outside of $U$) is inserted into the forest $\calF$, or a region inside $U$ becomes touched.
\end{enumerate}
Note that the above definition does not include applying a merging edge between two untouched regions in $U$, i.e., if $(u,v)$ is a merging edge between two untouched regions inside $U$, then we do \emph{not} say that $(u,v)$ affects $A_U$.
We say that $(u,v)$ \textbf{destroys} $A_U$ iff during the time-interval from the start of round $t+1$ until the end of round $H$, $(u,v)$ is the very \textit{first edge} that affects $A_U$.

\medskip
\noindent
\textbf{Definition of $M(t+1)$.}
Now, we define $M(t+1)$ to be the set of all potential stable $(t+1)$-chains $A_U$, for $U \in \calS_{t+1}$, that \emph{are not destroyed} during the time-interval from the start of round $t+1$ until the end of round $H$.
In the following, we show the correctness of our definition of $M(t+1)$.

\begin{claim}\label{new:claim:not-destroyed-is-stable}
    Each $A_U \in M(t+1)$ is a stable $(t+1)$-chain.
\end{claim}

\begin{proof}
(Sketch)
Consider any $U \in \calS_{t+1}$ such that $A_U \in M(t+1)$. By \Cref{new:ob:pseudo}, we have that $A_U = (w_0, z_1, w_1, \ldots, z_{h}, w_{h})$ is a pseudo chain w.r.t.~$\left(\calF^\init_{t+1}, B^\init_{t+1}, \calC^\init_{t+1} \right)$.
Moreover, $A_U$ is not destroyed, meaning that every effective edge hitting $A_U$ that have been applied after the start of round $t+1$ is a merging edge $(u,v)$ with both endpoints  $u, v \in U$.
At the end of round $t+1$, all such edges must have been applied (see \Cref{new:obs:cond-end-round}), meaning that the structure of the junction and regions in $U$ is settled.
Hence, $A_U$ is a stable $(t+1)$-chain, because of the following two reasons: (i) The structure of its tail does not change. (ii) The structure of the $i^\th$ block of $A_U$, for all $i \leq h-1$, also does not change, due to $(w_0,\ldots,z_{h-1},w_{h-1}) \subseteq A_U$ being a stable $t$-chain.
\end{proof}

\medskip
\Cref{new:claim:not-destroyed-is-stable} shows that our set $M(t+1)$ consists of a collection of stable $(t+1)$-chains.
Moreover, the tails of these stable $(t+1)$-chains are singly-attached $(\calF_{t+1}^\init, B_t^\out)$-regions (elements of $\calS_{t+1}$) that are mutually node-disjoint.
Hence, $M(t+1)$ is indeed a tail-independent set of stable $(t+1)$-chains.
It now remains to analyze the size of $M(t+1)$.

\subsubsection{Lower Bounding the Size of $M(t+1)$: Concluding the Proof}\label{sec:count-true-chains}

We start with the following important claim.

\begin{claim}\label{new:claim:destroying-is-improving}
    If $(u,v)$ destroys at least one sequence $A_U$ for some $U \in \calS_{t+1}$, then $(u,v)$ must be an improving edge.
\end{claim}

\begin{proof}(Sketch) 
First, all of the regions inside $U$ are untouched so far. Second, the edge $(u,v)$ must cross $U$, which means that the subroutine $\findChain(\cdot,\cdot)$, when called on $(u,v)$, has the option to call a $\backSearch$ on the root of $U$.
The nodes and edges corresponding to the sequence $A_U$ are possible choices for the recursive calls inside $\backSearch$, which leads to successful calls.
As a result, the $\backSearch$ can not fail, and return an augmenting chain (possibly a subsequence of $A_U$).
This implies that $(u,v)$ is an improving edge. 
\end{proof}

\begin{claim}\label{new:claim:number-of-destroys}
    Each improving edge $(u,v)$ can destroy at most $5H$ pseudo chains from $\{A_U\}_{U \in \calS_{t+1}}$.
\end{claim}

\begin{proof}(Sketch)
Let $A$ be the augmenting chain that gets applied when the algorithm decides to apply the improving edge $(u, v)$.
Note that the length of any augmenting chain applied by the algorithm is at most $H$.
So, at most $H+2$ many $(\calF, B)$-regions get touched while applying $A$, and there are at most $2H$ many critical edges of $A$  (see  \Cref{new:lem:apply-chain}).
These are the only problematic updates affecting $A_U = (w_0,\ldots,z_{h},w_h)$.
Moreover, $(w_0,\ldots,z_{h-1},w_{h-1}) \subseteq A_U$ is a stable $t$-chain, meaning that $A_U$ is destroyed iff its tail $U$ is affected by an update.
These tails $U \in \calS_{t+1}$ are mutually node-disjoint, meaning that applying $A$ can destroy $3H+2 \leq 5H$ many sequences $A_U$.
\end{proof}

\medskip
\noindent
\textbf{Bounding the Size of $M(t+1)$.}
Combining Claims \ref{new:claim:size-B-outtt} and \ref{new:claim:size-singly-attached}, we get  
$$|\calS_{t+1}| \geq \frac{2\Delta^\star-1}{\Delta^\star} \cdot q \geq (3/2)\cdot q, $$
From  \Cref{new:claim:destroying-is-improving} and \Cref{new:claim:number-of-destroys}, we conclude that  at most $\tau \cdot 5H \leq f/5$ many pseudo chains from $\{ A_U \}_{U \in \calS_{t+1}}$ get destroyed  (see \Cref{new:assumption:small-improvement}).
Hence, we have $|M(t+1)|\geq 
|\calS_{t+1}| - f/5 \geq (3/2)\cdot q - f/5 \geq (5/4)\cdot q  $.
The last inequality follows since $q = |M(t)| \geq |M(1)| \geq 4f/5$, as per Property  \ref{new:lem:potential-1} of \Cref{new:lem:potential}.
This concludes the proof of Property \ref{new:lem:apply-chain-P2} of \Cref{new:lem:apply-chain}.

\subsection{Running Time Analysis}
Almost all parts of the algorithm are quite easy to implement with elementary data structures.
There are a few operations where we need to deal with the unique path between two given nodes in a dynamic forest (after applying an augmenting chain the forest will change).
For those operations, we use the Top-Tree data structure \cite{TopTree}.
Equipped with this data structure, it is fairly straightforward to implement all parts of every round of the algorithm in $\tilde O(m)$ time, except the calls to the $\findChain$ subroutine.

\medskip
\noindent
\textbf{Time Spent on Find-Chain.}
To analyze the running time of the $\findChain$ subroutine and recursive calls to $\backSearch$, we charge the time of every failed $\backSearch$ to the increased value in the level of the object on which it is called.
Since the levels of the nodes and edges are bounded by $H$ throughout the entire algorithm, the total time spent on all calls to $\findChain$ is bounded by $\tilde{O}(m\cdot H)$ as well.
We refer the reader to the full version (\Cref{full-version}) for a comprehensive implementation and running time analysis.

\medskip
The above discussion, along with \Cref{new:lem:potential}, implies the proof of \Cref{new:key:lemma}.

\iffalse
\subsection{Proof of \Cref{lem:potential}: Part (ii)}\label{sec:proof-incresing-potential2}

Fix a value $t \in [1, H]$, and assume that $M(t) = \{A_1,A_2, \ldots, A_q\}$ for some $q \in \mathbb{N}$.
The goal is to construct the set $M(t+1)$ of tail-independent set of stable $(t+1)$-chains satisfying $|M(t+1)| \geq (5/4) \cdot q $.
We denote the tail of $A_i$ by $T_i$ for each $i \in [1, H]$.
\fi

\newpage

\newpage

\part{Full Version}\label{part:full}\label{full-version}

\begin{remark}
    In this part of the paper, we use different notations and definitions than \Cref{part:EA}.
    The goal of this part is to provide a comprehensive and precise description, analysis, and implementation of our algorithm.
    Although we use different notations than \Cref{part:EA}, this part of the paper is self-contained.
\end{remark}

\section{Introduction}

\subsection{Basic Notations}\label{full:sec:notation}
Throughout this part of the paper, we let $n := |V(G)|$ and $m := |E(G)|$ be the number of nodes and edges of the input graph. We use $\tilde{O}(\cdot)$ to hide $\mathrm{poly}\log n$ factors.

For each subgraph $K \subseteq G$, we denote by $N_K(x) \subseteq V$ and $N^+_K(x) \subseteq V$, the set of neighbors of $x$ in $K$ excluding and including $x$ itself, respectively. 
For each subset of nodes $U \subseteq V(G)$ and each subgraph $K \subseteq G$, we denote by $K[U]$ the induced subgraph of $K$ on the nodes $U$.

\medskip
\noindent
\textbf{Mostly Used Parameters.}
Throughout the paper, we use the following parameters;
\begin{align*}
    \Delta^\star &:= \min_{\substack{T \subseteq G \\ \text{is a spanning tree of $G$}}} \max_{u \in V(G)} \deg_T(u), \nonumber  \\
    k^\add &:= \Delta^\star + 1 
    \text{ and } 
    k^\mult := \lceil (1+\varepsilon)\cdot \Delta^\star\rceil + 1 \text{ for some arbitrary $0 < \varepsilon < 1$,}
    \\
    H^\add &:= \lceil 10n/f \rceil  \text{ and }
    H^\mult := \lceil \log_{1 + \varepsilon}(10n/f) \rceil + 1,
    \text{where $f$ is the number of components of $\calF$} \nonumber
\end{align*}
$\Delta^\star$ is the maximum degree of the optimum spanning tree in $G$.
We consider a general value of $k$ that represents the maximum degree of the spanning tree that we are searching.
$k$ can be either equal to $k^\add$, which provides our additive one approximation algorithm, or can be equal to $k^\mult$, which provides our multiplicative approximation algorithm.
We consider a generic parameter $H$ as well that takes the value of $H^\add$ or $H^\mult$ corresponding to the value of $k$, i.e., $H = H^\add$ if $k=k^\add$ and $H=H^\mult$ if $k=k^\mult$.
$H$ is an upper bound on the length of some object called \textit{augmenting chain} (see \Cref{full:def:aug-chain}) that we use throughout our algorithm.

From this point forwards, we provide our algorithm, and the analysis w.r.t.~the generic parameters $k$ and $H$.
Whenever a distinction between $ k = k^\add$ and $k = k^\mult$ as well as $H=H^\add$ and $H=H^\mult$ is needed, we explicitly note that.
So, throughout the paper, we assume
\begin{equation}\label{full:eq:param}
    \text{either } (k,H) = \left(k^\add, H^\add\right) \text{, or } (k,H) = \left(k^\mult, H^\mult\right).
\end{equation}

\medskip
\noindent
\textbf{Notations for Forest.}
Whenever we refer to a forest $\calF$, we consider $\calF$ as a forest of the original graph $G$ that contains all nodes of the graph, i.e., $V(\calF) = V(G)$.

For each forest $\calF$ of $G$ and each pair of nodes $u,v \in V(G)$ that are in the same connected component of $\calF$, we denote by $P^\calF_{u,v}$ the unique path in $\calF$ connecting $u$ and $v$.
We also let $T_{u \leftarrow v}^\calF$ be the connected sub-tree of $\calF$ containing $u$, obtained by removing the unique edge incident on $v$ in $P_{u,v}^\calF$.
We refer to a forest $\calF$ as a \textbf{feasible} forest, if for all $u \in V(G)$, we have $\deg_\calF(u) \leq k$.
For any arbitrary $X \subseteq V(G)$, we refer to each connected component of $\calF[V(G) - X]$ as a $(\calF, X)$-\textbf{region}.

\subsection{The Main Subroutines}\label{full:sec:two-main-subroutines}

We have the following two main lemmas.

\begin{lemma}\label{full:lem:main}
    There exists a deterministic algorithm that, given a forest $\calF$ with $f \geq 50$ components whose maximum degree is $\leq k$, returns another forest of maximum degree $\leq k$ with at least $f/(100 \cdot H)$ many fewer components in a total of $\tilde{O}(m \cdot H)$ time. The parameters $k$ and $H$ are defined in \Cref{full:eq:param}.
\end{lemma}

\begin{lemma}[Lemma 3.1 in the arXiv version of \cite{BFW26}]\label{full:lem:reduce-one}
    There exists a deterministic algorithm that, given a forest $\calF$ with $f \geq 2$ components whose maximum degree is $\leq k$, returns another forest of maximum degree $\leq k$ with one fewer component in a total of $\tilde{O}(m)$ time. The parameters $k$ and $H$ are defined in \Cref{full:eq:param}.
\end{lemma}

In \Cref{full:sec:concluding-main-results}, we explain how we achieve the main results of our paper (\Cref{thm:main1,thm:main2}) using the above lemmas.
The rest of the paper from \Cref{full:sec:building-blocks} forward is devoted to proving \Cref{full:lem:main}.

\subsection{Concluding Main Results}\label{full:sec:concluding-main-results}

As a corollary of \Cref{full:lem:main,full:lem:reduce-one}, we get the main results of our paper as follows.

\subsubsection{Proof of \Cref{thm:main1}}
Fix the value of $k := k^\mult = \lceil (1+\varepsilon)\Delta^\star \rceil + 1$ and $H := H^\mult = \lceil \log_{1+\varepsilon} (10n/f) \rceil + 1 \leq 2\varepsilon^{-1}\log(10n)$.
We start with an empty forest $\calF$.
While $\calF$ has $f \geq 50$ components, then we call \Cref{full:lem:main} and obtain another forest containing at most
$f - f/(100 \cdot H) \leq f(1-\varepsilon/(200 \log(10n)))$ components.
This call takes $\tilde{O}(m \cdot H) = \tilde{O}(m\cdot \varepsilon^{-1})$ time.
After  $O(\varepsilon^{-1}\log^2 n)$ iterations, the number of components of the forest eventually reduces to less than $50$.
Then, by repeatedly calling \Cref{full:lem:reduce-one} at most $50$ times, we get a spanning tree of $G$ of maximum degree at most $k = \lceil (1+\varepsilon)\Delta^\star \rceil + 1$.
The total running time of the algorithm is $\tilde{O}(m\cdot \varepsilon^{-2})$.

\subsubsection{Proof of \Cref{thm:main2}}

Fix the value of $k := k^\add = \Delta^\star + 1$ and $H := H^\add  = \lceil 10n/f \rceil$.
We start with an empty forest $\calF$.
While $\calF$ has $f \geq n^{2/3}$ components, then we call \Cref{full:lem:main} and obtain another forest containing at most
$f - f/(100 \cdot H) \le f(1-f/1000n) = f(1-\Omega(n^{-1/3}))$ components. That is, the number of components is reduced by a $(1-\Omega(n^{-1/3}))$ factor.
This call takes $\tilde{O}(m \cdot H) = \tilde{O}(m n^{1/3})$ time.
After  $O(n^{1/3}\log n)$ iterations, the number of components of the forest eventually reduces to less than $n^{2/3}$.
Then, by repeatedly calling \Cref{full:lem:reduce-one} at most $n^{2/3}$ times, we get a spanning tree of $G$ of maximum degree at most $k = \Delta^\star+ 1$.
The total running time of the algorithm is $\tilde{O}(m n^{2/3})$.

\section{Building Blocks}\label{full:sec:building-blocks}

\paragraph{Roadmap.}
In this section, we provide the main objects used in our algorithm.
We start by defining the notion of \textit{reducible} and \textit{non-reducible} nodes in \Cref{full:sec:reducible}.
This is similar to \cite{FR92} and \cite{BFW26}.
We continue by defining a \textit{configuration} in \Cref{full:sec:configutation} that at any point in time, captures the structure of the forest together with some other objects and their relations.
For simplicity in describing our algorithm, we omit the details of the definition of the most critical objects of our paper called \textit{augmenting chain}.
In \Cref{full:sec:aug-chain-black-box}, we briefly discuss what are the main properties of an augmenting chain that we need to describe our algorithm without dealing with the details of such a complex object.

\subsection{Reducible vs Non-Reducible Nodes}\label{full:sec:reducible}

Let $T$ be any arbitrary connected subgraph of a forest $\calF$, where the maximum degree of the nodes in $\calF$ is at most $k$.
Consider the following procedure, which is essentially the algorithm of \cite{FR92} running on $T$.
Initially, we mark all of the nodes $u \in V(T)$ satisfying $\deg_{\calF}(u) = k$ as  \textit{saturated nodes}, and mark all $(T,B)$-regions as \textit{slack components} where $B$ is the set of saturated nodes.
Then, as long as there exists a non-forest edge $(x,y) \in E(G[T]) - E(\calF)$ between two nodes $x$ and $y$ that are contained in two different slack components, we consider the path $P_{x,y}^T$ between $x$ and $y$ in $T$, and {\em merge} all of the slack components {\em hitting} or {\em adjacent to} $P_{x,y}^T$ together with all saturated nodes on $P_{x,y}^T$, to form a new larger slack component.
Throughout this process, each slack component remains a sub-tree of $T$, and different slack components remain mutually node-disjoint separated by saturated nodes.

\begin{definition}
    At the end of the above process, we refer to every node inside a slack component as a \textbf{reducible} node w.r.t.~$T$, and refer to every remaining saturated node as a \textbf{non-reducible} node w.r.t.~$T$.
\end{definition}

Morally, a node $u$ is reducible w.r.t.~$T$ if we can locally change $\calF$, by inserting and removing some edges inside $T$, to achieve another forest $\calF'$ that is feasible and moreover, $\deg_{\calF'}(u) \leq k-1$.
This key property is summarized in \Cref{full:lem:degree-reduction}.

\begin{lemma}[Degree Reduction: Lemma 2.1 in the arXiv version of \cite{BFW26}]\label{full:lem:degree-reduction}     
    There is a {\bf degree-reduction} subroutine which works as follows.
    Let $C$ be any connected subgraph of a feasible forest $\calF$.
    The subroutine takes $(C, u)$ as input, where $u \in C$ is reducible w.r.t.~$C$, and modifies $\calF$ by inserting/deleting some edges $e \in E$ whose {\em both} endpoints lie in $V(C)$. Let $\calF^+$ denote the state of $\calF$ just after the subroutine finishes execution.
    The subroutine runs in $\tilde{O}\left(\sum_{x \in V(C)} \deg_G(x)\right)$ time,\footnote{More specifically, we need $\tilde{O}\left(\sum_{x \in V(C)} \deg_G(x)\right)$ time to identify explicitly the set of edges $E(G[V(C)])$ whose both endpoints lie on $V(C)$, then, the rest of the subroutine takes only $\tilde{O}\left(\left| E\big( G[V(C)] \big) \right| \right)$ time.} and guarantees that:
    \begin{enumerate}
        \item $\deg_{\calF^+}(u) \leq k-1$.
        \item $\calF^+$ remains a feasible forest. 
    \end{enumerate}
\end{lemma}

\subsection{Configuration and Related Objects}\label{full:sec:configutation}

\begin{definition}
    We refer to $(\calF, B, \calC, D)$ as a configuration if the following holds;
    \begin{enumerate}
        \item $\calF$ is a feasible forest of $G$, i.e., for each $u \in V(G)$, $\deg_\calF(u) \leq k$.
        
        \item $B \subseteq V(G)$.
        We refer to nodes in $B$ as \textbf{junctions}.

        \item $\calC$ is a collection of subgraphs of $\calF$ such that every $C \in \calC$ satisfies the following properties;
        \subitem(a) $C$ is an $(\calF,B)$-regions, i.e., a connected component of the forest $\calF[V(G) - B]$.
        \subitem(b) Each $u \in C$ is reducible w.r.t.~the sub-tree $\calF[C]$.

        Elements of $\calC$ are called $(B, \calC)$-\textbf{reducible} regions.
        The rest of the regions are called $(B, \calC)$-\textbf{touched} region.
        We also refer to nodes inside a $(B, \calC)$-touched region as $(B, \calC)$-touched nodes.
        When it is obvious from the context what $\calF$, $B$, and $\calC$ are, we simply use the terms region, reducible, and touched.

        \item $D$ is a subset of $(B, \calC)$-touched nodes that contains all $(B, \calC)$-touched nodes $u$ satisfying $\deg_\calF(u) = k$.
        Note that $D$ can contain some other $(B, \calC)$-touched nodes with degree $\leq k-1$ as well.
        Later, we describe that the algorithm runs in rounds and updates the configuration by some operations.
        The set $D$ will satisfy the following; 1) it is empty at the start of each round, 2) it is incremental throughout each round, 3) when a node of degree $k$ becomes touched, it will be inserted to $D$, and 4) it will be reset to empty at the end of the round (beginning of next round).
    \end{enumerate}
\end{definition}

\begin{definition}[Deficit]
For each configuration $(\calF,B,\calC,D)$, the \textbf{deficit} of $(\calF,B,\calC,D)$ denoted by $\deficit(\calF,B,\calC,D)$ is defined as the amount by which $\sum_{x \in B \cup D} \deg_\calF(x)$ differs from $k \cdot |B \cup D|$.
More precisely, we define $\deficit(\calF, B,\calC,D) := \sum_{x \in B \cup D} (k - \deg_\calF(x))$.    
\end{definition}

\begin{definition}[Boundary]
For each forest $\calF$, set of junctions $B \subseteq V(G)$, and each $x \in V(G)$, we denote by $\bound_{(\calF,B)}(x)$ the collection of all $(\calF,B)$-regions $T$ satisfying $V(T) \cap N^+_\calF(x) \neq \emptyset$.
For any subgraph $H \subseteq G$, we define $$\bound_{(\calF,B)}(H) := \bigcup_{x \in V(H)} \bound_{(\calF,B)}(x) . $$
\end{definition}

\begin{definition}[Effective Edge]\label{full:def:effective-edge}
    Assume $\conf{}{}$ is a configuration.
    We refer to $(u, v) \in E(G)$ as an \textbf{effective} edge w.r.t.~$\conf{}{}$ if the following holds;
    \begin{enumerate}
        \item $u$ belongs to a reducible region in $\calC$ denoted by $C_u$.
        
        \item $v \notin C_u$.

        \item $(u,v) \in E(G) - E(\calF)$.

        \item $v$ satisfies one of the following;
        \subitem(a) it is contained in a reducible region in $\calC$,
        \subitem(b) it is a touched node satisfying $v \notin D$.
        Accordingly, $\deg_\calF(v) \leq k-1$.
    \end{enumerate}
\end{definition}

\subsection{Augmenting Chains as Black Box}\label{full:sec:aug-chain-black-box}

For simplicity in the description of the overview of our algorithm in \Cref{full:sec:alg-overview}, we omit the details of the most critical and complex object of our algorithm, called the augmenting chain.
This is the main object that we use in order to reduce the number of components of the forest.
The comprehensive definition of this object and how we use it is deferred to \Cref{full:sec:aug-chain}.
Here, we consider this object as a black box.
The only property of augmenting chains that we need to describe our algorithm is that each augmenting chain $A$ is defined w.r.t.~a configuration (i.e., if $A$ is an augmenting chain w.r.t.~$\conf{}{}$, it is not necessarily an augmenting chain w.r.t.~another configuration $\conf{\prime}{}$), and it has a length $h \in [1, H]$.

\medskip
\noindent
\textbf{Finding Augmenting Chains.}
Since we are considering augmenting chains as black box objects in this section, we also need to consider the subroutine that searches for augmenting chains as a black box.
This subroutine is called $\findChain(h,(u,v))$ that takes two parameters, an integer $h \in [1, H]$ and an effective edge $(u, v)$.
Then, it either successfully finds an augmenting chain of length at most $h$, or it fails.
We defer the description of this subroutine to \Cref{full:sec:alg-find-chain}.
The only guarantee of this subroutine that we need here is the following:
\begin{equation}\label{full:eq:guarantee-find-chain}
    \text{If $\findChain(h,(u,v))$ fails, then $u$ and $v$ lie in the same connected component of $\calF$.}
\end{equation}

\medskip
\noindent
\textbf{Applying Augmenting Chains.}
Augmenting chains are used to reduce the number of components of the graph in our algorithm.
If $A$ is an augmenting chain w.r.t.~$\conf{}{}$, there is a procedure called \textit{applying $A$ to $\conf{}{}$}, which returns another configuration $\conf{\prime}{}$ such that $\calF'$ has exactly one fewer component that $\calF$.
We consider this subroutine as a black box here, and defer its comprehensive description to \Cref{full:sec:apply-chain}.

\section{Overview of Our Algorithm for \Cref{full:lem:main}}\label{full:sec:alg-overview}

\textbf{Roadmap.}
In this section, we describe the overview of our algorithm by considering augmenting chains as a black box.
In \Cref{full:sec:apply-effective-edge}, we start by providing a procedure that, given an effective edge $(u,v)$ w.r.t.~some configuration, updates the configuration towards an improving direction.
This improving direction can be either applying an augmenting chain (that reduces the number of components of the forest), or updating the configuration such that more potential augmenting chains can be found for the rest of the algorithm.
Next, in \Cref{full:sec:alg-description}, we describe our algorithm.

\subsection{Applying Effective Edges}\label{full:sec:apply-effective-edge}

We classify every effective edge $(u,v)$ w.r.t.~$(\calF, B, \calC, D)$ as one of the following types;
If the call to the subroutine $\findChain(h, (u,v))$ successfully finds an augmenting chain of length at most $h$ w.r.t.~$\conf{}{}$, we call $(u,v)$ an \textbf{improving} edge.
If $\findChain(h, (u,v))$ fails, we conclude that $u$ and $v$ are in the same connected component of $\calF$ (according to Guarantee (\ref{full:eq:guarantee-find-chain})).
Now, consider the unique path $P^\calF_{u,v}$ connecting $u$ and $v$ in $\calF$.
If $\bound_{(\calF,B)}(P^\calF_{u,v})$ does \emph{not} contain any $(B, \calC)$-\emph{touched} regions, then we call $(u,v)$ a \textbf{merging} edge.
If $\bound_{(\calF,B)}(P^\calF_{u,v})$ contains at least one $(B, \calC)$-touched region, we call $(u,v)$ a \textbf{touching} edge.

In the following, we provide a procedure that is called applying an effective edge $(u,v)$ to the configuration $\conf{}{}$.
This process returns another configuration $\conf{\prime}{}$ that is defined in the following.

\subsubsection{Applying an Improving Edge}\label{full:sec:apply-improving}
If $(u,v)$ is an improving edge, then $\findChain$ has successfully returned an augmenting chain $A$ w.r.t.~$\conf{}{}$.
We define $(\calF',B',\calC',D')$ to be the result of applying $A$ to $\conf{}{}$.
We provide the procedure of applying an augmenting chain to a configuration in \Cref{full:sec:apply-chain}, after providing the details of an augmenting chain.

\subsubsection{Applying a Merging Edge}\label{full:sec:apply-merging}
If $(u,v)$ is a merging edge, then we define $(\calF',B',\calC',D')$ as follows;
1) $\calF' = \calF$, 2) $B' := B \setminus V(P^{\calF}_{u,v})$, 3) $\calC' := (\calC \setminus \bound_{(\calF, B)}(P^{\calF}_{u,v})) \cup \{\calF[W]\}$, where 
$$W := \left( V(P^{\calF}_{u,v}) \right) \cup \left( \cup_{C \in \bound_{(\calF, B)}(P^{\calF}_{u,v})} V(C) \right),$$
and 4) $D' := D$.

\subsubsection{Applying a Touching Edge}\label{full:sec:apply-touching}
If $(u,v)$ is a merging edge, then we define $(\calF',B',\calC',D')$ as follows;
1) $\calF' = \calF$, 2) $B' := B \setminus V(P^{\calF}_{u,v})$, 3) $\calC' := \calC \setminus \bound_{(\calF, B)}(P^{\calF}_{u,v}) $, and
4) $D' := D \cup \left(B \cap V(P^\calF_{u,v}) \right) \cup  \{x \in W \mid \deg_{\calF'}(x) = k\}$, where 
$$ W :=  \cup_{C \in \bound_{(\calF, B)}(P^{\calF}_{u,v})} V(C) . $$

\subsubsection{Correctness}

\begin{claim}\label{full:claim:apply-effective-remain-configuration}
    For any effective edge $(u,v)$, the result $(\calF', B', \calC', D')$ of applying $(u, v)$ to $(\calF, B, \calC, D)$ is a valid configuration.
\end{claim}

\begin{proof}
    If $(u,v)$ is an improving edge, the claim follows from the guarantee of applying an augmenting chain to a configuration.
    We provide the details in \Cref{full:lem:apply-min-chain} after describing the details of augmenting chains.
    
    If $(u,v)$ is a merging edge, the only non-trivial property of $(\calF',B',\calC', D')$ being a configuration, is why every $C \in \calC'$ is a $(\calF',B')$-region, and why every $u \in C$ is reducible w.r.t.~$C$.
    Note that the forest $\calF' = \calF$ does not change.
    According to the definition of $B'$, none of the junctions of $B$ lying on $P^{\calF}_{u,v}$ are part of $B'$, and we have combined all $(\calF,B)$-regions hitting or adjacent to $\bound_{(\calF,B)}(P^\calF_{u,v})$ into a single component $\calF[W]$ and add it to $\calC'$.
    Hence, every $C \in \calC'$ is indeed a connected component of $\calF'[V(G) - B']$ (i.e., a $(\calF', B')$-region).
    
    Finally, if $C \in \calC' \cap \calC$, every node $u \in C$ is reducible w.r.t.~$C$ since the forest $\calF'=\calF$ did not change and $(\calF,B,\calC,D)$ is a configuration.
    If $C \in \calC' - \calC$, we must have $C = \calF[W]$.
    In this case, every $u \in C$ is reducible w.r.t.~$C$ because of the process of defining reducible nodes.
    Note that there is a non-forest edge $(u,v)$ and we are combining all reducible regions hitting or adjacent to $P^{\calF}_{u,v}$ (this is directly aligned with the procedure that we described in \Cref{full:sec:reducible}).

    If $(u,v)$ is a touching edge, the claim follows from a similar argument as when $(u,v)$ is merging.
\end{proof}

\subsection{Description of the Algorithm for \Cref{full:lem:main}}\label{full:sec:alg-description}

Assume the given forest is $\calF$ whose maximum degree is at most $k$.
The algorithm runs in $H$ rounds.
Initially, we define $\calF_{0}^\final = \calF$, $B_{0}^\final = \{x \in V \mid \deg_{\calF}(x) = k\}$, $\calC_{0}^\final = \emptyset$, and $D_0^\final = \emptyset$ (this notation is for consistency of the algorithm description).
Now, we explain the procedure of our algorithm in round $t \in [1, H]$.

\subsubsection{Initialization of Round $t$}
At the beginning of round $t$, we let $\calF^\init_t = \calF^\final_{t-1}$, $B^\init_t := B_{t-1}^\final \cup D^\final_{t-1}$, and $D_t^\init := \emptyset$.
Then, we define $\calC^\init_t$ to be the collection of \emph{all} $(\calF^\init_t, B^\init_t)$-regions.

\subsubsection{Main Body of Round $t$}
In the main body of the $t^\th$ round, the goal is to apply as many augmenting chains of length at most $t$ as possible.
We start with the configuration $(\calF_t, B_t, \calC_t, D_t) := (\calF^\init_t, B^\init_t, \calC^\init_t, D^\init_t)$.
Then, as long as there exists an effective edge $(u,v)$ w.r.t.~the current configuration $(\calF_t, B_t, \calC_t, D_t)$, we first call the subroutine $\findChain(t, (u,v))$ to decide the type of $(u,v)$.
Then, we apply this edge to the configuration immediately after termination of $\findChain(t, (u,v))$, and update the configuration.
As a result, either an augmenting chain of length at most $t$ (returned by $\findChain(t, (u,v))$) is applied to the configuration, or we apply a merging or touching edge (if $\findChain(t, (u,v))$ fails).

We will continue this process until there is no more effective edge w.r.t.~the current configuration $(\calF_t, B_t, \calC_t, D_t)$.
In that case, we finish round $t$, denote by $ (\calF^\final_t, B^\final_t, \calC^\final_t, D^\final_t) := (\calF_t, B_t, \calC_t, D_t)$ the final configuration of round $t$, and start round $(t+1)$.

\begin{observation}\label{full:obs:cond-end-round}
    At the end of any round $t \in [1, H]$, there does not exists any effective edge w.r.t.~the final configuration $(\calF^\final_t, B^\final_t, \calC^\final_t, D^\final_t)$.
    Accordingly, if an edge $(u,v)$ satisfies the first three properties of an effective edge (\Cref{full:def:effective-edge}) w.r.t.~$(\calF^\final_t, B^\final_t, \calC^\final_t, D^\final_t)$, then we must have that $v \in B_t^\final \cup D_t^\final$.
\end{observation}

\subsubsection{The Output}
After running the algorithm for $H$ rounds as explained above, we return the final forest $\calF_{H}^\final$ as the final output of the algorithm for \Cref{full:lem:main}.
\Cref{full:code:rounds} summarizes the procedure of our algorithm.

\begin{algorithm}[h]
  \DontPrintSemicolon
  $\conf{}{0} := \left( \calF, \{x\in V \mid \deg_\calF(x) = k\}, \emptyset, \emptyset \right)$.
  
  \For{$t = 1$ to $H$}
  {
  $\calF_t := \calF_{t-1}$,
  $B_t := B_{t-1} \cup D_{t-1}$,
  $\calC_t :=$ the collection of all $(\calF_t, B_t)$-regions, and
  $D_t := \emptyset$.
  
  \While{there exist an effective edge $(u,v)$ w.r.t.~$\conf{}{t}$}{
    Call $\findChain(t,(u,v))$.
    
    \If{$(u,v)$ is improving, i.e., $\findChain(t,(u,v))$ returns $A = (w_0,z_1,\ldots, z_r, u,v)$}
    {
    Apply $A$ to $\conf{}{t}$ according to \Cref{full:sec:apply-improving} and update the configuration.
    }
    \ElseIf{$(u,v)$ is merging}{
        Apply $(u,v)$ to $\conf{}{t}$ according to \Cref{full:sec:apply-merging} and update the configuration.
    }\ElseIf{$(u,v)$ is touching}{
        Apply $(u,v)$ to $\conf{}{t}$ according to \Cref{full:sec:apply-touching} and update the configuration.
    }
  }
  }
    \Return $\calF_H$.
    \caption{Our Algorithm for \Cref{full:lem:main}, Input: forest $\calF$}\label{full:code:rounds}
\end{algorithm}

\section{Augmenting Chains}\label{full:sec:aug-chain}

\textbf{Roadmap.}
In this section, we provide the details of the most intricate and complex object in our paper, called an augmenting chain.
We start with the definition of an augmenting chain in \Cref{full:sec:def-aug-chain}.
Augmenting chains are used to reduce the number of components of the forest.
This is done by a procedure called applying the augmenting chain.
We provide this procedure and its analysis in \Cref{full:sec:apply-chain}.

\subsection{Augmenting Chain and Its Application}\label{full:sec:def-aug-chain}

\begin{definition}[Augmenting Chain]\label{full:def:aug-chain}
    Let $(\calF, B, \calC, D)$ be a configuration and $h \geq 0$ be an integer.
    We call a sequence of nodes $A = \left(w_0, z_1, w_1, z_2, w_2, \ldots, z_h, w_h, z_{h+1}\right)$ an {\em augmenting chain} of {\em length} $h + 1$ w.r.t.~$(\calF, B, \calC, D)$ iff the following properties hold; 
    \begin{enumerate}

        \item\label{full:aug-chain-prop-1}
        $z_1, w_1, \ldots, z_h, w_h$, and $z_{h+1}$ are in the same connected component of $\calF$, while $w_0$ is not in this component.
    
        \item\label{full:aug-chain-prop-2} For all $i \in [1, h]$, $z_{i} \in B$, i.e., each $z_i$ is a junction.
        
        \item\label{full:aug-chain-prop-3}
        For every $i \in [0, h]$, $w_i$ belongs to a reducible region in $\calC$.
       
        For every $i \in [1, h]$, we denote the reducible region containing $w_i$ by $C_{w_i}$.
        We refer to $T_{w_i \leftarrow z_i}^\calF$ as the {\em $i^{\th}$ block of $A$} that contains $C_{w_i}$ as well.
        We also refer to the entire component of $\calF$ containing $w_0$ as the $0^\th$ block of $A$.

        \item\label{full:aug-chain-prop-4}
        $z_{h+1}$ satisfies one of the following; 
        \subitem (a) it is contained in a reducible region in $\calC$ denoted by $C_{z_{h+1}}$,
        \subitem (b) it is a $(B, \calC)$-touched node and $z_{h+1} \notin D$.
        Accordingly, $\deg_\calF(z_{h+1}) \leq k-1$. 
        
        In case (a), we denote the reducible region containing $z_{h+1}$ by $C_{z_{h+1}}$, and in case (b), we define $C_{z_{h+1}} := \{z_{h+1}\}$.

        \item\label{full:aug-chain-prop-5} For each $i \in [1, h]$, we have that $z_i \in P^\calF_{w_i, z_{i+1}}$, in other words, $z_{i+1}$ is outside of the $i^\th$ block.

        \item\label{full:aug-chain-prop-6} 
        For every $i \in [1, h]$, all of the nodes $ z_{i+1}, w_{i+1}, z_{i+2}, \ldots, w_h,z_{h+1}$
        are contained in $T^\calF_{z_{i+1}\leftarrow z_i}$.

        \item\label{full:aug-chain-prop-7} For all $i \in [0, h]$, we have $(w_{i}, z_{i+1}) \in E(G)- E(\calF)$, i.e., $(w_{i}, z_{i+1})$ is a non-forest edge.
    \end{enumerate}
\end{definition}

\subsection{Applying an Augmenting Chain to a Configuration}\label{full:sec:apply-chain}

This section is devoted to proving the following.

\begin{lemma}\label{full:lem:apply-chain}
    Assume $(\calF, B, \calC, D)$ is a configuration and $A = \left(w_0, z_1, w_1, z_2, w_2, \ldots, z_{h+1}\right)$ is an augmenting chain w.r.t.~$(\calF, B, \calC, D)$.
    There exists a procedure called applying $A$ to $(\calF, B, \calC, D)$ that returns another configuration $(\calF', B', \calC', D')$ with the following guarantees;
    \begin{enumerate}
        \item\label{full:lem:apply-chain-P1} $\calF'$ has exactly one fewer component than $\calF$.
        
        \item\label{full:lem:apply-chain-P2} There are at most $2h+1$ edges $(u,v) \in E(\calF') \oplus E(\calF)$ that $\{u,v\} \cap B \neq \emptyset$.

        (we refer to such edges as \textbf{critical} edges of $A$.)

        \item\label{full:lem:apply-chain-P3} $B' = B$, $D' \supseteq D$, and $ 0 \leq \deficit(\calF', B',\calC',D') - \deficit(\calF, B,\calC,D) \leq 2h$.

        \item\label{full:lem:apply-chain-P4} $\calC' \subseteq \calC$ and $\calC-\calC' \subseteq \{C_{z_{h+1}}, C_{w_0}, C_{w_1}, \cdots, C_{w_h}\}$, which concludes $ |\calC - \calC'| \leq h+2 $. 
    \end{enumerate}
\end{lemma}

\subsubsection{Description of Applying an Augmenting Chain}
Assume $(\calF, B, \calC, D)$ is a configuration and $A = (w_0, z_1, w_1, \ldots, w_h, z_{h + 1})$ is an augmenting chain w.r.t.~$(\calF, B, \calC, D)$.
We now define a new configuration $(\calF', B', \calC', D')$ in the following which is called \textit{the result of applying $A$ to $(\calF, B, \calC, D)$}.
\begin{itemize}
    \item Defining $\calF'$: first, we apply \Cref{full:lem:degree-reduction} on each $(C_{w_i}, w_i)$.
    If $C_{z_{h+1}}$ is a reducible region, we also apply \Cref{full:lem:degree-reduction} on $(C_{z_{h+1}}, z_{h+1})$.
    Note that it is possible to do all these degree reductions simultaneously (see \Cref{full:claim:property-minimal-chain}).
    Let $\calF^+$ be the new forest after these calls to \Cref{full:lem:degree-reduction}.
    For each $i \in [1, h]$ assume $(z_i,y_i)$ is the unique edge that connects $z_i$ to $T^\calF_{w_i \leftarrow z_i}$, i.e., the unique edge in $P^\calF_{z_i,w_i}$ incident on $z_i$.
    We define $\calF' := \calF^+ - \{(z_i,y_i)\}_{i=1}^h + \{(w_i,z_{i+1})\}_{i=0}^h$.

    \item Defining $B'$: we let $B' := B$.

    \item Defining $\calC'$: we let $\calC' = \calC - \{C_{w_0}, C_{w_1}, \ldots, C_{w_h}, C_{z_{h+1}}\}$.
    In words, we remove from the collection of reducible regions those regions that have already been affected by the degree-reduction process in \Cref{full:lem:degree-reduction}.
    
    \item Defining $D'$: we let $D' := D \cup \{x \in W \mid \deg_{\calF'}(x) = k \}$, where $W = V(C_{h+1}) \cup \cup_{i=0}^h V(C_{w_i})$.
    In words, if we have a touched node whose degree is $k$ and is not already in $D$, we add it to $D'$.
\end{itemize}

The intuition for defining $\calC'$ is that if a reducible region is used to reduce the degree of $w_i$ for $i \in [0,h]$ or the degree of $z_{h+1}$, then we do not know whether the rest of the nodes inside this region remain reducible w.r.t.~this region.
Hence, we remove them from the collection of reducible regions.

In the definition of $D'$, we just make sure that if a node is both touched and has degree $k$ at the same time, it will become part of $D$, and remains there.
But it is possible that in the future, the algorithm applies more augmenting chains, and the degree of the nodes in $D$ reduces below $k$.

\subsubsection{A Property of Augmenting Chains}

\begin{claim}\label{full:claim:property-minimal-chain}
    Assume $(\calF, B, \calC, D)$ is a configuration and $A = (w_0, z_1, w_1, \ldots, w_h, z_{h + 1})$ is an augmenting chain w.r.t.~$(\calF, B, \calC, D)$.
    Then, we have the following properties;
    \begin{enumerate}
        \item\label{full:property-minimal-chain-1} The nodes $w_0, z_1, w_1, \ldots, w_h, z_{h + 1}$ are all distinct.

        \item\label{full:property-minimal-chain-2} $C_{z_{h+1}}$ and all the reducible regions $C_{w_i} \in \calC$ for $i \in [0, h]$ are mutually node-disjoint.
    \end{enumerate}
\end{claim}

\begin{proof}
    Both of these properties can be easily proved by the definition of an augmenting chain.
    Assume property 1 does not hold.
    Since $z_i$s are junctions, $z_{h+1}$ and $w_i$s are not junctions, we have one of the following;
    \begin{itemize}
        \item $z_i = z_j$ for $1 \leq i < j \leq h$. In this case, $z_j$ is not contained in $T^\calF_{z_{i+1} \leftarrow z_i}$ since $T^\calF_{z_{i+1} \leftarrow z_i}$ does not contain its root $z_i = z_j$. 

        \item $w_i = w_j$ for some $0 \leq i \leq h$. In this case, according to the definition of an augmenting chain, we have $z_i \in P^\calF_{w_i, z_{i+1}}$, which concludes $w_j = w_i \notin T^\calF_{z_{i+1} \leftarrow z_i}$.

        \item $w_i = z_{h+1}$ for some $0 \leq i \leq h$.
        Similar to the argument of the previous case, we conclude that $z_{h+1}$ is not in $T^\calF_{z_{i+1} \leftarrow z_i}$.
    \end{itemize}
    All of the above cases contradict the definition of $A$ being an augmenting chain. 
    This concludes Property \ref{full:property-minimal-chain-1}.
    Property \ref{full:property-minimal-chain-2} is similar to the proof of why $z_{h+1}$ and $w_i$s are distinct.
    The only important note is that the regions are separated by junctions.
    Hence, a region is either completely contained in $T^\calF_{z_{i+1} \leftarrow z_i}$, or has empty intersection with it.
\end{proof}

\subsubsection{Correctness of Applying an Augmenting Chain}

\begin{claim}\label{full:lem:apply-min-chain}
    Assume $A = (w_0, z_1, w_1, \ldots, w_h, z_{h + 1})$ is an augmenting chain w.r.t.~the configuration $(\calF, B, \calC, D)$.
    The result $(\calF', B', \calC', D')$ of applying $A$ to $(\calF, B, \calC, D)$ is a configuration.
\end{claim}

\begin{proof}
    According to the definition of the result $(\calF', B', \calC', D')$ of applying $A$ to $(\calF, B, \calC, D)$, the only non-trivial property regarding $(\calF', B', \calC', D')$ being a configuration is why $\calF'$ is a feasible forest.
    First of all, as explained in the definition of $\calF'$, we do the degree reduction process (\Cref{full:lem:degree-reduction}) on all the nodes $z_{h+1},w_0,w_1,\ldots,w_h$ simultaneously in order to get a new forest $\calF^+$ such that $\deg_{\calF^+}(x) \leq k-1$ for each $x \in \{z_{h+1},w_0,w_1,\ldots,w_h\}$.
    Since the degree reduction process only affects the $(\calF,B)$-regions internally, we still have the following properties for $\calF^+$;
    \begin{itemize}
        \item $z_i \in P^{\calF^+}_{w_i, z_{i+1}}$ for all $i \in [1,h]$.
        \item $T^{\calF^+}_{z_{i+1} \leftarrow z_{i}} = T^{\calF}_{z_{i+1} \leftarrow z_{i}}$ for all $i \in [1,h]$.
        \item $(w_i,z_{i+1}) \in E(G) - E(\calF^+)$ for $i \in [0,h]$.
    \end{itemize}
    Now, we iteratively insert the edge $(w_i,z_{i+1})$ into and remove the edge $(z_i,y_i)$ from the forest for $i$ going down from $h$ to $1$, and show that the forest remains a feasible forest such that the degree of $z_i$ is at most $k-1$ (before inserting $(w_{i-1},z_i)$ to the forest).
    
    This is obvious for $i=h \geq 1$ since $z_h \in P^{\calF^+}_{w_h, z_{h+1}}$, and the degree of both $w_h$ and $z_{h+1}$ is at most $k-1$ in $\calF^+$.
    Moreover, after removing $(z_h,y_h)$, the degree of $z_h$ in the forest becomes at most $k-1$.

    Now, assume that we have inserted $\{(w_i,z_{i+1})\}_{i=j}^h$ into and removed $\{(z_i,y_{i})\}_{i=j}^h$ from $\calF^+$ and it remains a feasible forest such that the degree of $z_j$ is at most $k-1$.
    According to the definition of an augmenting chain, all of the edges that we inserted into and removed from $\calF^+$ are completely contained in $T^{\calF^+}_{z_{j} \leftarrow z_{j-1}}$.
    As a result, the path connecting $w_{j-1}$ and $z_j$ in the forest, still contains the node $z_{j-1}$ and the edge $(z_{j-1}, y_{j-1})$.
    As a result if we insert the edge $(w_{j-1},z_{j})$ and remove the edge $(z_{j-1}, y_{j-1})$, we will still have a forest.
    Moreover, the degree of $w_{j-1}$ and $z_{j}$ were at most $k-1$ before inserting $(w_{j-1},z_{j})$, hence the forest remains a feasible forest.

    Finally, we insert the edge $(w_0, z_1)$ into the forest that does not create any cycle.
    The reason is that $w_0$ and $z_1$ are in different components of the forest initially (according to the definition of an augmenting chain), and they remain in different connected components of the forest since all changes on the forest so far were completely inside the component containing $z_1,w_1,\ldots,z_h,w_h,z_{h+1}$.
    Hence, the final forest $\calF'$ is a feasible forest.
\end{proof}

\subsubsection{Concluding the Proof of \Cref{full:lem:apply-chain}}

All of the guarantees of \Cref{full:lem:apply-chain} follow from the process of applying the augmenting chain $A$ to $(\calF, B, \calC, D)$.
Properties \ref{full:lem:apply-chain-P1} and \ref{full:lem:apply-chain-P4} of \Cref{full:lem:apply-chain} are trivial.
Property \ref{full:lem:apply-chain-P2} follows since the edges that are updated inside the regions due to the degree reduction (see \Cref{full:lem:degree-reduction}) for $z_{h+1}, w_0,w_1,\ldots, w_h$ are not incident on $B$ (are not critical edges of $A$).
Hence, the only edges that can be critical are $(w_i,z_{i+1})$ for $i \in [0, h]$ and $(z_i,y_i)$ for $i \in [1, h]$.
To show Property \ref{full:lem:apply-chain-P3}, it is straightforward to see that the deficit of the configuration does not decrease while applying the augmenting chain, i.e., $\deficit(\calF',B',\calC',D') - \deficit(\calF,B,\calC,D) \geq 0$.
Moreover, the degree of a node in $x \in B$ can reduce because of only one reason; $x=y_i$ for some $i \in [1, h]$ and after removing the edge $(z_i,y_i)$, the degree of $y_i$ is reduced by one unit.
Together with the fact that every node $x \in D' - D$ has degree $k$, we conclude Property \ref{full:lem:apply-chain-P3}.

\section{Find-Chain Subroutine}\label{full:sec:alg-find-chain}

\textbf{Roadmap.}
In this section, we describe the main subroutine of our algorithm called $\findChain$.
The goal of this subroutine is to search for augmenting chains.
It will recursively call another critical subroutine of our algorithm called $\backSearch$.
We describe $\findChain$ in \Cref{full:sec:find-chain-describe} that uses another subroutine called $\backSearch$.
Next, we describe the $\backSearch$ subroutine in \Cref{full:sec:backward-search}.
Finally, in \Cref{full:sec:correctness-find-chain}, we show the correctness of this subroutine.

\medskip
\noindent
\textbf{Level of Nodes and Edges.}
Throughout our algorithm, for each node $x \in V(G)$, we maintain an integer $\lvl_x \in [1, H+1]$.
Similarly, for each edge $(u,v) \in E(G)$, we maintain an integer $\lvl_{(u,v)} \in [0, H+1]$.
We call $\lvl_x$ the \textbf{level} of $x$, and $\lvl_{(x,y)}$ the level of $(x,y)$.
At the beginning of round $0$ of the algorithm, all nodes have level $1$, and all edges have level $0$.
These levels remain increasing throughout the entire algorithm until the end of round $H$.

\subsection{Description of Find-Chain}\label{full:sec:find-chain-describe}

The subroutine $\findChain$ takes two parameters; and integer $h \in [1, H]$, and an effective edge $(u,v)$.
The goal is to find an augmenting chain $A = (w_0, z_1, \ldots, z_r, u, v)$, whose last edge is $(u,v)$.
The condition that the input $(u,v)$ to $\findChain$ is an effective edge is directly aligned with Properties \ref{full:aug-chain-prop-3} and \ref{full:aug-chain-prop-4} of an augmenting chain.

\medskip
\noindent
\textbf{Description.}
As long as $\lvl_{(u,v)} \leq h-1$, we call the subroutine $\backSearch$ on the edge $(u,v)$.
Each call to a $\backSearch$ on an edge either succeeds (which returns a sequence $A = (w_0, z_1, \ldots, z_r, u, v)$) or fails (as you will see in the description of $\backSearch$, if it fails, it would increment the level of $(u,v)$).
If at least one of these calls successfully returns a sequence $(w_0, z_1, \ldots, z_r, u, v)$, we say that $\findChain$ is successful and return the same sequence as the output of $\findChain$.
If all of these calls are failed, and the level of $(u,v)$ is increased to $h$, we say that $\findChain$ has failed.
\Cref{full:code:find-chain} summarizes the $\findChain$ subroutine.

\begin{algorithm}[h]
  \DontPrintSemicolon
  \While{$\lvl_{(u,v)} \leq h-1$}{
    Call $\backSearch$ on $(u,v)$.
    
    \If{$\backSearch$ on $(u,v)$ returns $(w_0,z_1,w_1,\ldots, z_r, u,v)$}{\Return $(w_0,z_1,w_1,\ldots, z_r, u,v)$.}
    }

    \Return null
    \caption{$\findChain(h,(u,v))$}\label{full:code:find-chain}
\end{algorithm}

\subsection{Backward-Search}\label{full:sec:backward-search}

In this section, we provide a DFS-based procedure called $\backSearch$.
We have two types of $\backSearch$, which can be either called on a non-forest edge $(u,v) \notin \calF$, or on a junction $x \in B$.
The main purpose of a $\backSearch$ is to find a partial sequence of nodes that can be completed to an augmenting chain.
If such a partial sequence exists, it is possible that $\backSearch$ finds it, but it is also very well possible that $\backSearch$ misses it.
The main reason is our limited running time of $\tilde{O}(m \cdot H)$ for the final algorithm in \Cref{full:lem:main}.
But in \Cref{full:sec:analysis}, we show that this special type of search is delicate enough to get a large number of augmenting chains.

\subsubsection{Description of Backward-Search on an Edge}

The main purpose of the subroutine $\backSearch$ on a non-forest edge $(y, x)$ is to find a partial sequence $(w_0,z_1,\ldots,w_r,z_{r+1})$ such that $y = w_r$ and $z_{r+1} = x$ and this sequence is potentially a subsequence of an augmenting chain.

\medskip
\noindent
\textbf{Description.}
First, we check if $x$ and $y$ are in different connected components of $\calF$ (this is aligned with Property \ref{full:aug-chain-prop-1} of an augmenting chain);
If $y$ is outside of the connected component of $\calF$ containing $x$, then we say that the $\backSearch$ on $(y,x)$ is successful, and return $(y, x)$ as the output.

Now, assume that $y$ and $x$ are in the same connected component of the forest $\calF$.
We consider the path $P^\calF_{y,x}$ and iterate over all junctions $x' \in B \cap P^\calF_{y,x}$ on this path (this is directly aligned with Properties \ref{full:aug-chain-prop-2} and \ref{full:aug-chain-prop-5} of an augmenting chain) whose levels are at most $\lvl_{(x,y)}$ (this condition is for having fast running time).
We consider such $x^\star$ with \textbf{minimum} level (among all such $x'$) and do a recursive call to $\backSearch$ on $x^\star$
(the condition of $x^\star$ having minimum level is for achieving Property \ref{full:aug-chain-prop-6} of an augmenting chain).
There are two cases for this recursive call;
\begin{itemize}
    \item\label{full:case:successful} The recursive call successfully finds $(w_0,z_1,w_1,\ldots,w_{r-1}, x^\star)$.
    In this case, we attach $(y, x)$ to the end of this sequence and return it as the output of the $\backSearch$ on $x$, i.e., we say that the $\backSearch$ on $x$ is successful and output $(w_0,z_1,w_1,\ldots, w_{r-1}, x^\star, y, x)$.

    \item The recursive call on $x^\star$ fails.
    In this case, we continue by considering the next junction $x^\star_\new \in B \cap P^\calF_{y,x}$ with minimum level.
\end{itemize}
We continue the above process until at least one of the recursive calls is successful, or we end up in a situation that there does not exist any $x' \in B \cap P^{\calF}_{y,x}$ satisfying $\lvl_{x'} \leq \lvl_{(y,x)}$.
In the second case, we say that the $\backSearch$ has failed.
Then, we increment the level of $(y,x)$ and terminate the $\backSearch$.
\Cref{full:code:backward-search-edge} summarizes this procedure.

\begin{algorithm}[h]
\caption{$\backSearch$ on an edge $(y,x)$\label{full:code:backward-search-edge}}
  \DontPrintSemicolon
    \If{$y$ and $x$ are in different connected components of $\calF$\label{full:line:dif-comps}}
    {\Return $(y,x)$.}
    \While{there exists $x' \in B \cap P^{\calF}_{y,x}$ such that $\lvl_{x'} \leq \lvl_{(y,x)}$\label{full:line:condition-back-on-node}}{
    
    Let $x^\star = \arg\min\left\{\lvl_{x'} \mid x' \in B \cap P^\calF_{y,x} \text{ and } \lvl_{x'} \leq \lvl_{(y,x)} \right\}$

        Call $\backSearch$ on $x^\star$.

        \If{$\backSearch$ on $x^\star$ returns $(w_0,z_1,w_1,\ldots, w_{r-1},x^\star)$}{\Return $(w_0,z_1,w_1,\ldots, w_{r-1},x^\star,y,x)$.}
    }
    $\lvl_{(y,x)} \gets \lvl_{(y,x)} + 1$.
  
  \Return null.
\end{algorithm}

\subsubsection{Description of Backward-Search on a Node}
The main purpose of the subroutine $\backSearch$ on a junction $x$ is to find a partial sequence $(w_0,z_1,\ldots,w_r,z_{r+1})$ such that $z_{r+1} = x$ and this sequence is potentially a subsequence of an augmenting chain.

\medskip
\noindent
\textbf{Description.}
We iterate over all edges $(y, x)$ incident on $x$ satisfying the following conditions;
\begin{itemize}\label{full:cond:neighbor}
    \item $(y, x) \in E(G) - E(\calF)$ (this is directly aligned with Property \ref{full:aug-chain-prop-7} of an augmenting chain).
    
    \item $y$ is contained in a reducible region $C_y \in \calC$ (this is directly aligned with Property \ref{full:aug-chain-prop-3} of an augmenting chain).

    \item $\lvl_{(y,x)} \leq \lvl_x - 1$ (this condition is for two purposes: first, having a fast running time, and second, achieving Property \ref{full:aug-chain-prop-6} of an augmenting chain).
\end{itemize}
We refer to such an edge as a \textbf{backward} edge.
If there exists at least one backward edge, we find the backward edge $(y^\star, x)$ with the \textbf{minimum} possible level and make a recursive call to the $\backSearch(y^\star, x)$.
There are two cases for this call;
\begin{itemize}
    \item $\backSearch(y^\star, x)$ successfully returns a sequence $(w_0,z_1,\ldots,w_{r-1},z_r, y^\star, x)$.
    In this case, we say that $\backSearch$ on $x$ is successful, and return this sequence as the output.

    \item $\backSearch(y^\star, x)$ fails.
    In this case, we consider the next backward edge $(y^\star_\new, x)$.
\end{itemize}
As long as there exists at least one backward edge, we repeatedly call $\backSearch$ on such an edge with minimum level.
Every time that $\backSearch$ on an edge $(y,x)$ fails, it will increase the level of $(y,x)$.
As a result, this process will finish, and the algorithm eventually halts. 
If at least one of these recursive calls is successful, we say that this $\backSearch$ on $x$ is successful as well.
If we end up in a situation where there is no more backward edge, we say that $\backSearch$ on $x$ has failed.
Then, we increment the level of $x$, and terminate the $\backSearch$.
\Cref{full:code:backward-search-node} summarizes this procedure.

\begin{algorithm}[h]
  \DontPrintSemicolon
  \While{there exists $y \in N_{G}(x)$ such that [$(y, x) \notin E(\calF)$], [$y$ in contained in a reducible region in $\calC$], and [$\lvl_{(y, x)} \leq \lvl_x - 1$]\label{full:line:conditions-back-on-edge}}{
    Let $(y^\star, x) = \arg \min \left\{ \lvl_{(y,x)} \mid (y,x) \notin E(\calF) \text{ and } y \text{ is in a reducible region in } \calC \text{ and } \lvl_{(y,x)} \leq \lvl_x - 1\right\}$.
    
    Call $\backSearch$ on $(y^\star, x)$.

    \If{$\backSearch(y^\star, x)$ returns $(w_0,z_1,w_1,\ldots, w_{r-1},z_r,y^\star,x)$}{\Return $(w_0,z_1,w_1,\ldots, w_{r-1},z_r,y^\star,x)$.}
    }
  $\lvl_{x} \gets \lvl_{x} + 1$.
  
  \Return null.
  \caption{$\backSearch$ on a node $x$}\label{full:code:backward-search-node}
\end{algorithm}

\subsection{Correctness of Find-Chain}\label{full:sec:correctness-find-chain}

In this section, we show the following.

\begin{claim}\label{full:claim:return-chain}
    If $\findChain(h,(u,v))$ returns a sequence $(w_0,z_1,\ldots, z_h, u, v)$, then it must be an augmenting chain w.r.t.~$(\calF, B, \calC, D)$. 
\end{claim}
\begin{proof}
Properties \ref{full:aug-chain-prop-1}, \ref{full:aug-chain-prop-2}, \ref{full:aug-chain-prop-3}, \ref{full:aug-chain-prop-4}, \ref{full:aug-chain-prop-5}, and \ref{full:aug-chain-prop-7} of an augmenting chain for the sequence $(w_0,z_1,\ldots, z_h, u, v)$ easily follow from the definition of an effective edge $(u,v)$ and the conditions in the descriptions of $\backSearch$ and $\findChain$.
While describing $\backSearch$ and $\findChain$, we explicitly mentioned the correspondence between each property of an augmenting chain and each condition in the procedure of $\backSearch$ and $\findChain$.
The only non-trivial property is Property \ref{full:aug-chain-prop-6}, which we prove below.

\medskip
\noindent
\textbf{Proof of Property \ref{full:aug-chain-prop-6}.}
Let $A =(w_0,z_1,\ldots,z_h,w_h,z_{h+1}) = (w_0,z_1,\ldots,z_h,u,v)$ and assume that Property \ref{full:aug-chain-prop-6} is not satisfied for an index $i \in [1, h]$.
Consider the first node $x$ in the following order $ z_{i+1}, w_{i+1}, z_{i+2}, \ldots, w_h, z_{h + 1}$ that is not contained in $T^\calF_{z_{i+1}\leftarrow z_i}$.
Obviously, $z_{i+1} \in T^\calF_{z_{i+1}\leftarrow z_i}$.
We have the following two cases for $x$.

\medskip
\noindent
\textbf{Case I:} $x = w_j$ for some $j \in [i+1,h]$, which means $w_j \notin T^\calF_{z_{i+1} \leftarrow z_i}$.
According to the definition of $x$, we have $z_j \in T^\calF_{z_{i+1} \leftarrow z_i}$.
The two conditions $w_j \notin T^\calF_{z_{i+1} \leftarrow z_i}$ and  $z_j \in T^\calF_{z_{i+1} \leftarrow z_i}$ conclude $z_i \in P^\calF_{w_j, z_j}$.
According to Property \ref{full:aug-chain-prop-5} for $A$ (that we showed above), we have that $z_j \in P^\calF_{w_j, z_{j+1}}$.
The two conditions $z_i \in P^\calF_{w_j, z_j}$ and $z_j \in P^\calF_{w_j, z_{j+1}}$ imply that $z_i$ is also contained in the path $P^\calF_{w_j, z_{j+1}}$.
Moreover, the descriptions of $\findChain$ and $\backSearch$ conclude $\lvl_{z_i} \leq \lvl_{(w_i, z_{i+1})} \leq \lvl_{z_{i+1}} - 1 \leq \lvl_{z_j}-1$.

We now have a contradiction with the procedure of $\findChain$ and $\backSearch$ since they always call a recursive $\backSearch$ on a node with \textbf{minimum} possible level.
In other words, the node $z_i$ in the path $P^\calF_{w_j, z_{j+1}}$ have priority than the node $z_j$.
So, the subroutine must have called a $\backSearch$ directly on $z_i$ instead of $z_j$.

\medskip
\noindent
\textbf{Case II:}
$x = z_j$ for some $j\in[i+2,h+1]$, which means $z_j \notin T^\calF_{z_{i+1} \leftarrow z_i}$.
According to the definition of $x$, we have that $w_{j-1} \in T^\calF_{z_{i+1} \leftarrow z_i}$.
The two conditions $z_j \notin T^\calF_{z_{i+1} \leftarrow z_i}$ and $w_{j-1} \in T^\calF_{z_{i+1} \leftarrow z_i}$ conclude that $z_i \in P^\calF_{w_{j-1}, z_j}$.
Similar to Case I, we get a contradiction since $\lvl_{z_i} < \lvl_{z_j}$ and the subroutine must have called a $\backSearch$ directly on $z_i$ instead of $z_j$.

\medskip
In both cases, we have a contradiction.
This shows that Property \ref{full:aug-chain-prop-6} must be correct for $A$.
\end{proof}

\section{Analysis}\label{full:sec:analysis}

\textbf{Roadmap.}
In \Cref{full:sec:alg-correctness}, we provide the correctness of our algorithm, i.e., why the final output is a forest of maximum degree at most $k$.
Next, in \Cref{full:sec:reduced-num-comps}, we show that the output forest has the specified number of components less than the input forest.
The proof relies on the crucial \Cref{full:lem:potential}.
The proof of \Cref{full:lem:potential} is provided in \Cref{full:sec:proof-incresing-potential1,full:sec:proof-incresing-potential2}.
\Cref{full:sec:deferred-proofs} contains the omitted proofs of the claims of \Cref{full:sec:proof-incresing-potential1,full:sec:proof-incresing-potential2}.

\subsection{Correctness}\label{full:sec:alg-correctness}

\begin{observation}\label{full:obs:valid-init-conf}
    For any $t \in [1, H]$, if $\conf{\final}{t-1}$ is a valid configuration, then $\conf{\init}{t}$ is also a valid configuration.
\end{observation}

\medskip
\noindent
\textbf{Correctness of the Algorithm.}
Our algorithm starts with $\conf{\init}{t}$ at the beginning of round $t$, which is a valid configuration (according to \Cref{full:obs:valid-init-conf}) at the beginning of each round.
Every time the algorithm updates $(\calF_t, B_t, \calC_t, D_t)$ by applying an effective edge, it remains a valid configuration (see \Cref{full:claim:apply-effective-remain-configuration}).
Hence, the procedure of the algorithm is valid and the final $(\calF_{H}^\final, B_{H}^\final, \calC_{H}^\final, D_{H}^\final)$ at the end of round $H$ is a valid configuration as well.
As a result, the final forest $\calF_{H}^\final$ returned by the algorithm is a feasible forest, i.e., the maximum degree of $\calF_{H}^\final$ is at most $k$.

\subsection{Number of Components of the Final Forest}\label{full:sec:reduced-num-comps}

\begin{lemma}\label{full:lem:main-reduce-comps}
At the end of round $H$ of the algorithm, the number of components of the forest $\calF^\final_{H}$ is at most $f - f/(100 \cdot H)$ where $f$ is the number of components of the input forest $\calF$.
\end{lemma}

The rest of this section is devoted to proving this lemma.
Throughout this section, we consider the following assumption, and show that this assumption leads to a contradiction.
This shows \Cref{full:lem:main-reduce-comps}.

\begin{assumption}\label{full:assumption:small-improvement}
    We assume that all of the augmenting chains that the algorithm has applied are $A^{(1)}, A^{(2)}, \ldots, A^{(\tau)}$, where $\tau < f/(100 \cdot H)$.
\end{assumption}

\subsubsection{Stable Chains}\label{full:sec:def-stable-chains}

\begin{definition}[Pseudo-Chain]\label{full:def:pseudo-chain}
    Let $(\calF, B, \calC, D)$ be a configuration and $h \geq 0$ be an integer.
    We call a sequence $A = \left(w_0, z_1, w_1, z_2, w_2, \ldots, z_h, w_h \right)$ a {\em pseudo-chain} of {\em length} $h$ w.r.t.~$(\calF, B, \calC, D)$ iff the following properties hold; 
    \begin{enumerate}

        \item\label{full:pseudo-chain-P1} For all $i \in [1, h]$, $z_{i} \in B$.

        \item\label{full:pseudo-chain-P2}
        For every $i \in [0, h]$, $w_i$ and $z_i$ are in the same connected component of $z_i$, and $w_i$ belongs to a reducible region in $\calC$.

        We refer to $T_{w_i \leftarrow z_i}^\calF$ as the {\em $i^{\th}$ block of $A$}.
        We also refer to the entire component of $\calF$ containing $w_0$ as the $0^\th$ block of $A$.

        \item\label{full:pseudo-chain-P3} For each $i \in [1, h-1]$, we have that $z_{i+1} \not\in T^\calF_{w_i \leftarrow z_{i}}$, i.e., $z_{i+1}$ is outside of the $i^\th$ block of $A$.
        
        \item\label{full:pseudo-chain-P4} For all $i \in [1, h]$, we have $(w_{i-1}, z_i) \in E(G)- E(\calF)$, i.e., $(w_{i-1}, z_i)$ is a non-forest edge.
    \end{enumerate}
\end{definition}

Note that in the above definition, the nodes $w_0,z_1,w_1,\ldots,z_h,w_h$ are not necessarily distinct.
Moreover, the blocks of a pseudo-chain may overlap, and it does not have a neat structure like an augmenting chain.
We use this definition only for analysis purposes in this section, and the algorithm works with augmenting chains as described in \Cref{full:sec:alg-overview,full:sec:aug-chain,full:sec:alg-find-chain}.

\begin{definition}[Stable $t$-Chain]\label{full:def:stable-chain}
Consider the configuration $\conf{\final}{t}$ at the end of round $t$.
Let $A = (w_0,z_1,w_1,\ldots,w_{h-1}, z_h, w_h)$ be a pseudo-chain of length $h \in [0, t-1]$ w.r.t.~$\conf{\final}{t}$.
We refer to $A$ as a \textbf{stable $t$-chain} if the following hold;
\begin{enumerate}
    \item\label{full:P1:cond:stable-1} Throughout the rest of the algorithm (i.e., from the end of round $t$ until the end of round $H$), for each $i \in [1, h]$, $z_i$ is never removed from the set $B$ of junctions, and the edge $(z_i, y_i)$ connecting $z_i$ to the $i^\th$ block of $A$ is never removed from the forest $\calF$.

    \item\label{full:P1:cond:stable-2} For each $i \in [0, h]$, the $(\calF^\final_t,B^\final_{t})$-regions inside the $i^\th$ block are all reducible regions in $\calC^\final_t$ and they never be removed from $\calC$ throughout the rest of the algorithm.
    In other words, the structure of the junctions and reducible regions inside the $i^\th$ block of $A$ will never change for the rest of the algorithm. 
\end{enumerate}
\end{definition}

\begin{observation}\label{full:obs:remain-pseudo-chain}
    If $A$ is a stable $t$-chain, then $A$ remains a pseudo-chain w.r.t.~any configuration from the beginning of round $t+1$ until the end of round $H$.
\end{observation}

\subsubsection{Tail-Independent Set of Stable Chains}
Let $t \in [1, H]$ be arbitrary.
Assume $A = (w_0,z_1,w_1, \ldots, w_{h-1},z_{h},w_h)$ is a stable $t$-chain where $h \in [0, t-1]$.
We refer to the $h^\th$ block of $A$ as the \textbf{tail} of $A$.
Consider a set $\{A_1,A_2,\ldots, A_q\}$ of $q$ many stable $t$-chains.
We refer to this set as a \textbf{tail-independent} set of stable $t$-chain, if the tails of these stable $t$-chains are mutually node-disjoint.
Now, we construct a tail-independent set of stable $t$-chains denoted by $M(t)$ for each value of $t$ recursively.
Our construction will satisfy the main property, as in the following lemma.

\begin{lemma}\label{full:lem:potential}
We have the following;
\begin{enumerate}
    \item[(i)]
    $|M(1)| \geq f/5$, and
    \item[(ii)] for each $t \in [1, H - 1]$, we have
    $ |M(t+1)| \geq \frac{k-2}{\Delta^\star - 1} \cdot |M(t)| + f/5 $. 
\end{enumerate}
\end{lemma}

\noindent
Before we explain the construction of $M(t)$ and prove \Cref{full:lem:potential}, we show how this lemma implies \Cref{full:lem:main-reduce-comps}.
Then, we construct $M(t)$ and prove \Cref{full:lem:potential} in \Cref{full:sec:proof-incresing-potential1,full:sec:proof-incresing-potential2}.

\subsubsection{Concluding \Cref{full:lem:main-reduce-comps}}
We have two cases as follows;
\begin{itemize}
    \item $k = k^\add = \Delta^\star+1$ and $H = H^\add = \lceil 10n/f \rceil$ accordingly.
    In this case, we have 
    $|M(t+1)| \geq |M(t)| + f/5  $, and by induction, we conclude $|M(H)| \geq H \cdot (f/5) > n$.
    This is clearly a contradiction since the elements of $M(H)$ are mutually node-disjoint, but the number of nodes in the graph is $n$.
    
    \item $k = k^\mult = \lceil (1+\varepsilon) \cdot \Delta^\star \rceil + 1$ and $H = H^\mult = \lceil \log_{1+\varepsilon}(10n/f)\rceil + 1$ accordingly.
    Let $\alpha = (k-2)/(\Delta^\star-1)$ that satisfies $1 + \varepsilon \leq \alpha < 4$ (since $\Delta^\star \geq 2$ and $\varepsilon < 1$).
    We have 
    $|M(t+1)| \geq \alpha \cdot |M(t)| + f/5 $, and by induction, we conclude that $$ |M(H)| \geq \frac{\alpha^H - 1}{\alpha - 1} \cdot (f/5) \geq \alpha^{H-1} \cdot (f/5) \geq (1+\varepsilon)^{\lceil \log_{1+\varepsilon}(10n/f) \rceil} \cdot (f/5) > n. $$
    Similar to the previous case, we have a contradiction.
\end{itemize}
In both cases, we end up with a contradiction.
This contradiction arises from \Cref{full:assumption:small-improvement}, which is the core of the proof of \Cref{full:lem:potential}.
Hence, \Cref{full:assumption:small-improvement} is false and we have \Cref{full:lem:main-reduce-comps}.

\subsection{Proof of \Cref{full:lem:potential}: Part (i)}\label{full:sec:proof-incresing-potential1}

\textbf{Untouched Components of $\calF$.}
Each augmenting chain affects only two connected components of $\calF$ (see Property \ref{full:aug-chain-prop-1} of \Cref{full:def:aug-chain}).
Moreover, the algorithm applied $\tau$ many augmenting chains in total.
We conclude that there are at least $f - 2\tau$ many components of the input forest $\calF_1^\init$ that remain `untouched', i.e., they have empty intersection with all of the augmenting chains applied by the algorithm.

\medskip
\noindent
\textbf{Stable $1$-Chains Corresponding to Untouched Components.}
Assume $T$ is one of these $\geq f-2\tau$ connected components of $\calF$.
Initially, all of the $(\calF_1^\init, B_1^\init)$-regions in $T$ are reducible, and during round $1$ of the algorithm, no touched node will emerge in $T$ since no augmenting chain hitting $T$ is applied.
As a result, during round $1$, only merging edges will be applied internally in $T$.
At the end of round $1$, there is no more edge between the reducible $(\calF_1^\final, B_1^\final)$-regions since the algorithm stops round $1$ only when there is no more effective edge (see \Cref{full:obs:cond-end-round}).

It is obvious that throughout the rest of the algorithm (from the end of round $1$ forward), the structure of the junctions and reducible regions inside $T$ remains unchanged since no augmenting chain hitting $T$ is applied and there does not exist any merging or touching edge inside $T$ anymore.
If $y_T \in V(T)$ is an arbitrary leaf of $F_1^\final$ (that is definitely contained in a reducible region in $\calC_1^\final$), we have that $(y_T)$ is a stable $1$-chain whose tail is $T$. 

\medskip
\noindent
\textbf{Number of $1$-Chains.}
Considering all such components $T$ and the corresponding stable $1$-chains $(y_T)$, we have a set of at least $f-2\tau$ tail-independent stable $1$-chains.
$|M(1)| \geq f/5$ follows since we have $f - 2\tau \geq f - f/(50 \cdot H) > f/5 $
According to Assumption \ref{full:assumption:small-improvement}.

\subsection{Proof of \Cref{full:lem:potential}: Part (ii)}\label{full:sec:proof-incresing-potential2}

Fix a value $t \in [1, H]$, and assume that $M(t) = \{A_1,A_2, \ldots, A_q\}$ for some $q \in \mathbb{N}$.
The goal is to construct $M(t+1)$ such that it is a tail-independent set of stable $(t+1)$-chains and moreover $|M(t+1)| \geq \frac{k-2}{\Delta^\star - 1} \cdot q + f/5$.

\medskip
\noindent
\textbf{Roadmap.}
We start by providing some independent statements in \Cref{full:sec:two-claims} that will be used throughout this section.
Then, in \Cref{full:sec:structure-of-tails}, we elaborate on the structure of the tails of the stable $t$-chains of $M(t)$. 
We continue by providing some important objects in \Cref{full:sec:tail-exapnding,full:sec:stable-t+1-chains} that we need to describe the construction of $M(t+1)$.
In \Cref{full:sec:defining-potential-chains}, we construct a collection of sequences that are potential stable $(t+1)$-chains.
In \Cref{full:sec:def-M(t+1)}, we provide some properties such that if one of the constructed sequences satisfies these properties, then it must be a stable $(t+1)$-chain.
$M(t+1)$ is then defined as the set of such sequences satisfying those properties.
Then, in \Cref{full:sec:count-true-chains}, we provide an upper bound on the number of sequences that violate these properties.
Eventually, in \Cref{full:sec:putting-together}, we show the desired bound on the size of the set $M(t+1)$.

\subsubsection{Two Useful Claims}\label{full:sec:two-claims}
We start with two independent Claims \ref{full:claim:number-leafs} and \ref{full:claim:deficit-bound}.
These Claims and their corollaries will be used throughout this section.   

\begin{claim}\label{full:claim:number-leafs}
    Assume $\calF \subseteq G$ is an arbitrary forest.
    Let $H = T^\calF_{w \leftarrow z}$ be an arbitrary subtree of $\calF$ for $w,z \in V(\calF)$, $B \subseteq V(H)$, and $R$ be the collection of $(H, B)$-regions inside $H$ that are incident (according to the edges in $E(\calF)$) to \emph{at most one node} in $B$.
    Then, we have that $ |R| \geq 1 + \sum_{x \in B} (\deg_\calF(x) - 2)$.
\end{claim}

This claim follows from a simple counting argument.
We refer the reader to \Cref{full:sec:proof-number-leafs} for a complete formal proof.

\begin{corollary}\label{full:cor:number-leafs}
    If $H$ is a connected component of $\calF$ and $B \subseteq V(H)$.
    Then, the number of $(H,B)$-regions in $H$ that are incident (according to $E(\calF)$) on at most one node in $B$ is at least $1 + \sum_{x \in B} (\deg_\calF(x) - 2)$.
\end{corollary}

\begin{proof}
    This is trivial if $B = \emptyset$.
    Now, assume $B \neq \emptyset$ and let $x \in B$ be arbitrary.
    We can simply consider a dummy node $z^\star$ and a dummy edge $(z^\star, x)$ and define a new forest $\tilde{\calF}$ as follows; $V(\tilde{\calF}) := V(\calF) + z^\star$, and $E(\tilde{\calF}) := E(\calF) + (z^\star, x)$.
    Then, we have $H = T^{\tilde{\calF}}_{x \leftarrow z^\star}$.
    The degree of all the original nodes in $\calF$ and $\tilde{\calF}$ are the same except $\deg_{\tilde{\calF}}(x) = \deg_{\calF}(x) + 1$.
    We have the desired statement by \Cref{full:claim:number-leafs} for $\tilde{\calF}$ and $H = T^{\tilde{\calF}}_{x \leftarrow z^\star}$.
\end{proof}

\begin{claim}\label{full:claim:deficit-bound}
    Consider an arbitrary configuration $(\calF, B, \calC, D)$ maintained by the algorithm at some time throughout the entire execution of the algorithm.
    We have that $\deficit(\calF, B, \calC, D) \leq 2 \cdot \tau \cdot H$.
\end{claim}

The proof of this claim is deferred to \Cref{full:sec:proof-deficit-bound}.
It easily follows from Assumption \ref{full:assumption:small-improvement} and Property \ref{full:lem:apply-chain-P3} of \Cref{full:lem:apply-chain}.

\begin{corollary}\label{full:cor:sum-degree-is-large}
    Throughout the entire execution of the algorithm, for any arbitrary subset $B' \subseteq B$ of the set of junctions, we have that $\sum_{x \in B'} \deg_{\calF}(x) \geq k \cdot |B'| - f/10$.
\end{corollary}

\begin{proof}
    The proof easily follows since
    \begin{align*}
        \sum_{x \in B'} (k -\deg_{\calF}(x)) \leq \sum_{x \in B \cup D} (k -\deg_{\calF}(x)) = \deficit(\calF, B, \calC, D) \leq 2 \cdot \tau \cdot H \leq f/10.
    \end{align*}
    The first inequality holds since $\calF$ has maximum degree at most $k$, the second inequality follows from \Cref{full:claim:deficit-bound}, and the last inequality holds by  Assumption \ref{full:assumption:small-improvement}.
\end{proof}

\subsubsection{Structure of the Tails of Stable $t$-Chains}\label{full:sec:structure-of-tails}
Let $T_i$ be the tail of $A_i$ for each $i \in [1, q]$.
According to the definition, the sets $V(T_i)$ are mutually node-disjoint among $i \in [1, q]$.
Denote by $B_{t, i}^\inn$ the set of junctions in $T_i$ at the end of round $t$, i.e., $B_{t,i}^\inn := B_{t}^\final \cap V(T_i)$ for each $i \in [1, q]$.
Moreover, let $\calC_{t,i}^\inn$ be the collection of reducible $(\calF^\final_{t}, B^\final_{t})$-regions inside $T_i$ at the end of round $t$, i.e., $\calC_{t,i}^\inn := \{ C \in \calC_{t}^\final \mid V(C) \subseteq V(T_i)\} $ for each $i \in [1, q]$.
Note that according to the definition of a stable $t$-chain at the end of round $t$, there is no touched node inside $T_i$ and every $(\calF^\final_{t}, B^\final_{t})$-region in $T_i$ is a reducible region.
Moreover, define $B_t^\inn := \cup_{i=1}^q B_{t,i}^\inn$ and $\calC_{t}^\inn := \cup_{i=1}^q \calC_{t,i}^\inn$.

\begin{claim}\label{full:claim:size-C_t-inn}
    We have that $ |\calC_{t}^\inn| \geq q + \sum_{x \in B_{t}^\inn} \left(\deg_{\calF^\final_t}(x) - 2\right) $.
\end{claim}

\begin{proof}
Fix an $i \in [1, q]$.
If the length of $A_i$ is more than zero, according to \Cref{full:claim:number-leafs} for $\calF_t^\final$, $H = T_i$, and $B = B_{t,i}^\inn $, we have that $ |\calC_{t,i}^\inn| \geq 1 + \sum_{x \in B_{t,i}^\inn} \left(\deg_{\calF^\final_t}(x) - 2\right) $.
If the length of $A_i$ is zero ($T_i$ is an entire connected component of $\calF_t^\final$), we have the same inequality by \Cref{full:cor:number-leafs}.
By summing up these inequalities for all $i \in [1, q]$, we get the desired inequality.
\end{proof}

\subsubsection{Tail-Expanding Edges}\label{full:sec:tail-exapnding}

Fix some optimal spanning tree $T^\star$ of $G$ with maximum degree $\Delta^\star$.
Let $i \in [1, q]$ be arbitrary.
We refer to an edge $(u^\star,v^\star) \in E(T^\star)$ in the optimal spanning tree $T^\star$ as a $T_i$-\textbf{expanding} edge if the following holds;
\begin{enumerate}
    \item $u^\star$ is inside a reducible region of $\calC^\final_{t}$ contained entirely in $T_i$.
    \item $v^\star$ is outside of $T_i$.
\end{enumerate}
Note that we consider the pair $(u^\star, v^\star)$ as an ordered pair which means $(u^\star, v^\star)$ and $(v^\star, u^\star)$ are considered distinct in this definition.
We refer to an edge $(u^\star, v^\star) \in E(T^\star)$ as an expanding edge if there exists an $i \in [1, q]$ such that $(u^\star, v^\star)$ is $T_i$-expanding.

\begin{observation}\label{full:obs:expanding-types}
    For each expanding edge $(u^\star, v^\star)$, we have that $(u^\star, v^\star) \notin E(\calF^\final_t)$ unless $v^\star$ is the root of $T_i$ and $u^\star$ is the unique neighbor (according to the edges of $\calF^\final_t$) of $v^\star$ in $T_i$.
\end{observation}

Define
$ B^\out_{t} := \left\{v^\star \mid (u^\star, v^\star) \text{ is expanding}\right\} - B^\inn_t $.
In words, $B_{t}^\out$ is the set of all endpoints $v^\star$ of expanding edges that does not belong to $B^\inn_t$. 
We proceed by providing some properties of $B_t^\out$.

\begin{claim}\label{full:claim:B-out-in-B-init}
    We have $B_{t}^\out \subseteq B_{t+1}^\init$.
\end{claim}

\begin{proof}
If $(u^\star, v^\star) \in E(\calF_{t}^\final)$, we have that $v^\star$ is the root of the tail of some stable $t$-chain according to \Cref{full:obs:expanding-types}.
In this case, we have $v^\star \in B_t^\final$ by the definition of pseudo-chain (see Property \ref{full:pseudo-chain-P1} in \Cref{full:def:pseudo-chain}).
Now, assume $(u^\star, v^\star) \notin E(\calF_{t}^\final)$.
According to \Cref{full:obs:cond-end-round} and the definition of an expanding edge, we have that $v^\star \in B^\final_t \cup D_t^\final$.
Finally, $B_{t+1}^\init$ is defined as $B^\final_t \cup D_t^\final$, which concludes $v^\star \in B_{t+1}^\init$ in both cases above.
Hence, $B_t^\out \subseteq B_{t+1}^\init$.
\end{proof}

Later, we show that every $x \in B_t^\out$ remains a junction for the rest of the algorithm until the end of round $H$.
We provide the claim later when we complete the construction of $M(t+1)$.
We also have the following property about $B_t^\out$ whose proof is deferred to \Cref{full:sec:proof-size-B-out}.

\begin{claim}\label{full:claim:size-B-out}
    We have that $ |B_t^\out| \geq \frac{q - f/10}{\Delta^\star - 1} $.
\end{claim}

\subsubsection{Singly-Attached Regions}\label{full:sec:stable-t+1-chains}

Consider the collection of all $(\calF^\init_{t+1}, B^\out_{t})$-regions.
We refer to such a region as a \textbf{singly-attached} $(\calF^\init_{t+1}, B^\out_{t})$-region if it is adjacent (according to $E(\calF^\init_{t+1})$) to at most one node in $B^\out_{t}$.
If such a region $R$ is not adjacent to any node in $B_t^\out$, then $R$ itself must be an entire connected component of $\calF_{t+1}^\init$.
If such a region $R$ is adjacent to $v \in B_t^\out$, we refer to $v$ as the root of $R$.
Assume $\calS_{t+1}$ is the collection of all singly-attached $(\calF^\init_{t+1}, B^\out_{t})$-regions.

\begin{claim}\label{full:claim:size-singly-attached}
    We have that $|\calS_{t+1}| \geq (k-2) \cdot |B^\out_{t}| + 9f/10$.
\end{claim}

\begin{proof}
    Let $X$ be a connected component of $\calF_{t+1}^\init$.
    According to \Cref{full:cor:number-leafs}, the number of singly-attached $(\calF^\init_{t+1}, B^\out_{t})$-regions in $X$ is at least 
    \begin{equation}\label{full:eq:num-singly-attached}
        1 + \sum_{x \in B^\out_{t} \cap V(X)} \left(\deg_{\calF^\init_{t+1}}(x) - 2\right). 
    \end{equation}
    The number of connected components of $\calF_{t+1}^\init$ is at least $f-f/(100\cdot H)$ since we assumed that at most $\tau \leq f/(100\cdot H)$ many augmenting chains are applied, (see Assumption \ref{full:assumption:small-improvement}).
    Hence, by summing up \Cref{full:eq:num-singly-attached} over all connected components of $\calF_{t+1}^\init$, the total number of singly-attached $(\calF_{t+1}^\init, B_t^\out)$-regions is at least 
    \begin{align*}
        (f - f/(100 \cdot H)) + \sum_{x \in B_t^\out} \left(\deg_{\calF^\init_{t+1}}(x) - 2\right) 
        &\geq f - f/(100 \cdot H) + (k-2) \cdot |B_t^\out| - f/50 \\
        &\geq  (k-2) \cdot |B_t^\out| + 9f/10.
    \end{align*}
    The first inequality follows from \Cref{full:cor:sum-degree-is-large} for $B' := B_t^\out \subseteq B_{t+1}^\init$ (see \Cref{full:claim:B-out-in-B-init}).
\end{proof}

\subsubsection{Constructing Potential Stable $(t+1)$-Chains}\label{full:sec:defining-potential-chains}

Assume that $U \in \calS_{t+1}$ is an arbitrary singly-attached $(\calF^\init_{t+1}, B^\out_{t})$-region.
In the following, we introduce a potential stable $(t+1)$-chain denoted by $A_U$ whose tail is $U$.
The goal is to construct these pseudo-chains $A_U$ for $U \in \calS_{t+1}$ and argue that a large number of such chains are actually stable $(t+1)$-chains.
We have two cases;

\medskip
\noindent
\textbf{Case I: $U$ Does Not Have a Root.}
In this case, $U$ itself is an entire connected component of $\calF^\init_{t+1}$, then we define $A_U$ as the sequence $(y)$ of length $0$, where $y$ is an arbitrary leaf node $\calF^\init_{t+1}$.
In this case, we refer to $A_U$ as a \textbf{new} chain. 

\medskip
\noindent
\textbf{Case II: $U$ Has a Root.}
Denote the root of $U$ by $x$.
According to the definition, there exists a stable $t$-chain $A_i = (w_0,z_1,w_1,\ldots, z_{h}, w_h) \in M(t)$ (for some index $i \in [1, q]$) of length $h \leq t-1$ whose tail is $T_i$.
Moreover, there exists a $T_i$-expanding edge $(u^\star,v^\star)$ such that $ u^\star \in T_i$, $v^\star = x \notin T_i$.
Now, we consider two subcases as follows:
\begin{itemize}
\item 
If $(u^\star, v^\star) \notin E(\calF^\init_{t+1})$, we consider any arbitrary leaf node $y$ of $\calF^\init_{t+1}$ inside $U$, and define $A_U := (w_0,z_1,w_1,\ldots, z_h, u^\star, x, y)$.
In this case, we refer to $A_U$ as an \textbf{extended} chain.
\item
If $(u^\star, v^\star) \in E(\calF^\init_{t+1})$, we have that $v^\star$ is the tail of $T_i$ according to \Cref{full:obs:expanding-types}, and we define $A_U := A_i$.
In this case, we refer to $A_U$ as a \textbf{repetitive} chain.
\end{itemize}

So far, for any $U \in \calS_{t+1}$, we have constructed a sequence $A_U$ (as a potential stable $(t+1)$-chain).
Now, we intend to count the number of $U \in \calS_{t+1}$ such that $A_U$ is indeed a stable $(t+1)$-chain.
Recall that the definition of a stable $(t+1)$-chain is w.r.t.~the configuration at the end of round $t+1$.
But we defined $A_U$ w.r.t.~the configuration at the beginning of round $t+1$.
We will show that many of these sequences $A_U$ remain a pseudo-chain w.r.t.~$\conf{\final}{t+1}$ and moreover, they are stable $(t+1)$-chains.

\subsubsection{Definition of $M(t+1)$}\label{full:sec:def-M(t+1)}

Let $A_U = (w_0, z_1, w_1, \ldots, z_{h}, w_{h})$ for some $U \in S_{t+1}$.
We aim to see when $A_U$ refuses to be a stable $(t+1)$-chain.
If $A_U$ is a pseudo-chain w.r.t.~$\conf{\init}{t+1}$, we call $A_U$ as \textbf{qualified}.
The only reason that can cause $A_U$ not to be a pseudo-chain w.r.t.~$\conf{\init}{t+1}$ is that the leaf node $y$ inside $U$ that is used to construct $A_U$ is a junction as well.
Later, we show that this can not happen too many times since the $\deficit$ of any configuration maintained by the algorithm is bounded and leaf nodes have degree $1$.

Now, assume that $A_U$ is qualified.
Let $(u,v)$ be an effective edge that has been applied at some time in the algorithm from the beginning of round $t+1$ until the end of round $H$.
We say that $(u,v)$ \textbf{affects} $A_U$ if after applying $(u, v)$ at least one of the following happens, 
\begin{enumerate}
    \item\label{full:destroy-P1} $z_{h}$ is removed from $B$ and/or the edge connecting $z_h$ and $T_{w_h \leftarrow z_h}$ is removed from the forest (this condition applies only if $A_U$ has length at least $1$).

    \item\label{full:destroy-P2} An edge crossing $U$ is inserted into the forest. 
    
    \item\label{full:destroy-P3} A touched region emerges inside $U$.
\end{enumerate}
Note that this definition does not include applying a merging edge between two reducible regions inside $U$, i.e., if $(u,v)$ is a merging edge between two reducible regions inside $U$, we do \emph{not} say that $(u,v)$ affects $A_U$.
Let $(u,v)$ be an effective edge that affects $A_U$.
We say that $(u,v)$ \textbf{destroys} $A_U$ if during the procedure of the algorithm from the beginning of round $t+1$ until the end of round $H$, $(u,v)$ is the very first edge that affects $A_U$.

\medskip
\noindent
\textbf{Definition of $M(t+1)$.}
Eventually, we define $M(t+1)$ to be the set of all sequences $A_U$ for $U \in \calS_{t+1}$ that \emph{are qualified} and \emph{are not destroyed} throughout the algorithm from the beginning of round $t+1$ until the end of round $H$.
In the following, we show the correctness of our definition of $M(t+1)$.

\begin{claim}\label{full:claim:v-star-in-B}
    If $A_U$ is qualified and is not destroyed by any effective edge throughout the entire algorithm from the beginning of round $t+1$ until the end of round $H$, then $A_U$ remains a pseudo-chain until the end of round $H$.
\end{claim}

The proof of this claim is deferred to \Cref{full:sec:proof-v-star-in-B}.
This concludes that if $A_U$ is not destroyed, then it is a pseudo-chain of length at most $t$ w.r.t.~$\conf{\final}{t+1}$.
This shows that our definition of $M(t+1)$ makes sense, since at the beginning of \Cref{full:sec:def-stable-chains}, we defined stable $t$-chains as pseudo-chains w.r.t.~$\conf{\final}{t}$ for each $ t\in[1, H]$.

\begin{claim}\label{full:claim:not-destroyed-is-stable}
    If $A_U$ is qualified and is not destroyed by any effective edge throughout the entire algorithm from the beginning of round $t+1$ until the end of round $H$, then $A_U$ is a stable $(t+1)$-chain.
\end{claim}

The proof of this claim is deferred to \Cref{full:sec:proof-not-destroyed-is-stable}.
This claim shows that $M(t+1)$ must be a set of stable $(t+1)$-chains.
Moreover, the tails of these stable $(t+1)$-chains are singly-attached $(\calF_{t+1}^\init, B_t^\out)$-regions which concludes they are mutually node-disjoint.
Hence, our construction is correct, and $M(t+1)$ is indeed a tail-independent set of stable $(t+1)$-chains.
It remains to analyze the size of $M(t+1)$.
For this reason, we must bound the number of sequences $A_U$ that are not qualified, and the number of qualified sequences $A_U$ that are destroyed.

\subsubsection{Counting Stable $(t+1)$-Chains}\label{full:sec:count-true-chains}

\textbf{Number of Qualified Sequences.}
We start by showing that a large number of sequences $A_U$ for $U \in \calS_{t+1}$ are qualified.

\begin{claim}\label{full:claim:well-defined-AU}
    The number of elements $U \in \calS_{t+1}$ such that $A_U$ is not qualified is bounded by $f/10$.
\end{claim}

\noindent
\textit{Proof Sketch.}
The only reason that $A_U$ is not a pseudo-chain w.r.t.~$\conf{\init}{t+1}$ is that $U$ contains a leaf node that is a junction as well.
But the number of junctions with low degree (in this case, degree equal to $1$ for leaf nodes) is bounded.
We refer the reader to \Cref{full:sec:proof-well-defined-AU} for a complete formal proof.

\medskip
\noindent
\textbf{A Useful Property.}
We have the following property about the execution of the algorithm from the start of round $t+1$ until the end of round $H$.

\begin{claim}\label{full:claim:levels-are-bounded}
    Throughout the execution of the algorithm from the start of round $t+1$ until the end of round $H$, the following holds for any $A_U = (w_0, z_1, \ldots, z_{h+1}, w_{h+1}) \in M(t+1)$;
    \begin{enumerate}
        \item\label{full:level-edge-bound} If $\lvl_{(w_{h}, z_{h+1})}$ reaches $t$, then every call to the $\backSearch$ on $(w_{h}, z_{h+1})$ for the rest of the algorithm must be successful.
        Accordingly, $\lvl_{(w_{h}, z_{h+1})}$ would never exceed $t$.
        \item\label{full:level-node-bound} If $\lvl_{z_{h+1}}$ reaches $t+1$, then every call to the $\backSearch$ on $z_{h+1}$ for the rest of the algorithm must be successful.
        Accordingly, $\lvl_{z_{h+1}}$ would never exceed $t+1$.
    \end{enumerate}
\end{claim}

\noindent
\textit{Proof Sketch.}
In summary, the main reason for its correctness is the delicate procedure of the $\backSearch$ and the assumption that $A_U$ is not destroyed (see the definition of $M(t+1)$).
More precisely, the nodes and edges corresponding to $A$ always remain successful choices for the recursive calls inside the $\backSearch$ subroutine.
Hence, the levels of those specific nodes and edges never exceed the specified values.
We refer the reader to \Cref{full:sec:proof-levels-are-bounded} for a complete formal proof.

\medskip
\noindent
\textbf{Number of Qualified Sequences That Are Destroyed.}
Now, we count how many of the qualified sequences $A_U$ can be destroyed throughout the execution of the algorithm from the beginning of round $t+1$ until the end of round $H$.
We have the following two claims.

\begin{claim}\label{full:claim:destroying-is-improving}
    If $(u,v)$ destroys at least one qualified sequence $A_U$ for some $U \in \calS_{t+1}$, then $(u,v)$ is an improving edge.
\end{claim}

\noindent
\textit{Proof Sketch.}
The main reason for its correctness is that first, all of the regions inside $U$ are reducible so far, and second, $(u,v)$ must cross $U$, which means the algorithm calls $\backSearch$ on the root of $A_U$.
Intuitively, the nodes and edges corresponding to the sequence $A_U$ are possible choices for the recursive calls in $\backSearch$.
As a result, the $\backSearch$ can not fail, which means that an augmenting chain will be found, i.e., $(u,v)$ is an improving edge.
We refer the reader to \Cref{full:sec:proof-destroying-is-improving} for a formal complete proof.

\begin{claim}\label{full:claim:number-of-destroys}
    Each improving edge $(u,v)$ can destroy at most $20H$ many qualified $A_U$.
\end{claim}

\noindent
\textit{Proof Sketch.}
The main reason for its correctness is that the length of any augmenting chain $A$ applied by the algorithm is bounded by $H$.
The number of reducible regions that are touched by applying $A$ is bounded by $O(H)$ (see Property \ref{full:lem:apply-chain-P4} of \Cref{full:lem:apply-chain}), and moreover, the number of critical edges of $A$ is bounded by $O(H)$ as well (see Property \ref{full:lem:apply-chain-P2} of \Cref{full:lem:apply-chain}).
We refer the reader to \Cref{full:sec:proof-number-of-destroys} for a formal, complete proof.

\subsubsection{Putting Everything Together}\label{full:sec:putting-together}

We started by assuming $M(t) = \{A_1,A_2,\ldots,A_q\}$.
Then, for each $i \in [1,q]$, we considered all expanding edges $(u^\star, v^\star)$ exiting the tail $T_i$ of $A_i$.
The endpoints of these expanding edges introduced the set $B_t^\out$.
Then, by defining $\calS_{t+1}$ as the set of all singly-attached $(\calF_{t+1}^\final, B_t^\out)$-regions, we constructed the set $M(t+1)$ of stable $(t+1)$-chains whose tails are in $\calS_{t+1}$ and are mutually node-disjoint.
It remains to show that this set is large enough, i.e., $|M(t+1)| \geq \frac{k-2}{\Delta^\star-1}  \cdot q + f/5$.

\medskip
\noindent
\textbf{Bounding the Size of $M(t+1)$.}
Combining Claims \ref{full:claim:size-B-out} and \ref{full:claim:size-singly-attached} concludes 
$$|\calS_{t+1}| \geq \frac{k-2}{\Delta^\star-1} \cdot (q-f/10) + 9f/10 \geq \frac{k-2}{\Delta^\star-1} \cdot q + 5f/10, $$
where the last inequality holds since $k \leq 2\Delta^\star+1$ (see \Cref{full:eq:param}) and $\Delta^\star \geq 2$.
According to \Cref{full:claim:well-defined-AU}, the number of qualified sequences $A_U$ for $U \in \calS_{t+1}$ is at least $ |\calS_{t+1}| - f/10 \geq \frac{k-2}{\Delta^\star-1} \cdot q + 4f/10$.
Combining with Claims \ref{full:claim:destroying-is-improving} and \ref{full:claim:number-of-destroys}, we conclude that among these qualified sequences $A_U$, there are at most $\tau \cdot 20H \leq f/5$ many that are destroyed (the inequality holds by Assumption \ref{full:assumption:small-improvement}).
Hence, the number of qualified sequences that are not destroyed, i.e., stable $(t+1)$-chains in $M(t+1)$, is at least 
$$\frac{k-2}{\Delta^\star-1} \cdot q + 4f/10 - f/5 = \frac{k-2}{\Delta^\star-1} \cdot q + f/5. $$

\subsection{Deferred Proofs}\label{full:sec:deferred-proofs}

\subsubsection{Proof of \Cref{full:claim:number-leafs}}\label{full:sec:proof-number-leafs}

If $B = \emptyset$, the claim is trivial.
Now, consider $H = T^\calF_{w \leftarrow z}$ as a rooted tree with $z$ being its root.
Define an ordering $B = \{x_1,\ldots,x_b\}$ of the elements of $B$ according to the increasing distance from the root $z$.
We remove elements of $B$ in this order iteratively.
Let $B_0 := \emptyset$ and $B_i := \{x_1,\ldots, x_i\}$ for each $i \in [1, b]$.
Initially, we have one component $H$ that is incident to at most one node in $B_0$.
For each $i \in [1, b]$, assume that we have already removed all of the nodes of $B_{i-1}$ from $H$, and the number of $(H,B_{i-1})$-regions that are incident to (w.r.t.~$E(\calF)$) at most one node in $B_{i-1}$ is at least $1 + \sum_{x \in B_{i-1}} (\deg_{\calF}(x) - 2)$.

When we further remove $x_i$, the component containing $x_i$ is split into some new components.
At most one of these components can now be incident to two nodes in $B_i$ (the potential component derived by removing the edge connecting $x_i$ to its parent), and all of the other new components (derived by removing the edges between $x_i$ and its children) are incident to at most one node in $B_i$ (which is $x_i$ itself).
As a result, after removing $x_i$, one component is destroyed and at least $\deg_{\calF}(x_i) - 1$ new components incident to at most one node in $B_i$ are formed.
As a result, the number of components incident to at most one node in $B_i$ is at least
$$\left(1 + \sum_{x \in B_{i-1}} (\deg_{\calF}(x)-2)\right) + \left(\deg_\calF(x_i) - 1 -1\right) = 1 + \sum_{x \in B_{i}} (\deg_{\calF}(x)-2) . $$
Hence, the claim follows for $i = b$, after removing all nodes in $B$.

\subsubsection{Proof of \Cref{full:claim:deficit-bound}}\label{full:sec:proof-deficit-bound}

It is trivial that $\deficit\conf{\init}{1} = 0$.
Moreover, when transferring from the end of round $t$ to the beginning of round $t+1$, we have $$\deficit\conf{\init}{t+1} = \deficit\conf{\final}{t}, $$
since $B^\init_{t+1} \cup D^\init_{t+1} = B^\final_{t} \cup D^\final_{t}$.
Assume that an effective edge $(u,v)$ has been applied to $\conf{\old}{}$, and we have the new Configuration $\conf{\new}{}$.
We have the following cases;
\begin{itemize}
    \item $(u,v)$ is an improving edge.
    According to \Cref{full:lem:apply-chain} (Property 3), we conclude that the deficit can increase at most $2H$ units since the algorithm only applies augmenting chains of length at most $H$.
    
    \item $(u,v)$ is a merging edge.
    Then, we have $B^\new \cup D^\new \subseteq B^\old \cup D^\old$ that concludes the deficit will not increase since the degree of all nodes in the forest is at most $k$.

    \item $(u,v)$ is a touching edge.
    In this case, $B^\new = B^\old$ and $D^\new \supseteq D^\old$, but every node $x \in D^\new - D^\old$ has degree exactly $k$ in the forest.
    Hence, the deficit will not increase.
\end{itemize}

Finally, by \Cref{full:assumption:small-improvement}, we apply $\tau$ many augmenting chains throughout the entire algorithm, which concludes that the deficit of any configuration throughout the algorithm is bounded by $2 \cdot \tau \cdot H$.

\subsubsection{Proof of \Cref{full:claim:size-B-out}}\label{full:sec:proof-size-B-out}

Recall the definitions of $T_i$, $B_{t,i}^\inn$ and $\calC_{t,i}^\inn$ for each $i \in [1, q]$, $B_t^\inn$, $\calC_t^\inn$, and $B_t^\out$.
We define a collection $\calP$ of subsets of nodes of $V(G)$ as follows.
Every node in $B_t^\inn$ as well as $B_t^\out$ appears as a singleton in $\calP$.
For every region $C \in \calC_t^\inn$, the set $V(C)$ appears as an element in $\calP$.
According to the definitions of $\calC_t^\inn$, $B_t^\inn$, and $B_t^\out$, we have that the elements of $\calP$ are mutually node-disjoint.
Now, consider the optimal spanning tree $T^\star$.
The sets in $\calP$ must be connected together in $T^\star$, which concludes that there are at least $|\calP| - 1$ many edges of $E(T^\star)$ that have one endpoint in an element of $\calP$ and the other endpoint is outside of that element.
Let $E^\star$ be the set of such edges.

\begin{claim}
    Every edge $e^\star \in E^\star$ has an endpoint in $B_t^\inn \cup B_t^\out$.
\end{claim}

\begin{proof}
$e^\star \in E^\star$ has an endpoint in an element of $\calP$.
It can be in one of $B_t^\inn$ and $B_t^\out$, or it can be in a region $C \in \calC_t^\inn$.
There is nothing to prove in the former case.
Now, assume that one endpoint of $e^\star$ is in a region $C \in \calC_{t,i}^\inn$ for some $i \in [1,q]$.
According to \Cref{full:obs:cond-end-round}, for every $i \in [1, q]$, there is no edge between different regions $C, C' \in \calC_{t,i}^\inn$ at the end of round $t$.
As a result, the other endpoint of $e^\star$ must be either in $B_t^\inn$, or outside of $T_i$. 
In both cases, according to the definition of $B_t^\out$, we have that the other endpoint must be in $B_t^\inn \cup B_t^\out$.
\end{proof}

\noindent
\textbf{Concluding \Cref{full:claim:size-B-out}.}
By a simple counting argument, we have that
\begin{align*}
    |E^\star| &\geq |\calP| - 1 = |\calC_t^\inn| + |B_t^\inn| + |B_t^\out| - 1 \\
    &\geq q + \sum_{x \in B_t^\inn} \left( \deg_{\calF_t^\final}(x) - 2 \right) + |B_t^\inn| + |B_t^\out| - 1 \\
    &\geq q + (k-2) \cdot |B_t^\inn| - f/50 + |B_t^\inn| + |B_t^\out| - 1.
\end{align*}
The second inequality follows from \Cref{full:claim:size-C_t-inn}, and the third inequality follows from \Cref{full:cor:sum-degree-is-large}.
Moreover, $(|B_t^\inn| + |B_t^\out|) \cdot \Delta^\star \geq |E^\star|$ since every edge $e^\star \in E^\star$ has an endpoint in $B_t^\inn \cup B_t^\out$.
Hence,
$$ (|B_t^\inn| + |B_t^\out|) \cdot \Delta^\star \geq q + (k-1) \cdot |B_t^\inn| + |B_t^\out| - f/50 - 1. $$
Together with $k \geq \Delta^\star + 1$, we conclude that
$|B_t^\out| \geq \frac{q - 1 - f/50}{\Delta^\star - 1} \geq \frac{q - f/10}{\Delta^\star - 1} $
as desired.
The last inequality follows from $f \geq 50$ (see the assumption of \Cref{full:lem:main}).

\subsubsection{Proof of \Cref{full:claim:v-star-in-B}}\label{full:sec:proof-v-star-in-B}

Recall that each $A_U$ can be new, extended, or repetitive according to our construction in \Cref{full:sec:defining-potential-chains}.
Hence, we have the following cases.

\medskip
\noindent
\textbf{$A_U$ Is New.}
In this case, $U$ itself is an entire connected component of $\calF^\init_{t+1}$, and we defined $A_U = (y)$ where $y$ is an arbitrary leaf node inside $U$.
Since $A_U$ is qualified, it is a pseudo-chain w.r.t.~$\conf{\init}{t+1}$, i.e., $y$ is inside a reducible region of $U$ in $\calC_{t+1}^\init$.
Moreover, since $A_U$ is not destroyed, it means that $y$ always remains inside a reducible region in $U$ (see Condition \ref{full:destroy-P3}).
As a result, $A_U = (y)$ always remains a pseudo-chain w.r.t.~any configuration until the end of round $H$.

\medskip
\noindent
\textbf{$A_U$ Is Extended.}
In this case, we defined $A_U := (w_0, z_1, w_1, \ldots, z_{h}, u^\star, v^\star, y)$ such that 1) $ A_i = (w_0, z_1, w_1, \ldots, z_{h}, w_{h}) \in M(t)$ is a stable $t$-chain (for some $i \in [1, q]$) of length $h \leq t-1$ whose tail is $T_i$, 2) $(u^\star, v^\star)$ is a $T_i$-expanding edge such that $(u^\star, v^\star) \notin E(\calF_{t+1}^\init)$, and 3) $y$ is an arbitrary leaf node of $\calF_{t+1}^\init$ inside $U$.
Since $A_i$ is a stable $t$-chain, it remains a pseudo-chain w.r.t.~any configuration after the beginning of round $t+1$ (according to \Cref{full:obs:remain-pseudo-chain}). 
Hence, to prove \Cref{full:claim:v-star-in-B} for $A_U$, it suffices to show the following;
\begin{itemize}
    \item Property \ref{full:pseudo-chain-P1} of \Cref{full:def:pseudo-chain} for $v^\star$, i.e., $v^\star$ remains a junction.
    Initially, $v^\star \in B_{t+1}^\init$ since we assumed that $A_U$ is qualified. 
    The fact that $v^\star$ remains in the set of junctions follows directly from the definition of $A_U$ not being destroyed (see Condition \ref{full:destroy-P1}).

    \item Property \ref{full:pseudo-chain-P2} of \Cref{full:def:pseudo-chain} for $v^\star$ and $y$, i.e., $v^\star$ and $y$ are in the same connected component of the forest and $y$ belongs to a reducible region.
    This is correct w.r.t.~$\conf{\init}{t+1}$ since $A_U$ is qualified. 
    The fact that $y$ and $v^\star$ remain in the same connected component of the forest follows directly from the definition of $A_U$ not being destroyed (see Condition \ref{full:destroy-P1} for $A_U$).
    Moreover, according to Condition \ref{full:destroy-P3} for $A_U$, since no touched node emerges inside the tail of $A_U$ (i.e., $U$), $y \in U$ remains inside a reducible region of $U$.
    
    \item Property \ref{full:pseudo-chain-P3} of \Cref{full:def:pseudo-chain} for $v^\star$ and $T^{\calF}_{u^\star \leftarrow z_h}$, i.e., $ v^\star \notin T^{\calF}_{u^\star \leftarrow z_h}$.
    This property holds w.r.t. $\conf{\init}{t+1}$ since $A_U$ is qualified.
    Since $A_i$ is a stable $t$-chain the structure of its tail $T^{\calF}_{u^\star \leftarrow z_h}$ never changes.
    More precisely, no edge incident on $T^{\calF}_{u^\star \leftarrow z_h}$ is ever inserted into the forest since $A_i$ is stable and not destroyed (see Condition \ref{full:destroy-P2} for $A_i$ not being destroyed).
    Hence, $v^\star$ always remains outside of the block connecting $z_h$ to $u^\star$ in $A_U$.

    \item Property \ref{full:pseudo-chain-P4} of \Cref{full:def:pseudo-chain} for $(u^\star, v^\star)$, i.e., $ (u^\star, v^\star) \notin E(\calF)$.
    This property holds w.r.t. $\conf{\init}{t+1}$ since we assumed $A_U$ is extended.
    The edge $(u^\star, v^\star)$ that is incident on the tail of $A_i$ is never inserted into the forest since $A_i$ is stable and not destroyed (see Condition \ref{full:destroy-P2} for $A_i$ not being destroyed).
\end{itemize}
As a result, $A_U$ remains a pseudo-chain.

\medskip
\noindent
\textbf{$A_U$ Is Repetitive.}
In this case, we have $A_U \in M(t)$ which is a stable $t$-chain.
As a result, it is a pseudo-chain at the beginning of round $t+1$ and remains a pseudo-chain until the end of round $H$.

\subsubsection{Proof of \Cref{full:claim:not-destroyed-is-stable}}\label{full:sec:proof-not-destroyed-is-stable}

Recall that each $A_U$ can be new, extended, or repetitive according to our construction in \Cref{full:sec:defining-potential-chains}.
Hence, we have the following cases.

\medskip
\noindent
\textbf{$A_U$ Is New.}
In this case, $U$ itself is an entire connected component of $\calF^\init_{t+1}$, and we defined $A_U = (y)$ where $y$ is an arbitrary leaf node of $\calF_{t+1}^\init$ inside $U$.
Property \ref{full:P1:cond:stable-1} of \Cref{full:def:stable-chain} for $A_U$ holds vacuously.
Now, we show Property \ref{full:P1:cond:stable-2} of \Cref{full:def:stable-chain} holds as well.

At the beginning of round $t+1$, all regions inside $U$ are reducible regions, and $y$ belongs to such a region since $A_U$ is qualified.
According to Condition \ref{full:destroy-P3}, since $A_U$ is not destroyed, no touched node will appear in $U$.
Now, consider the end of round $t+1$.\footnote{recall that for each $t \in [1, H]$, the stable $t$-chains are defined w.r.t.~the configuration at the end of round $t$}
We show that the structure of $U$ is settled at the end of round $t+1$.
More precisely, we have the following;
\begin{itemize}
    \item
    All of the $(\calF_{t+1}^\final, B_{t+1}^\final)$-regions are reducible (because we assumed $A_U$ is not destroyed, then no touched node emerges in $U$).

    \item 
    There is no edge between two reducible regions of $\calC_{t+1}^\final$ inside $U$ (see \Cref{full:obs:cond-end-round} for the end of round $t+1$)
    
\end{itemize}
Moreover, $A_U$ is not destroyed, which concludes Property \ref{full:P1:cond:stable-2} of \Cref{full:def:stable-chain} for $A_U$ holds.

\medskip
\noindent
\textbf{$A_U$ Is Extended.}
In this case, we defined $A_U := (w_0, z_1, w_1, \ldots, z_{h}, u^\star, v^\star, y)$ such that 1) $ A_i = (w_0, z_1, w_1, \ldots, z_{h}, w_{h}) \in M(t)$ is a stable $t$-chain (for some $i \in [1, q]$) of length $h \leq t-1$ whose tail is $T_i$, 2) $(u^\star, v^\star)$ is a $T_i$-expanding edge such that $(u^\star, v^\star) \notin E(\calF_{t+1}^\init)$, and 3) $y$ is an arbitrary leaf node of $\calF_{t+1}^\init$ inside $U$.
We have the following.
\begin{itemize}
\item 
Since $A_i$ is a stable $t$-chain, for all $i \in [1, h]$, $z_i$ is not removed from the set of junctions, and the edge $(z_i,y_i)$ connecting $z_i$ to the $i^\th$ block of $A_i$ (as well as the $i^\th$ block of $A_U$) is not removed from the forest after the beginning of round $t$.
Moreover, $A_U$ is not destroyed, which concludes $v^\star$ is not removed from the set of junctions, and the edge connecting $v^\star$ to the last block of $A_U$ (containing $y$) is not removed from the forest after the beginning of round $t+1$ (see Condition \ref{full:destroy-P1} of $A_U$ being destroyed).
As a result, we have Property \ref{full:P1:cond:stable-1} of \Cref{full:def:stable-chain} for $A_U$.

\item
Property \ref{full:P1:cond:stable-2} of \Cref{full:def:stable-chain} for $A_U$ follows with a similar argument as in the previous case above (when $A_U$ is a new chain).
\end{itemize}

\medskip
\noindent
\textbf{$A_U$ Is Repetitive.}
In this case, we have $A_U \in M(t)$ which is a stable $t$-chain.
As a result, it is a stable $(t+1)$-chain as well, according to the definition.

\subsubsection{Proof of \Cref{full:claim:well-defined-AU}}\label{full:sec:proof-well-defined-AU}

First, we show that if $A_U$ is not qualified, there must exist a node in $U$ that is both a leaf node of $\calF_{t+1}^\init$ and a junction in $B_{t+1}^\init$.
Then, we bound the number of such nodes.
Recall that each $A_U$ can be new, extended, or repetitive according to our construction in \Cref{full:sec:defining-potential-chains}.
Hence, we have the following cases.

\medskip
\noindent
\textbf{$A_U$ Is New.}
If $A_U$ is a new chain, we defined $A_U = (y)$ where $y$ is an arbitrary leaf node in $U$. 
All of the Properties of \Cref{full:def:pseudo-chain} for $A_U = (y)$ w.r.t.~$\conf{\init}{t+1}$ become vacuous except that $y$ belongs to a reducible region in $\calC_{t+1}^\init$.
At the beginning of round $t+1$ all $(\calF_{t+1}^\init, B_{t+1}^\init)$-regions are reducible regions.
Hence, if $y$ does not belong to a reducible region, we must have that $y \in B_{t+1}^\init$.
Hence, there is a leaf node in $U$ which is a junction as well.

\medskip
\noindent
\textbf{$A_U$ Is Extended.}
If $A_U$ is an extended chain, we defined $A_U := (w_0, z_1, w_1, \ldots, z_{h}, u^\star, v^\star, y)$ such that 1) $ A_i = (w_0, z_1, w_1, \ldots, z_{h}, w_{h}) \in M(t)$ is a stable $t$-chain (for some $i \in [1, q]$) of length $h \leq t-1$ whose tail is $T_i$, 2) $(u^\star, v^\star)$ is a $T_i$-expanding edge such that $(u^\star, v^\star) \notin E(\calF_{t+1}^\init)$, and 3) $y$ is a leaf node inside $U$.
Since $A_i$ is a stable $t$-chain, it remains a pseudo-chain w.r.t.~$\conf{\init}{t+1}$ according to \Cref{full:obs:remain-pseudo-chain}.

Now, we analyze when $A_U$ refuses to be a pseudo-chain w.r.t.~$\conf{\init}{t+1}$.
Property \ref{full:pseudo-chain-P1} of \Cref{full:def:pseudo-chain} for $A_U$ follows from $A_i$ being a pseudo-chain and $v^\star \in B_t^\out \subseteq B_{t+1}^\init$ (see \Cref{full:claim:B-out-in-B-init}).
Property \ref{full:pseudo-chain-P3} of \Cref{full:def:pseudo-chain} for $A_U$ follows from $A_i$ being a pseudo-chain and $(u^\star, v^\star)$ being an expanding edge.
Property \ref{full:pseudo-chain-P4} of \Cref{full:def:pseudo-chain} for $A_U$ follows from $A_i$ being a pseudo-chain and $(u^\star, v^\star) \notin E(\calF_{t+1}^\init)$ (since $A_U$ is extended).
Property \ref{full:pseudo-chain-P2} of \Cref{full:def:pseudo-chain} for $A_U$ holds except maybe for the last node $y$ (since $A_i$ is a pseudo-chain and $(u^\star, v^\star)$ is an expanding edge).
As a result, if $A_U$ is not a pseudo-chain w.r.t.~$\conf{\init}{t+1}$, we must have that $y$ does not belong to a reducible region in $\calC_{t+1}^\init$.
Similar to the explanation of the case where $A_U$ is a new chain, we must have that $y \in B_{t+1}^\init$.
Hence, there is a leaf node in $U$ which is a junction as well.

\medskip
\noindent
\textbf{$A_U$ Is Repetitive.}
If $A_U$ is a repetitive chain, i.e., $A_U \in M(t)$, according to \Cref{full:obs:remain-pseudo-chain}, $A_U$ is a pseudo-chain.
So, $A_U$ is qualified.

\medskip
\noindent
\textbf{Bounding the Number of Leaf Nodes That Are Junctions.}
It remains to bound the number of $U \in \calS_{t+1}$ such that there exists a leaf node $y$ of $\calF_{t+1}^\init$ inside $U$ that is a junction as well.
Assume $B^\text{leaf}_{t+1}$ is the set of leaf nodes that are junctions as well. 
According to \Cref{full:cor:sum-degree-is-large} for this subset $B' = B^\text{leaf}_{t+1} \subseteq B_{t+1}^\init$ of junctions, we conclude that 
$$ |B^\text{leaf}_{t+1}| = \sum_{x \in B^\text{leaf}_{t+1}} \deg_{\calF_{t+1}^\init}(x) \geq k \cdot |B^\text{leaf}_{t+1}| - f/50, $$
Since elements of $\calS_{t+1}$ are node-disjoint, we conclude that there are at most $f/50$ many $U \in \calS_{t+1}$ that contain a leaf node $y$ that is a junction as well.

\subsubsection{Proof of \Cref{full:claim:levels-are-bounded}}\label{full:sec:proof-levels-are-bounded}

The proof is by induction on $t$.
Let $A_U = (w_0, z_1, \ldots, w_h, z_{h+1}, w_{h+1}) \in M(t+1)$.
Recall from \Cref{full:sec:defining-potential-chains} that each such chain $A_U$ can be new, extended, or repetitive.
If $A_U$ is a new chain, then the claim is vacuously true since the length of $A_U$ is zero (there is no such edge $(w_h,z_{h+1})$ and node $z_{h+1}$).
If $A_U$ is a repetitive chain, then $A_U \in M(t)$ and the claim follows inductively.
For the rest of the proof, we assume that $A_U$ is an extended chain.
we defined $A_U := (w_0, z_1, w_1, \ldots, z_{h}, u^\star, v^\star, y)$ such that 1) $ A_i = (w_0, z_1, w_1, \ldots, z_{h}, w_{h}') \in M(t)$ is a stable $t$-chain (for some $i \in [1, q]$) of length $h \leq t-1$ whose tail is $T_i$, 2) $(u^\star, v^\star)$ is a $T_i$-expanding edge such that $(u^\star, v^\star) \notin E(\calF_{t+1}^\init)$, and 3) $y$ is a leaf node inside $U$.
We have $w_h = u^\star$, $z_{h+1} = v^\star$ and $w_{h+1} = y$.

\begin{observation}\label{full:obs:levels-round-t}
    For any $t \in [1, H]$, throughout the entire round $t$, the level of all nodes and edges is bounded by $t$, and all of the augmenting chains applied in this round have length $\leq t$.
\end{observation}

\begin{proof}
    This observation easily follows from the $\findChain$ procedure and the fact that in round $t$, we only call $\findChain(h,(u,v))$ on parameter $h = t$. 
\end{proof}

\paragraph{Proof of \Cref{full:level-edge-bound} in \Cref{full:claim:levels-are-bounded}.}
According to \Cref{full:obs:levels-round-t} the level of $(w_{h}, z_{h+1})$ is bounded by $t$ at the beginning of round $t+1$.
Now, assume that the level of $(w_{h}, z_{h+1})$ becomes $t$ (or was $t$ from the beginning of round $t$) and we have a call to $\backSearch$ on $(w_h,z_{h+1})$.
Denote by $\conf{}{}$ the configuration maintained by the algorithm at the time of this call.
We show that this $\backSearch$ will be successful.
If $w_{h}$ and $z_{h+1}$ are in different connected components of $\calF$, then the $\backSearch$ successfully returns $(w_{h}, z_{h+1})$ according to the description (see Line \ref{full:line:dif-comps} of \Cref{full:code:backward-search-edge}).
For the rest of the proof, we assume that $w_{h}$ and $z_{h+1}$ are in the same connected component of $\calF$.

\begin{claim}\label{full:claim:zh-ramain-juncion-on-path}
    Throughout the execution of the algorithm after the start of round $t+1$, we have that:
    \begin{itemize}
        \item $z_h \in B$ is a junction, and
        \item we have that $z_{h} \in P^\calF_{w_{h}, z_{h+1}}$.
    \end{itemize}
\end{claim}

\begin{proof}
    This follows from the definition of $A_i \in M(t)$ being a stable $t$-chain.
    The tail of $A_i$ is $T_i$, and the root of $T_i$ is $z_h$.
    Since $A_i$ is stable, $z_h$ remains a junction.
    Moreover, no edge crossing $T_i$ is inserted into the forest, and the edge $(z_h, y_h)$ connecting $z_h$ to $T_i $ is not removed from the forest.
    Hence, $z_{h} \in P^\calF_{w_{h}, z_{h+1}}$.
\end{proof}

Now, consider the procedure of $\backSearch$ on $(w_{h},z_{h+1})$, where the level of $(w_{h},z_{h+1})$ equals $t$.
According to \Cref{full:claim:zh-ramain-juncion-on-path}, we have that $z_h$ is a junction and $z_h \in P^\calF_{w_h,z_{h+1}}$.
Moreover, according to the induction for $A_i \in M(t)$, we have that the level of $z_h$ never exceeds $t \leq \lvl_{(w_h,z_{h+1})}$ after the start of round $t$ of the algorithm.
As a result, as long as the level of $z_h$ remains at most $t$, the $\backSearch$ on $(w_{h},z_{h+1})$ has the option to call a recursive $\backSearch$ on $z_h$ (see Line \ref{full:line:condition-back-on-node} of \Cref{full:code:backward-search-edge}).
Note that according to the induction hypothesis for $A_i \in M(t)$, if the level of $z_h$ becomes $t$ after the start of round $t$, then any call to the $\backSearch$ on $z_h$ will be successful.
As a result, the $\backSearch$ on $(w_{h},z_{h+1})$, where the level of $(w_{h},z_{h+1})$ equals $t$ can not fail.
Note that it is very well possible for $\backSearch$ on $(w_{h},z_{h+1})$ to be successful by recursive calls to $\backSearch$ on other junctions, but the main point is that it can not miss $z_h$.

\paragraph{Proof of \Cref{full:level-node-bound} in \Cref{full:claim:levels-are-bounded}.}
According to \Cref{full:obs:levels-round-t}, the level of $z_{h+1}$ is bounded by $t$ at the beginning of round $t+1$.
Now, assume that the level of $z_{h+1}$ becomes $t+1$ and we have a call to $\backSearch$ on $z_{h+1}$.
Denote by $\conf{}{}$ the configuration maintained by the algorithm at the time of this call.
We show that this $\backSearch$ will be successful.

\begin{claim}\label{full:claim:whzh-remain-non-forest}
    Throughout the execution of the algorithm after the start of round $t+1$, we have that:
    \begin{itemize}
        \item $(w_{h},z_{h+1}) \notin E(\calF)$,
        \item $w_{h}$ is inside a reducible region of $\calC$.
    \end{itemize}
\end{claim}

\begin{proof}
    This follows from the definition of $A_i \in M(t)$ being a stable $t$-chain.
    The tail of $A_i$ is $T_i$, whose structure does not change after the beginning of round $t$.
    In other words, no edge crossing $T_i$ is inserted into the forest, especially the edge $(w_{h}, z_{h+1})$, and the reducible region containing $w_{h}$ (inside $T_i$) always remains a reducible region.
\end{proof}

Now, consider the procedure of $\backSearch$ on $z_{h+1}$, where the level of $z_{h+1}$ equals $t+1$.
According to \Cref{full:level-edge-bound} in \Cref{full:claim:levels-are-bounded} (that we just proved above), the level of $(w_{h},z_{h+1})$ is bounded by $t$ at the beginning of the call to $\backSearch$ on $z_{h+1}$.
According to \Cref{full:claim:whzh-remain-non-forest}, as long as the the level of $(w_{h},z_{h+1})$ is at most $t < \lvl_{z_{h+1}}$, the $\backSearch$ on $z_{h+1}$ has the option to call a $\backSearch$ on $(w_{h},z_{h+1})$ (see Line \ref{full:line:conditions-back-on-edge} of \Cref{full:code:backward-search-node}).
Moreover, if the level of $(w_h,z_{h+1})$ becomes $t$, then the $\backSearch$ on $(w_h,z_{h+1})$ must be successful according to \Cref{full:level-edge-bound} in \Cref{full:claim:levels-are-bounded} (that we just proved above).
As a result, the $\backSearch$ on $z_{h+1}$ where $\lvl_{z_{h+1}} = t+1$ can not fail.
Note that it is very well possible for $\backSearch$ on $z_{h+1}$ to be successful by recursive calls to $\backSearch$ on other edges as well, but the main point is that it can not miss $(w_h,z_{h+1})$.

\subsubsection{Proof of \Cref{full:claim:destroying-is-improving}}\label{full:sec:proof-destroying-is-improving}

Assume that $U \in \calS_{t+1}$ such that $A_U$ is qualified and is destroyed by $(u,v)$.
Recall that each $A_U$ can be new, extended, or repetitive according to our construction in \Cref{full:sec:defining-potential-chains}.
Hence, we have the following cases.

\medskip
\noindent
\textbf{$A_U$ Is New.}
In this case, $U$ itself is an entire connected component of $\calF^\init_{t+1}$, and we defined $A_U = (y)$ where $y$ is an arbitrary leaf node of $\calF_{t+1}^\init$ inside  $U$.
At the beginning of round $t+1$, all regions inside $U$ are reducible regions.
Now, consider the edge $(u,v)$ that destroys $A_U$.
Applying a merging or touching edge does not change the edges of the graph.
We conclude that if $(u,v)$ is merging or touching and it destroys $A_U$, both of the endpoints $u$ and $v$ must be inside $U$.
But before applying $(u,v)$, all regions in $U$ were still reducible, and applying $(u,v)$ does not create any touched regions in $U$.
This is in contradiction with $A_U$ being destroyed by $(u,v)$.
Hence, $(u,v)$ must be an improving edge.

\medskip
\noindent
\textbf{$A_U$ Is Extended.}
In this case, we defined $A_U := (w_0, z_1, w_1, \ldots, z_{h}, u^\star, v^\star, y)$ such that 1) $ A_i = (w_0, z_1, w_1, \ldots, z_{h}, w_{h}) \in M(t)$ is a stable $t$-chain (for some $i \in [1, q]$) of length $h \leq t-1$ whose tail is $T_i$, 2) $(u^\star, v^\star)$ is a $T_i$-expanding edge such that $(u^\star,v^\star) \notin E(\calF_{t+1}^\init)$, and 3) $y$ is an arbitrary leaf node of $\calF_{t+1}^\init$ inside $U$.
We consider the following two cases: $v^\star \notin P^\calF_{u,v}$ and $v^\star \in P^\calF_{u,v}$.

\medskip
\noindent
\textbf{Case I: $v^\star \notin P^\calF_{u,v}$.}
Assume that $(u,v)$ is either merging or touching.
Since $v^\star \notin P^\calF_{u,v}$, it is obvious that after applying $(u,v)$, $v^\star$ remains a junction.
Moreover, applying $(u,v)$ does not change the edges of the forest since it is merging or touching.
As a result, the only reason that $(u,v)$ destroys $A_U$ is that $(u,v)$ is a touching edge and applying $(u,v)$ creates a touched region in $U$.
Since $v^\star \notin P^\calF_{u,v}$, we have the two sub-cases below;

\begin{itemize}
\item 
If both $u$ and $v$ are in $U$, it is obvious that $(u,v)$ is merging since $A_U$ is not destroyed before applying $(u,v)$ and all of the regions inside $U$ are reducible.

\item
If both $u$ and $v$ are outside of $U$, then obviously applying $(u,v)$ can not create any touched regions inside $U$.
\end{itemize}
In both sub-cases above, $(u,v)$ can not destroy $A_U$, which is in contradiction with the assumption of \Cref{full:claim:destroying-is-improving}.
Hence, $(u,v)$ can not be merging or touching.

\medskip
\noindent
\textbf{Case II: $v^\star \in P^\calF_{u,v}$.}
In this case, we show that while considering the effective edge $(u,v)$, the call to the main subroutine $\findChain$ on this edge would successfully find an augmenting chain.
This means $(u,v)$ is an improving edge according to the definition.
Recall the procedure of $\findChain$ on $(u,v)$ at some round $r \geq t+1$.
For simplicity, we denote the configuration maintained by the algorithm at the start of this call by $\conf{}{}$.
In $\findChain$, as long as there exists a junction $x \in B \cap P^\calF_{u,v}$ satisfying $\lvl_x \leq r-1$, the algorithm calls a $\backSearch$ on $x$.

According to \Cref{full:claim:levels-are-bounded} for $A_U \in M(t+1)$, we have that the level of $v^\star$ never exceeds $t$ and every call to the $\backSearch$ on $ v^\star$ is successful.
As a result, $\findChain(r,(u,v))$ can never fail since it has the option to call the $\backSearch$ on $x = v^\star$.
This satisfies all the condition $v^\star \in B \cap P^\calF_{u,v}$, $\lvl_{v^\star} \leq t \leq r-1$.

Note that it is very well possible that $\findChain(r,(u,v))$ does not call a $\backSearch$ on $v^\star$ and becomes successful by calling another $\backSearch$ on another node.
But $v^\star$ always remains a valid option and can not be missed.

\medskip
\noindent
\textbf{$A_U$ Is Repetitive.}
In this case, we have that $A_U \in M(t)$.
But according to the definition of $M(t)$, we have that $A_U$ is a stable $t$-chain and can not be destroyed.
This is in contradiction with the assumption that $A_U$ is destroyed by $(u,v)$.

\subsubsection{Proof of \Cref{full:claim:number-of-destroys}}\label{full:sec:proof-number-of-destroys}

Here, we provide a loose analysis since the constant behind the number of destroyed chains (i.e., $O(H)$) is not important.
Many of the cases in this proof actually overlap, and the final number of destroyed chains is smaller.
But we do not need to provide such deliberate analysis.

Assume $(u,v)$ is an arbitrary improving edge and the call to $\findChain$ on $(u,v)$ finds the augmenting chain $A = (w_0,z_1,w_1,\ldots,w_{h_A},z_{h_A+1})$ at some time after the beginning of round $t+1$.
Let $U \in \calS_{t+1}$ be such that $A_U$ is qualified and is destroyed by $(u,v)$.
Recall that each $A_U$ can be new, extended, or repetitive according to our construction in \Cref{full:sec:defining-potential-chains}.
Hence, we have the following cases.

\medskip
\noindent
\textbf{$A_U$ Is New.}
In this case, $U$ itself is an entire connected component of $\calF^\init_{t+1}$, and we defined $A_U = (y)$ where $y$ is an arbitrary leaf node of $\calF^\init_{t+1}$ inside  $U$.
Since at the beginning of round $t+1$, all of the $(\calF_{t+1}^\init, B_{t+1}^\init)$-regions are reducible, it is obvious that if $A_U$ is destroyed by $(u,v)$, then a touched region must emerge in $U$.
According to \Cref{full:lem:apply-chain}, the number of regions that can become touched by applying $(u,v)$, i.e., applying $A$, is bounded by $h_A + 2 \leq H+2$.
As a result, $(u,v)$ can destroy at most $H+2$ many such new chains $A_U$ since the elements $U \in \calS_{t+1}$ are mutually node disjoint and remain separated as long as they are not destroyed.

\medskip
\noindent
\textbf{$A_U$ Is Extended.}
In this case, we defined $A_U := (w_0, z_1, w_1, \ldots, z_{h}, u^\star, v^\star, y)$ such that 1) $ A_i = (w_0, z_1, w_1, \ldots, z_{h}, w_{h}) \in M(t)$ is a stable $t$-chain (for some $i \in [1, q]$) of length $h \leq t-1$ whose tail is $T_i$, 2) $(u^\star, v^\star)$ is a $T_i$-expanding edge such that $(u^\star,v^\star) \notin = E(\calF_{t+1}^\init)$, and 3) $y$ is an arbitrary leaf node of $\calF_{t+1}^\init$ inside $U$.

Since $(u,v)$ is an improving edge, after applying $A$, the set $B$ of junctions is not changed. Hence, $v^\star$ can not be removed from $B$.
As a result, by the assumption that $(u,v)$ destroys $A_U$, at least one of the following must happen;

\begin{itemize}
    \item The edge $(v^\star, y^\star)$ connecting $v^\star$ to $U$ (the last block of $A_U$) is removed from the forest.
    Note that this is a critical edge of $A$ since $v^\star \in B$.
    According to \Cref{full:lem:apply-chain}, there are at most $2h_A+1 \leq 2H+1$ such edges while applying $A$.
    As a result, there are at most $2H+1$ many different $U \in \calS_{t+1}$ such that the edge $(v^\star, y^\star)$ (which is unique for each $U$) is removed because of applying $A$.

    \item An edge crossing $U$ is inserted to the forest.
    It is obvious that this edge is either the last edge $(w_{h_A}, z_{h_A+1})$ of $A$ or a critical edge of $A$.
    Similarly, the number of such edges is bounded by $2H+2$ according to \Cref{full:lem:apply-chain}.
    Each such edge has two endpoints, which concludes there are at most $2 \cdot (2H+2)$ many different $U \in \calS_{t+1}$ that can be destroyed while applying $A$ (since these elements $U \in \calS_{t+1}$ are node-disjoint).

    \item Some touched region inside $U$ has emerged.
    According to \Cref{full:lem:apply-chain}, there are at most $h_A+2 \leq H+2$ many regions that can be removed from $\calC$ because of applying $A$.
    As a result, the number of such destroyed $A_U$ in this case is bounded by $H+2$.
\end{itemize}

\medskip
\noindent
\textbf{$A_U$ Is Repetitive.}
In this case, $A_U \in M(t)$.
But according to our inductive construction of $M(t)$, we have that $A_U$ is a stable $t$-chain and is not destroyed by any effective edge after the beginning of round $t$.
This means that $(u,v)$ can not destroy any repetitive $A_U$.

\medskip
\noindent
\textbf{Final Bound.}
By summing up the number of destroyed $A_U$ for each type of chains as above, we conclude that the total number of pseudo chains $A_U$ that can be destroyed by $(u,v)$ is at most $(H+2) + (2H+1) + 2\cdot(2H+2) + (H+2) \leq 20H$.

\section{Running Time Analysis: Implementation Details}\label{full:sec:time-analysis}

\paragraph{Roadmap.}
We start by describing the main data structures in \Cref{full:sec:data-strucutre}.
Then, we provide the implementation of round $t$ of our algorithm for a fixed $t \in [1, H]$ in \Cref{full:sec:implement-round-t}.
Next, we provide the implementation of $\findChain$ and $\backSearch$ in \Cref{full:sec:implement-main-subroutine}.
Next, we describe the implementation of how to apply an augmenting chain in
\Cref{full:sec:implement-apply-chain}.
Finally, we provide the running time analysis of our implementation in \Cref{full:sec:running-time,full:sec:update-time}.

\subsection{Main Data Structures}\label{full:sec:data-strucutre}

Throughout our algorithm, we use the following data structures.

\begin{lemma}\label{full:lem:forest-query-path}
    Assume $G$ is a fixed graph, $\calF$ is a dynamic forest of $G$, and $U \subseteq V(G)$ is a dynamic set, where each $u \in U$ has a value $\lambda_u$. 
    There exists a dynamic data structure on $(\calF, U, \lambda)$ that supports the following operations, each in $\tilde{O}(1)$ time;
    \begin{enumerate}
        \item Update I: Insert/Delete an edge in $\calF$, while maintaining that $\calF$ remains a forest.
        
        \item Update II: Delete a vertex $u \in U$ from $U$, or insert a vertex $u \in V(G)-U$ with an arbitrary value into $U$.

        \item Query: Given two nodes $u,v \in V(G)$ that belong to the same component of $\calF$, return a node $x$ in $V(P^\calF_{u,v}) \cap U$ with smallest $\lambda_x$ possible ($x$ is arbitrary among all nodes in $V(P^\calF_{u,v}) \cap U$ of smallest value).
    \end{enumerate}
\end{lemma}

\begin{proof}
    This lemma follows directly from the Top Tree data structure \cite{TopTree}.
\end{proof}

\begin{lemma}\label{full:lem:forest-query-edge}
    Assume $G$ is a fixed graph, and $\calF$ is a dynamic forest of $G$.
    There exists a dynamic data structure that supports the following operations, each in $\tilde{O}(1)$ time;
    \begin{enumerate}
        \item Update: Insert/Delete an edge in $\calF$, while maintaining that $\calF$ remains a forest.

        \item Query: Given two nodes $u \in V(G)$ and $v \in V(G)$, return the first edge $(u,y)$ in the path $P^\calF_{u,v}$ (if $u$ and $v$ are connected in $\calF$).
    \end{enumerate}
\end{lemma}

\begin{proof}
    This lemma follows directly from the Top Tree data structure \cite{TopTree}.
\end{proof}

\subsection{Implementing Round $t$}\label{full:sec:implement-round-t}

\subsubsection{Ordering of the Edges}\label{full:sec:order-edges}
We provide an ordering of the edges of the graph in which we scan those edges $(u, v)$ and search for augmenting chains.
This ordering is used in the main body of round $t$.

We consider a first-in-first-out (FIFO) queue $Q$.
The goal is that if an edge becomes effective for the first time, we should immediately add it to $Q$.
We consider edges $(u,v)$ as \emph{ordered edges} while adding them to $Q$, which means that $(u,v)$ can be inserted into $Q$, and $(v,u)$ can be inserted into $Q$ as well in the future.
We show how to update $Q$ when the configuration changes.
Assume an effective edge $(u,v)$ is applied to the configuration.
We have the following cases for the type of $(u,v)$, and explain how to update $Q$ in each case.

\medskip
\noindent
\textbf{Improving Edge.}
If $(u,v)$ is an improving edge, it is straightforward to see that no edge of the graph can become effective after applying an augmenting chain.
So, we do not need to update $Q$.

\medskip
\noindent
\textbf{Merging Edge.}
When $(u,v)$ is applied, a new region $C$ is added into $\calC$ and no touched node in emerged.
Assume an edge $(x, y)$ was not effective and becomes effective after applying $(u,v)$.
This can happen only if at least one of $x$ and $y$ was not part of a reducible region, and now it becomes part of the new reducible region $C$.
Hence, we only need to update $Q$ as follows.
For each $z \in B$ that was removed from the set of junctions $B$ and became part of the new reducible region $C$, i.e., $z \in B \cap V(P^\calF_{u,v})$, we consider all the edges incident on $z$, and add them into $Q$.

\medskip
\noindent
\textbf{Touching Edge.}
After applying $(u,v)$, some of the nodes become touched.
Similar to the previous case, it is straightforward to see that in this case, no new effective edge emerges since all the junctions on $B \cap V(P^\calF_{u,v})$ are now put into $D$.

\begin{observation}\label{full:obs:guarantee-Q}
    Throughout round $t$, if a node becomes effective for the first time, it will surely be inserted into $Q$.
\end{observation}

\subsubsection{Initialization}
The initial configuration $(\calF_t^\init, B_t^\init, \calC_t^\init, D_t^\init)$, initial levels of the nodes (all equal to $1$), and initial levels of the edges (all equal to $0$) can be trivially computed in $\tilde{O}(m)$ time.

\medskip
\noindent
\textbf{Data Structures.}
We initiate an instance of the data structure of \Cref{full:lem:forest-query-path} w.r.t.~the forest $\calF = \calF_t^\init$ and the set $U = B$, where initially, $B = B_t^\init$ and the value $\lambda_x$ of all the nodes $x \in B_t^\init$ is $1$.
This value $\lambda_x$ would be equal to the level of $x$ throughout round $t$.
When the level $\lvl_x^\old$ of a node $x \in B$ changes, we remove $x$ from $U$ in the data structure, and reinsert it into $U$ with the new value $\lambda_x = \lvl_x^\new$.
We denote this data structure throughout our implementation by $\ds^{\text{path}}$.

We maintain a union-find data structure to keep track of the $(\calF, B)$-regions.
Initially, this data structure contains $\{x\}$ as a singleton set for each node $x \in B$, and contains each $C \in \calC_t^\init$ as one set.
Moreover, we maintain a label for each $(\calF,B)$-region to indicate whether it is a reducible region or a touched region.
We denote this data structure by $\ds^{\text{region}}$.

We maintain another simple data structure (for implementing $\backSearch$) as follows.
For each $x \in B$ and $i \in [0, H]$, we keep track of the edges $(x, y)$ satisfying the following; 1) $(x,y) \in E(G)-E(\calF)$, 2) $y$ is inside a reducible region in $\calC$, and 3) $\lvl_{(x,y)} = i$.
It is straightforward to update this data structure when $\calF, \calC$, or the level of a node changes (similar to how we update our queue $Q$ in \Cref{full:sec:order-edges}).
We denote this data structure for $x \in B$ by $\ds^{\text{neighbor}}$.

The last data structure that we maintain is as described in \Cref{full:lem:forest-query-edge}.
We denote this data structure by $\ds^{\text{first-edge}}$.
We need $\ds^{\text{first-edge}}$ in order to implement the process of applying an augmenting chain. More specifically, to find the edge $(z_i,y_i)$ connecting $z_i$ to the $i^\th$ block of the augmenting chain.

\medskip
\noindent
\textbf{Updating Data Structures.}
Throughout round $t$, whenever the forest $\calF_t$, the set of junctions $B_t$, the set of reducible regions $\calC_t$, the set $D_t$, and/or the levels of the nodes and edges have changed, we update all the relevant data structures.

\subsubsection{Main Body of Round $t$}

We start from the configuration $(\calF_t, B_t, \calC_t, D_t) := (\calF^\init_t, B^\init_t, \calC^\init_t, D_t^\init)$ at the beginning of the main body of round $t$, and initialize the FIFO queue $Q$ by inserting all the edges crossing $\calC_t^\init$.
Throughout the main body of round $t$, when the configuration changes, some new edges might be inserted into $Q$ (as explained in \Cref{full:sec:order-edges}).
Now, we proceed by scanning edges on $Q$ as follows.

\medskip
\noindent
\textbf{Scanning Edges.}
We repeatedly pop an edge $(u, v)$ from $Q$ and check if $(u, v)$ is an effective edge w.r.t.~the current configuration.
This can be done in $\tilde{O}(1)$ time by using our data structure $\ds^{\text{region}}$.
If $(u,v)$ is an effective edge w.r.t.~the current configuration $(\calF_t, B_t, \calC_t, D_t)$, we call $\findChain(t,(u,v))$.
We provide the implementation of this subroutine in \Cref{full:sec:implement-main-subroutine}.
We have two cases for the output of this subroutine.
Either it returns $(w_0,z_1,w_1,\ldots,w_{h-1},z_h,u,v)$, or fails.

In the former case, the sequence would be an augmenting chain w.r.t.~the current configuration according to \Cref{full:claim:return-chain}.
We then apply this augmenting chain to the configuration and update it.
We provide the implementation of applying an augmenting chain in \Cref{full:sec:implement-apply-chain}.

In the latter case ($\findChain(t,(u,v))$ fails), the edge $(u,v)$ is either a merging edge or a touching edge. 
We find all junctions $B$ lying on the path $P^{\calF_t}_{u,v}$ by repeatedly querying our data structure $\ds^{\text{path}}$.
Next, we find all regions incident on these junctions by our data structure $\ds^\text{region}$.
If all of these regions are reducible regions, we simply merge all of them, as well as all junctions lying on $ P^{\calF_t}_{u,v}$ into a new reducible region in $\calC_t$ (and update all data structures accordingly).
Otherwise, at least one of these regions is a touched region, which means $(u,v)$ is a touching edge.
In that case, we simply mark all these regions as touched (and update all data structures).

\medskip
\noindent
\textbf{Termination.}
We terminate the main body of round $t$ if there are no more edges left in $Q$.

\subsubsection{Correctness}
We proceed by showing that our implementation for the main body of round $t$ is correct.
We start with the following claim.

\begin{claim}\label{full:claim:remains-effective}
    Throughout the main body of round $t$, if an effective edge $(u,v)$ becomes ineffective, it remains ineffective until the end of round $t$ (recall that we consider $(u,v)$ as an ordered pair, which means $(v,u)$ and $(u,v)$ are considered different edges for the statement of this claim).
\end{claim}

\begin{proof}
    Assume that $(u, v)$ is an edge that becomes effective (or was already effective at the beginning of round $t$), and then it becomes ineffective for the first time during the execution of round $t$.
    According to the definition of an effective edge (\Cref{full:def:effective-edge}), one of the following must have happened.

\begin{itemize}
    \item $u$ becomes touched, i.e., the reducible region $C_u$ containing $u$ becomes touched.
    In this case, $u$ remains touched until the end of round $t$ which concludes that $(u,v)$ can not become effective again.

    \item $v$ becomes part of the reducible region $C_u$ containing $u$.
    In this case, $u$ and $v$ will always remain in the same region for the rest of round $t$.
    As a result, either both of them remain in a reducible region, or both of them become touched nodes at the same time and remain touched nodes.
    Hence, $(u,v)$ can not become effective again.

    \item $(u,v)$ is inserted into the forest.
    This can happen only by applying an augmenting chain $A = (w_0,z_1,\ldots,z_h,w_h,z_{h+1})$.
    There are two cases; (a) $(u,v)$ is not a critical edge of $A$, i.e., $(u,v)$ is one of the edges that are internally changed in some $C_x$ for $x \in \{z_{h+1}, w_0,w_1,\ldots,w_h\}$ by the degree-reduction process (\Cref{full:lem:degree-reduction}) during applying $A$, or (b) $(u,v)$ is a critical edge of $A$.
    Case (a) can not happen since we assumed that $(u,v)$ was effective before applying $A$, which means that $u$ and $v$ can not be in the same reducible region $C_x$.
    Now, assume Case (b) happens.
    According to the definition of an effective edge, neither $u$ nor $v$ is a junction.
    Since all the nodes $z_1,z_2,\ldots,z_h$ are junctions, we conclude that the only option for $(u,v)$ to be a critical edges of $A$ is that $(u,v) = (w_h,z_{h+1})$.
    But, in this case, after applying $A$, $u = w_h$ becomes touched and remains touched throughout the rest of round $t$.
    Hence, $(u, v)$ can not become effective again.
    
    \item $v$ was in a reducible region $C_v$, but now this region becomes touched.
    In this case, since we assume that $(u,v)$ is not an effective edge anymore, we must have that $v$ is inserted into $D_t$, as otherwise, $(u,v)$ still remains effective according to the definition.
    But, if $v \in D_t$, it remains in $D_t$ until the end of round $t$.
    Hence, $(u,v)$ can not become effective again.

    \item $v$ was a touched node where $v \notin D_t$, but now it is inserted into $D_t$.
    Similar to the previous case, $v$ remains in $D_t$ until the end of round $t$, which concludes that $(u,v)$ can not become effective until the end of round $t$.
\end{itemize}
\end{proof}

The correctness of our implementation easily follows from the above claim and the guarantee of $Q$.
Every edge $(u,v)$ that becomes effective during the main body of round $t$ (or is initially effective w.r.t.~$(\calF^\init_t, B^\init_t, \calC^\init_t, D^\init_t)$), will be inserted into $Q$ according to the guarantee of $Q$ in \Cref{full:obs:guarantee-Q}.
Hence, it is sufficient to consider the edge $(u,v)$ (to see if it is effective) at the last time\footnote{We used the term 'last time' since it is possible that the algorithm inserts an edge into $Q$ that has not become effective yet, but once it becomes effective it will be reinserted into $Q$ for sure according to \Cref{full:obs:guarantee-Q}.} that is popped from $Q$.
If it is not effective at this time, \Cref{full:claim:remains-effective} concludes that it can not become effective in the future until the end of round $t$.
As a result, at the end of round $t$ when $Q$ becomes empty, there is no effective edge w.r.t.~$(\calF_t^\final, B_t^\final, \calC_t^\final, D_t^\final)$ that shows the correctness of our implementation.

\subsection{Implementation of Find-Chain}\label{full:sec:implement-main-subroutine}

To implement $\findChain(h,(u,v))$ (see \Cref{full:code:find-chain}), we use our data structures as follows.
Suppose that we are given an effective edge $(u,v)$ w.r.t.~the current configuration $(\calF, B, \calC, D)$.
As long as $\lvl_{(u,v)} \leq h-1$, we repeatedly call $\backSearch$ on $(u,v)$.
It remains to provide the implementation of $\backSearch$.

\subsubsection{Implementation of Backward-Search on an Edge}\label{full:sec:implement-back-search-edge}

Consider a call to $\backSearch$ on $(y,x)$.
We can easily check if $x$ and $y$ are in different connected components of $\calF$ by our data structures.
If that is the case, we just return the sequence $(y,x)$.
Now, assume that $x$ and $y$ are in the same connected components of $\calF$.
We repeatedly query our data structure $\ds^\text{path}$ to find a junction $x^\star \in B$ lying on the path $P^\calF_{u,v}$ with the \textbf{smallest} level.
Then, we check if $\lvl_{x^\star} \leq \lvl_{(y,x)}$.
In this case, we call a $\backSearch$ on $x^\star$.
If this $\backSearch$ returns a sequence, we attach $(y, x)$ to the end of it and return it.
If the $\backSearch$ on $x^\star$ fails, we continue by querying the data structure $\ds^\text{path}$ again and find the next $x^\star_\new$.
Eventually, if the data structure $\ds^\text{path}$ returns a node $x^\star \in B \cap V(P^\calF_{y,x})$ such that $\lvl_{x^\star} = \lvl_{(y,x)} + 1$, we conclude that there is no junction on the path $P^\calF_{y,x}$ of level at most $\lvl_{(y,x)}$.
This means that we can increment the level of $(y,x)$ and terminate the $\backSearch$ with failure.

\subsubsection{Implementation of Backward-Search on an Node}\label{full:sec:implement-back-search-node}

Consider a call to $\backSearch$ on $x$.
We call our data structure $\ds^\text{neighbor}$ in $\tilde{O}(1)$ time to find an edge $(y^\star, x)$ with \textbf{minimum} possible level satisfying; 1) $(y^\star,x) \in E(G) - E(\calF)$ and 2) $y^\star$ is in a reducible region in $\calC$.
Then, we check if $\lvl_{(y^\star,x)} \leq \lvl_x - 1$.
If this is the case, we call $\backSearch$ on $(y^\star, x)$.
If this $\backSearch$ returns a sequence $A$, we return $A$ as well.
If the $\backSearch$ on $(y^\star, x)$ fails, we continue by querying the data structure $\ds^\text{neighbor}$ again and find the next $(y^\star_\new, x)$.
Eventually, if the data structure $\ds^\text{neighbor}$ returns an edge $(y^\star, x)$ such that $\lvl_{(y^\star, x)} = \lvl_{x}$, we conclude that there is no edge $(y, x)$ satisfying; 1) $(y,x) \in E(G) - E(\calF)$ and 2) $y$ is in a reducible region in $\calC$. 
This means that we can increment the level of $x$ and terminate the $\backSearch$ with failure.

\subsection{Implementation of Applying an Augmenting Chain}\label{full:sec:implement-apply-chain}

The process of applying an augmenting chain to $(\calF, B, \calC, D)$ is quite algorithmic, and straightforward to implement.
First, we call \Cref{full:lem:degree-reduction} to reduce the degrees of $\{z_{h+1}, w_0, w_1, \ldots, w_h\}$ to $\leq k-1$ (the implementation of this lemma is identical to the implementation of \cite{BFW26}).
Then, for each $i \in [1, h]$, we query our data structure $\ds^\text{first-edge}$ to find the first edge $(z_i,y_i)$ on the path $P^\calF_{z_i,w_i}$ connecting $z_i$ to the $i^\th$ block of the chain. 
As a result, we have access to the edges $(z_i,y_i)$ and $(w_{i-1}, z_i)$, each in $\tilde{O}(1)$ time.
Finally, we remove $\{(z_i,y_i)\}_{i=1}^h$ from and insert $\{(w_{i-1}, z_i)\}_{i=1}^{h+1}$ into the forest.
Note that after each update on the forest (insertion or deletion of an edge), we update all of our data structures.
We also update our data structure $\ds^\text{region}$ by marking each of the regions $C_{z_{h+1}},C_{w_0},C_{w_1},\ldots,C_{w_h}$ as touched.

\subsection{Running Time Analysis}\label{full:sec:running-time}

\subsubsection{Time Spent on Finding Effective Edges}
We bound the number of times that an edge is inserted into the queue $Q$.
Consider an arbitrary edge $(u, v)$.
This edge can only be inserted into $Q$ at most two times (once as the ordered pair $(u,v)$, and once as $(v,u)$) according to the status of its endpoints.
If $u$ was in $B$ and becomes either part of a reducible component, or a touched node, then our implementation inserts $(u,v)$ into $Q$.
This can only happen once for $u$ throughout round $t$.
The same holds for $v$, which shows that each edge $(u, v)$ can be inserted into $Q$ at most twice.
Finally, for each edge $(u,v)$, we can verify if it is an effective edge in $\tilde{O}(1)$ time.
As a result, the total time spent to find effective edges throughout the entire round $t$ of our implementation is at most $\tilde{O}(m)$.
Over all of the $H$ rounds of the algorithm, this is bounded by $\tilde{O}(H \cdot m)$.

\subsubsection{Time Spent on the Subroutine Find-Chain}

We start with the following claim.

\begin{claim}\label{full:claim:time-bound}
    Consider a call to the subroutine $\findChain$.
    Denote the levels of the nodes and edges at the beginning and at the end of this call by $\lvl$ and $\lvl'$ respectively.
    Let 
    $$ \Delta\lvl := \sum_{(x,y) \in E(G)} \left(\lvl'_{(x,y)} - \lvl_{(x,y)}\right) + \sum_{x \in V(G)} \left(\lvl'_{x} - \lvl_{x}\right). $$
    If the call fails, then its running time is bounded by $\tilde{O}(\Delta \lvl)$, and if the call is successful, then its running time is bounded by $\tilde{O}(H + \Delta \lvl)$.
\end{claim}

\begin{proof}
    According to our implementation of $\findChain$ and $\backSearch$, every time that we want to iterate over some nodes $x$ or edges $(y, x)$, we have access to the next element in this iteration in $\tilde{O}(1)$ time.
    Moreover, every condition that we need to check for a given node $x$ or edge $(y, x)$ in these subroutines also takes $\tilde{O}(1)$ time.
    As a result, the total running time would be proportional (up to a $\tilde{O}(1)$ factor) to the number of elements that we iterate over plus the number of recursive calls to the $\backSearch$.

    \medskip
    \noindent
    \textbf{Running Time of a Failed Call.}
    Every time a call $\backSearch$ on a node or an edge ends up in failure, we increment its level.
    Hence, we can charge the time of internal works of a call to a $\backSearch$ (iterating and checking conditions for elements, excluding the running time of recursive $\backSearch$ calls) to the increased value of their level.
    This shows that the total running time of a failed $\findChain$ is $\tilde{O}(\Delta \lvl)$.

    \medskip
    \noindent
    \textbf{Running Time of a Successful Call.}
    The depth of the recursive calls to $\backSearch$ inside $\findChain$ is at most $H$.
    So, there are at most $H$ many recursive calls to $\backSearch$ that do not end up in failure.
    If $\findChain$ is successful, we consider the $\tilde{O}(H)$ time spent on these $H$ many successful recursive calls explicitly.
    More precisely, we might have $\Delta \lvl = 0$ and this $\tilde{O}(H)$ time on the recursive successful calls is not charged to anything.
    This concludes that for a successful call, the running time is bounded by $\tilde{O}(H + \Delta \lvl)$.
\end{proof}

The above claim concludes that the total time spent on $\findChain$ throughout the entire algorithm is bounded by $ \tilde{O} \left( H \cdot (m + n) \right) + \tilde{O} \left( H \cdot f \right) = \tilde{O} \left( H \cdot m \right)$.
The first term $\tilde{O} \left( H \cdot (m + n) \right)$ comes from the fact that the levels of edges and nodes are monotone and remain in the interval $[0, H]$.
More precisely, the total contribution of $\Delta \lvl$ in \Cref{full:claim:time-bound} for all calls (either successful or failure) is bounded by the total increase in all levels of nodes and edges.
The second term $\tilde{O} \left( H \cdot f \right)$ comes from the fact that there are at most $f-1$ successful calls to $\findChain$ by our algorithm since the initial forest has $f$ components and every successful call to $\findChain$ leads to applying an augmenting chain that reduces the number of components of the forest.

Hence, the total time spent on $\findChain$ throughout the entire algorithm is $\tilde{O}(H\cdot m)$.

\subsubsection{Time Spent on Applying Merging/Touching Edges}

Assume $(u,v)$ is a merging or a touching edge.
According to the implementation, we can find all the nodes $x \in B \cap P^\calF_{u,v}$ in $\tilde{O}(|B \cap P^\calF_{u,v}|)$ time.
Then, for each $x \in B \cap P^\calF_{u,v}$, we consider all the neighbors $(x, y) \in \calF$, and find the $(\calF, B)$-region to which $y$ belongs.
As a result, we have access to $\bound_{(\calF, B)}(P^\calF_{u,v})$ in $\tilde{O}(|\bound_{(\calF, B)}(P^\calF_{u,v})|)$ time.

Eventually, we combine all these regions together with all junctions lying on $P^\calF_{u,v}$ into a new single region.
This new region is either a reducible region (if $(u,v)$ is a merging edge), or is a touched region (if $(u,v)$ is a touching edge).
Hence, the total time that we spend is $\tilde{O}(|\bound_{(\calF, B)}(P^\calF_{u,v})|)$.
Note that $|\bound_{(\calF, B)}(P^\calF_{u,v})|$ is \textbf{not} proportional to the size of the regions inside $\bound_{(\calF, B)}(P^\calF_{u,v})$.
In fact, it is proportional to the number of regions in this boundary.

For any $t \in [1, H]$, throughout the entire round $t$, the set of regions is monotone, i.e., some regions can be merged together, but they will never be split.
As a result, the summation of all these terms $\tilde{O}(|\bound_{(\calF, B)}(P^\calF_{u,v})|)$ for all merging/touching edges that are applied throughout the entire round $t$ is bounded by $\tilde{O}(n)$.
Over all of the $H$ rounds of the algorithm, the total time is bounded by $\tilde{O}(H \cdot n)$.

\subsubsection{Time Spent on Applying Improving Edges}\label{full:sec:time-apply-chains}

Assume $(u,v)$ is an improving edge.
According to our implementation, the algorithm applies an augmenting chain $(w_0,z_1,\ldots,z_h,w_h,z_{h+1})$ where $h \leq H$.

The degree-reduction process of \Cref{full:lem:degree-reduction} on each node $x \in \{z_{h+1},w_0,w_1,\ldots,w_h\}$ runs in $\tilde{O}(\sum_{u \in C_x} \deg_G(u))$ time.
Then, the algorithm finds the critical edges $(w_{i-1}, z_i)$ and $(z_i, y_i)$ to insert into and remove from the forest $\calF$.

\medskip
\noindent
\textbf{Time Spent on Degree-Reduction Process in \Cref{full:lem:degree-reduction}.}
Fix a round $t$.
Whenever a call to the degree-reduction process on a node $x$ is made, the region $C_x$ containing $x$ is a reducible region.
After $C_x$ becomes a touched region, all of the nodes inside $C_x$ remain touched nodes until the end of round $t$ of the algorithm.
As a result, if we consider all nodes $x$ whose degrees have been reduced by a call to the degree-reduction process in \Cref{full:lem:degree-reduction}, and sum up all the terms $\tilde{O}(\sum_{u \in C_x} \deg_G(u))$ for the running time, we conclude that throughout the entire round $t$, the total time spent on degree-reduction process among all augmenting chains that have been applied is at most $\tilde{O}(m)$.
Over all of the $H$ rounds of the algorithm, this is bounded by $\tilde{O}(H \cdot m)$.

\medskip
\noindent
\textbf{Time Spent on Inserting and Removing Critical Edges.}
According to the implementation, each of the critical edges can be found in $\tilde{O}(1)$ time using our $\ds^{\text{first-edge}}$ data structure.
As a result, since the length of the augmenting chain that has been applied is bounded by $H$, the total time spent on taking care of critical edges is at most $\tilde{O}(H)$ for each augmenting chain.
Finally, the total time spent by the algorithm on handling critical edges throughout applying all of the augmenting chains is bounded by $\tilde{O}(H \cdot f) = \tilde{O}(H \cdot n)$ since there are at most $f-1$ many augmenting chains applied by the algorithm.

\subsection{Time Spent on Updating Data Structures}\label{full:sec:update-time}

Throughout any round $t$, whenever the forest $\calF_t$, the set of junctions $B_t$, the set of reducible regions $\calC_t$, the set $D_t$, or the levels of the nodes and edges have changed, we update all the relevant data structures.

\subsubsection{Updating $\ds^{\text{path}}$}

\textbf{Changing $\calF_t$.}
According to the discussion in \Cref{full:sec:time-apply-chains}, the total number of edges that we insert or remove in the forest throughout the entire round $t$ is bounded by $\tilde{O}(m)$.
As a result, the number of times that we need to update $\ds^{\text{path}}$ because the forest has changed is $\tilde{O}(m)$ in each round.
Over all of the $H$ rounds of the algorithm, this is bounded by $\tilde{O}(H \cdot m)$.

\medskip
\noindent
\textbf{Changing $B_t$.}
The set $B_t$ is decreasing throughout round $t$, which means that at most $n$ nodes from $B_t$ can be removed in round $t$ and nothing will be inserted into $B_t$.
Hence, the number of times that we need to update $\ds^{\text{path}}$ because a node is removed from $B_t$ is at most $O(n)$.
Over all of the $H$ rounds of the algorithm, this is bounded by $\tilde{O}(H \cdot n)$.

\medskip
\noindent
\textbf{Changing $\calC_t$ or $D_t$.}
Updates on $\calC_t$ and $D_t$ does not update $\ds^{\text{path}}$.

\medskip
\noindent
\textbf{Changing the Levels.}
Finally, every time the level of some node in $x \in B_t$ is increased, we should update $\ds^{\text{path}}$ by deleting the corresponding node from $U = B_t$ and reinserting it with the new value $\lambda_x = \lvl_x^\new$.
Since the levels of the nodes are in the bounded interval $[1, H+1]$, we conclude that the total number of times that $\ds^{\text{path}}$ is updated because the level of a node is changed is bounded by $O(H \cdot n) $ throughout the entire algorithm.

\medskip
\noindent
\textbf{Total Time.}
As a result, the total number of times that we update $\ds^{\text{path}}$ is bounded by $\tilde{O}(H \cdot m)$ and each update takes $\tilde{O}(1)$ time to perform.
Hence, we only spend $\tilde{O}(H \cdot m)$ time to maintain $\ds^{\text{path}}$.

\subsubsection{Updating $\ds^{\text{region}}$}

Changes in the level of the nodes do not update $\ds^{\text{region}}$.
Hence, we only need to analyze the time spent on updating $\ds^{\text{region}}$ after applying an effective edge.
Assume that an effective edge $(u,v)$ is being applied.

\medskip
\noindent
\textbf{$(u,v)$ is Merging or Touching.}
If a merging or touching edge $(u,v)$ is applied, according to the implementation of applying $(u,v)$, we should merge all the regions that hit or are incident to $\bound_{(\calF,B)}(P^\calF_{u,v})$ into a single new region.
Since the set of regions is monotone, the total number of updates on $\ds^{\text{region}}$ because of applying a merging or touching edge is at most $O(n)$.
Over all of the $H$ round of the algorithm, this is bounded by $O(H \cdot n)$.

\medskip
\noindent
\textbf{$(u,v)$ is Improving.}
If $(u,v)$ is an improving edge, assume that $(w_0,z_1,\ldots,z_h,w_h,z_{h+1})$ is the corresponding augmenting chain that is applied.
The edges of the forest that are changed internally in each reducible region $C_x$ for $x \in \{z_{h+1}, w_0,w_1,\ldots,w_h\}$ because of the degree-reduction process (\Cref{full:lem:degree-reduction}) do not affect the set of regions since $C_x$ remains a region.
Moreover, all the critical edges of the augmenting chain (that are changed in the forest) are incident on a node in $B$ (all $z_1,z_2,\ldots,z_h$) except the very last edge $(w_h,z_{h+1})$.
As a result, the only change that happens to the set of regions by applying this augmenting chain is that the regions containing $w_h$ and $z_{h+1}$ are now merged to a single region.
Moreover, the regions containing $z_{h+1},w_0,w_1,\ldots,w_h$ become touched regions.
As a result, the total number of times that we update $\ds^{\text{region}}$ after applying an augmenting chain is at most $O(H)$ since the length of the chain is at most $H$.
Over all of the augmenting chains that are applied by the algorithm (that is at most $f-1$), the number of updates on $\ds^{\text{region}}$  is bounded by $O(H \cdot f)$.

\medskip
\noindent
\textbf{Total Time.}
As a result, the total time spent to update $\ds^{\text{region}}$ throughout the entire round $t$ is at most $\tilde{O}(m)$.

\subsubsection{Updating $\ds^{\text{neighbor}}$}

\textbf{Changing the Levels.}
When the level of an edge changes, the data structure $\ds^{\text{neighbor}}$ is updated in $\tilde{O}(1)$ time.
Since the levels of the edges are only increasing and are in the bounded interval $[0, H]$, we conclude that the total time spent updating $\ds^{\text{neighbor}}$ because of changes in the level of the edges is bounded by $\tilde{O}(H \cdot m) $.

\medskip
\noindent
\textbf{Applying an Improving Edge (an Augmenting Chain).}
Assume that an augmenting chain $A = (w_0,z_1,\ldots,z_h,w_h,z_{h+1})$ is applied.
The edges that are changed internally in the reducible regions $C_x$ for $x \in \{z_{h+1},w_0,w_1,\ldots,w_h\}$ because of the degree-reduction process of $x$ (\Cref{full:lem:degree-reduction}) do not update $\ds^{\text{neighbor}}$.
The set $B_t$ does not change as well by applying the augmenting chain, and the only updates in $\ds^{\text{neighbor}}$ are because of changing critical edges.
The number of such updates is bounded by $O(H \cdot f) = \tilde{O}(H \cdot n)$ since there are at most $f-1$ many augmenting chains that are applied and each of them has at most $O(H)$ critical edges.

\medskip
\noindent
\textbf{Applying a Merging Edge.}
When a merging edge $(u,v)$ is applied, we have to update $\ds^{\text{neighbor}}$ by considering the edges that are incident on every $ x \in B_t \cap P^{\calF_t}_{u,v}$ (there might be an edge $(a,b)$ such that $a \in B$ and $b=x$, where after applying $(u,v)$, $b=x$ becomes part of a reducible region and we need to add $(a,b)$ to the list of $a$ in $\ds^{\text{neighbor}}$).
These are the only edges that might need updates (which is similar to updating the queue $Q$ that maintains the ordering of the edges).
As a result, the data structure can be updated in $\tilde{O}(\sum_{x \in B_t \cap P^{\calF_t}_{u,v}} \deg_{G}(x))$.
Since the set $B_t$ is decreasing throughout the entire round $t$, we conclude that the summation of all such terms for all merging edges that have been applied throughout round $t$ is bounded by $\tilde{O} (\sum_{x \in V} \deg_G(x)) = \tilde{O}(m)$.
As a result, the total time spent for updating data structure $\ds^{\text{neighbor}}$ because of applying merging edges is bounded by $\tilde{O}(m)$.
Over all of, the $H$ round of the algorithm is bounded by $\tilde{O}(H \cdot m)$.

\medskip
\noindent
\textbf{Applying a Touching Edge.}
When a touching edge $(u,v)$ is applied, we need to update $\ds^{\text{neighbor}}$ by considering two types of edges: 
\begin{enumerate}
    \item edges that are incident on every $x \in B_t \cap P^{\calF_t}_{u,v}$ (there might be an edge $(a,b)$ where $a=x$ and $b$ is in a reducible region, but after applying $(u,v)$, $a=x$ becomes touched and then we do not need any list in $\ds^{\text{neighbor}}$ for $a$ anymore).

    \item edges with one endpoint in a reducible region in $\bound_{(\calF_t,B_t)}(P^{\calF_t}_{u,v})$ (there might be an edge $(a,b)$, where $a \in B$ and $b$ is one reducible region in $\bound_{(\calF_t,B_t)}(P^{\calF_t}_{u,v})$, but after applying $(u,v)$ this region becomes touched and we need to remove $(a,b)$ from the list of $a$ in $\ds^{\text{neighbor}}$).
\end{enumerate}
The total time spent for updating type 1 edges throughout the entire round $t$ is $\tilde{O}(m)$ by a similar argument as before. 
For type 2 edges, if a reducible region $C$, becomes touched, the time of updating edges incident on $C$ is $\tilde{O}(\sum_{x \in C} \deg_G(x))$, where all $x \in C$ become a touched node for the first time.
Since every node $x \in V$ can become touched at most once in each round, we conclude that the summation of all such terms throughout the entire round $t$ is at most $\tilde{O}(\sum_{x \in V} \deg_G(x)) = \tilde{O}(m)$.
Over all of the $H$ round of the algorithm, this is bounded by $\tilde{O}(H \cdot m)$.

\subsubsection{Updating $\ds^{\text{first-edge}}$}
The only reason that causes an update on $\ds^{\text{first-edge}}$ is changing the forest.
According to the discussion in \Cref{full:sec:time-apply-chains}, the total number of edges that we insert or remove in the forest throughout the entire round $t$ is bounded by $\tilde{O}(m)$.
As a result, the number of times that we need to update $\ds^{\text{first-edge}}$ is $\tilde{O}(m)$ in each round.
Over all of the $H$ rounds of the algorithm, this is bounded by $\tilde{O}(H \cdot m)$.
As a result, the time spent to update $\ds^{\text{first-edge}}$ is bounded by $\tilde{O}(m)$ as well since each update is handled in $\tilde{O}(1)$ time.

\section*{Acknowledgements}

Sayan Bhattacharya is funded by the European Union (ERC grant, DYNALP, 101170133). Views and opinions expressed are however those of the author(s) only and do not necessarily reflect those of the European Union or the European Research Council Executive Agency. Neither the European Union nor the granting authority can be held responsible for them.

\newpage
\bibliographystyle{alpha}
\bibliography{references,ref}

@article{furer1994approximating,
	title={Approximating the minimum-degree Steiner tree to within one of optimal},
	author={Furer, Martin and Raghavachari, Balaji},
	journal={Journal of Algorithms},
	volume={17},
	number={3},
	pages={409--423},
	year={1994},
	publisher={Elsevier}
}

@article{kiraly2012degree,
	title={Degree bounded matroids and submodular flows},
	author={Kir{\'a}ly, Tam{\'a}s and Lau, Lap Chi and Singh, Mohit},
	journal={Combinatorica},
	volume={32},
	number={6},
	pages={703--720},
	year={2012},
	publisher={Springer}
}

@article{lau2013additive,
	title={Additive approximation for bounded degree survivable network design},
	author={Lau, Lap Chi and Singh, Mohit},
	journal={SIAM Journal on Computing},
	volume={42},
	number={6},
	pages={2217--2242},
	year={2013},
	publisher={SIAM}
}

@article{fukunaga2015iterative,
	title={Iterative rounding approximation algorithms for degree-bounded node-connectivity network design},
	author={Fukunaga, Takuro and Nutov, Zeev and Ravi, R},
	journal={SIAM Journal on Computing},
	volume={44},
	number={5},
	pages={1202--1229},
	year={2015},
	publisher={SIAM}
}

@inproceedings{bansal2008additive,
	title={Additive guarantees for degree bounded directed network design},
	author={Bansal, Nikhil and Khandekar, Rohit and Nagarajan, Viswanath},
	booktitle={Proceedings of the fortieth annual ACM symposium on Theory of computing},
	pages={769--778},
	year={2008}
}

@inproceedings{furer1990nc,
	title={An NC approximation algorithm for minimum degree spanning tree problem},
	author={Furer, M},
	booktitle={Proc. of the 28th Annual Allerton Conf. on Communication, Control and Computing, 1990},
	year={1990}
}

@article{singh2015approximating,
	title={Approximating minimum bounded degree spanning trees to within one of optimal},
	author={Singh, Mohit and Lau, Lap Chi},
	journal={Journal of the ACM (JACM)},
	volume={62},
	number={1},
	pages={1--19},
	year={2015},
	publisher={ACM New York, NY, USA}
}

@article{TopTree,
  author       = {Stephen Alstrup and
                  Jacob Holm and
                  Kristian de Lichtenberg and
                  Mikkel Thorup},
  title        = {Maintaining information in fully dynamic trees with top trees},
  journal      = {{ACM} Trans. Algorithms},
  volume       = {1},
  number       = {2},
  pages        = {243--264},
  year         = {2005},
  url          = {https://doi.org/10.1145/1103963.1103966},
  doi          = {10.1145/1103963.1103966},
  bibsource    = {dblp computer science bibliography, https://dblp.org}
}

@article{DuanP20,
  author       = {Ran Duan and
                  Seth Pettie},
  title        = {Connectivity Oracles for Graphs Subject to Vertex Failures},
  journal      = {{SIAM} J. Comput.},
  volume       = {49},
  number       = {6},
  pages        = {1363--1396},
  year         = {2020}
}

@article{S24,
  author       = {Thatchaphol Saranurak},
  title        = {Open Problem:  Low-Degree Spanning Tree},
  journal      = {Dagstuhl Seminar 24471:
Graph Algorithms: Distributed Meets Dynamic},
  year         = {2024}
}

@article{P16,
  author       = {Seth Pettie},
  title        = {Open Problem 24: Hardness of Approximating NP-hard Problems},
  journal      = {Dagstuhl Seminar 16451:
Structure and Hardness in P},
  year         = {2016}
}

@book{LauRS11,
  title={Iterative Methods in Combinatorial Optimization},
  author={Lau, Lap Chi and Ravi, R. and Mohit Singh},
  year={2011},
  publisher={Cambridge university press}
}

@inproceedings{FR92,
  author       = {Martin F{\"{u}}rer and
                  Balaji Raghavachari},
  editor       = {Greg N. Frederickson},
  title        = {Approximating the Minimum Degree Spanning Tree to Within One from
                  the Optimal Degree},
  booktitle    = {Proceedings of the Third Annual {ACM/SIGACT-SIAM} Symposium on Discrete
                  Algorithms, 27-29 January 1992, Orlando, Florida, {USA}},
  pages        = {317--324},
  publisher    = {{ACM/SIAM}},
  year         = {1992},
  url          = {http://dl.acm.org/citation.cfm?id=139404.139469},
  bibsource    = {dblp computer science bibliography, https://dblp.org}
}

@inproceedings{DHZ20,
  author       = {Ran Duan and
                  Haoqing He and
                  Tianyi Zhang},
  editor       = {Yoshiharu Kohayakawa and
                  Fl{\'{a}}vio Keidi Miyazawa},
  title        = {Near-Linear Time Algorithm for Approximate Minimum Degree Spanning
                  Trees},
  booktitle    = {{LATIN} 2020: Theoretical Informatics - 14th Latin American Symposium,
                  S{\~{a}}o Paulo, Brazil, January 5-8, 2021, Proceedings},
  series       = {Lecture Notes in Computer Science},
  volume       = {12118},
  pages        = {15--26},
  publisher    = {Springer},
  year         = {2020},
  url          = {https://doi.org/10.1007/978-3-030-61792-9\_2},
  doi          = {10.1007/978-3-030-61792-9\_2},
  bibsource    = {dblp computer science bibliography, https://dblp.org}
}

@inproceedings{CQT21,
  author       = {Chandra Chekuri and
                  Kent Quanrud and
                  Manuel R. Torres},
  editor       = {Mary Wootters and
                  Laura Sanit{\`{a}}},
  title        = {Fast Approximation Algorithms for Bounded Degree and Crossing Spanning
                  Tree Problems},
  booktitle    = {Approximation, Randomization, and Combinatorial Optimization. Algorithms
                  and Techniques, {APPROX/RANDOM} 2021, August 16-18, 2021, University
                  of Washington, Seattle, Washington, {USA} (Virtual Conference)},
  series       = {LIPIcs},
  volume       = {207},
  pages        = {24:1--24:21},
  publisher    = {Schloss Dagstuhl - Leibniz-Zentrum f{\"{u}}r Informatik},
  year         = {2021},
  url          = {https://doi.org/10.4230/LIPIcs.APPROX/RANDOM.2021.24},
  doi          = {10.4230/LIPICS.APPROX/RANDOM.2021.24},
  bibsource    = {dblp computer science bibliography, https://dblp.org}
}

@inproceedings{RMRRH93,
  author       = {R. Ravi and
                  Madhav V. Marathe and
                  S. S. Ravi and
                  Daniel J. Rosenkrantz and
                  Harry B. Hunt III},
  editor       = {S. Rao Kosaraju and
                  David S. Johnson and
                  Alok Aggarwal},
  title        = {Many birds with one stone: multi-objective approximation algorithms},
  booktitle    = {Proceedings of the Twenty-Fifth Annual {ACM} Symposium on Theory of
                  Computing, May 16-18, 1993, San Diego, CA, {USA}},
  pages        = {438--447},
  publisher    = {{ACM}},
  year         = {1993},
  url          = {https://doi.org/10.1145/167088.167209},
  doi          = {10.1145/167088.167209},
  bibsource    = {dblp computer science bibliography, https://dblp.org}
}

@inproceedings{KR00,
  author       = {Jochen K{\"{o}}nemann and
                  R. Ravi},
  editor       = {F. Frances Yao and
                  Eugene M. Luks},
  title        = {A matter of degree: improved approximation algorithms for degree-bounded
                  minimum spanning trees},
  booktitle    = {Proceedings of the Thirty-Second Annual {ACM} Symposium on Theory
                  of Computing, May 21-23, 2000, Portland, OR, {USA}},
  pages        = {537--546},
  publisher    = {{ACM}},
  year         = {2000},
  url          = {https://doi.org/10.1145/335305.335371},
  doi          = {10.1145/335305.335371},
  bibsource    = {dblp computer science bibliography, https://dblp.org}
}

@inproceedings{CRRT05,
  author       = {Kamalika Chaudhuri and
                  Satish Rao and
                  Samantha J. Riesenfeld and
                  Kunal Talwar},
  editor       = {Chandra Chekuri and
                  Klaus Jansen and
                  Jos{\'{e}} D. P. Rolim and
                  Luca Trevisan},
  title        = {What Would Edmonds Do? Augmenting Paths and Witnesses for Degree-Bounded
                  MSTs},
  booktitle    = {Approximation, Randomization and Combinatorial Optimization, Algorithms
                  and Techniques, 8th International Workshop on Approximation Algorithms
                  for Combinatorial Optimization Problems, {APPROX} 2005 and 9th InternationalWorkshop
                  on Randomization and Computation, {RANDOM} 2005, Berkeley, CA, USA,
                  August 22-24, 2005, Proceedings},
  series       = {Lecture Notes in Computer Science},
  volume       = {3624},
  pages        = {26--39},
  publisher    = {Springer},
  year         = {2005},
  url          = {https://doi.org/10.1007/11538462\_3},
  doi          = {10.1007/11538462\_3},
  bibsource    = {dblp computer science bibliography, https://dblp.org}
}

@inproceedings{Goemans06,
  author       = {Michel X. Goemans},
  title        = {Minimum Bounded Degree Spanning Trees},
  booktitle    = {47th Annual {IEEE} Symposium on Foundations of Computer Science, {FOCS}
                  2006, Berkeley, California, USA, October 21-24, 2006, Proceedings},
  pages        = {273--282},
  publisher    = {{IEEE} Computer Society},
  year         = {2006},
  url          = {https://doi.org/10.1109/FOCS.2006.48},
  doi          = {10.1109/FOCS.2006.48},
  bibsource    = {dblp computer science bibliography, https://dblp.org}
}

@misc{BFW26,
      title={Additive One Approximation for Minimum Degree Spanning Tree: Breaking the $O(mn)$ Time Barrier}, 
      author={Sayan Bhattacharya and Ermiya Farokhnejad and Haoze Wang},
      year={2026},
      eprint={2602.23448},
      archivePrefix={arXiv},
      primaryClass={cs.DS},
      url={https://arxiv.org/abs/2602.23448}, 
}

\appendix

\section{Getting Rid of the Prior Knowledge of $\Delta^\star$}\label{appendix:unknown-Delta}

In this section, we show how to run our algorithm without prior knowledge of $\Delta^\star$.
The idea is simple: binary search.
We consider a variable $b \in [1, n]$, to do the binary search.

\subsection{Multiplicative Approximation.}
The desired maximum degree of the spanning tree is $k = \lceil (1+\varepsilon)\cdot \Delta^\star \rceil + 1$.
It is straightforward to see that in our algorithm and its analysis, we never use the fact that $k$ is exactly equal to this value.
In other words, as long as $k \geq \lceil (1+\varepsilon)\cdot \Delta^\star \rceil + 1$, since there exists an spanning tree of maximum degree bounded by $(k-2)/(1+\varepsilon)$, the analysis of our algorithm can be synchronized with the spanning tree of bounded degree $(k-2)/(1+\varepsilon)$ instead of the actual optimal spanning tree of maximum degree $\Delta^\star$.

\medskip
\noindent
\textbf{The Algorithm.}
Hence, we can do the following.
We run our algorithm for finding a spanning tree of bounded degree $k = b$.
Moreover, we consider a timer that stops after $\Tilde{O}(m/\varepsilon^2)$ time steps (which is the running time of our algorithm given that $\Delta^\star$ is known).
Hence, either the algorithm successfully finds a spanning tree of bounded degree $b$, or the timer stops.
In the former case, we decrease the value of $b$ according to the binary search step.
In the latter case, we increase the value of $b$ according to the binary search step.

\medskip
\noindent
\textbf{Correctness.}
As long as $b \geq \lceil (1+\varepsilon)\cdot \Delta^\star \rceil + 1$, the algorithm must be successful.
Only if $b < \lceil (1+\varepsilon)\cdot \Delta^\star \rceil + 1$, there is a non-zero chance that the algorithm fails.
This means that the final value $b^\star$ in the binary search scheme that the algorithm is successful, must satisfy $b^\star \leq \lceil (1+\varepsilon)\cdot \Delta^\star \rceil + 1$.
This process only increases the running time by an $O(\log n)$ factor.
We can also reduce it to $O(\log \Delta^\star)$ by performing a classical binary search starting from $b=1$, and doubling the increase/decrease step, until $b$ becomes larger than $b^\star$ (the value of $b$ remains bounded by $2b^\star = O(\Delta^\star)$).
For the rest of the binary search, the increase/decrease step keeps getting halved.

\subsection{Additive Approximation.}
This is completely similar to the multiplicative approximation algorithm, except that we set the timer to $\Tilde{O}(mn^{2/3})$ instead of $\Tilde{O}(m/\varepsilon^2)$.
Then, as long as $b \geq \Delta^\star+1$, the algorithm must be successful before the timer stops.

\section{Barriers of the Previous Works and Our Improvements}\label{appendix-previous-work-barriers}

\cite{BFW26} constructed a precise instance that shows the algorithm of \cite{DHZ20} can return a spanning tree whose maximum degree is larger than $\Delta^\star$ by a multiplicative factor of $\Omega(\log n/ \log \log n)$.
This means that although \cite{DHZ20} algorithm runs in near-linear time, it is incapable of returning a constant approximate MDST for general inputs.

For the rest of this section, we explain the algorithms of \cite{FR92} and \cite{BFW26}, show their barriers, and describe our techniques to improve them.

\subsection{The Algorithm of \cite{FR92}}\label{sec:O(mn)-time}

\cite{BFW26} has a nice perspective of the \cite{FR92} algorithm that we describe here.
Let $T$ be a connected sub-tree of a feasible forest $\calF$.
Consider the following procedure, which is essentially the algorithm of \cite{FR92} running on $T$.
Initially, we define the set $B$ of \textit{saturated nodes} as all $u \in V(T)$ satisfying $\deg_{\calF}(u) = k$, and mark all $(T, B)$-regions as \textit{slack regions}.
Then, the following two main steps are executed iteratively;
\begin{itemize}
    \item[] \textbf{Discovering Step.} Find an edge $(x,y) \in E(G) - E(\calF)$ between two distinct slack regions. 
    \item[] \textbf{Merging Step.} Consider the path $P_{x,y}^T$, and {\em merge} all of the slack regions in $\bound_{(T, B)}(P_{x,y}^T)$ together with all saturated nodes on $P_{x,y}^T$, to form a new larger slack region.
\end{itemize}
Throughout the whole process, each slack region remains a sub-tree of $\calF$, and different slack regions remain mutually node-disjoint, separated by remaining saturated nodes.
At the end, we refer to every node inside a slack region as a \textbf{reducible} node w.r.t.~$T$, and refer to every remaining saturated node as a \textbf{non-reducible} node w.r.t.~$T$.
Morally, a node $u$ is reducible w.r.t.~$T$ if we can locally change $\calF$, by inserting and removing some edges inside $T$ (in a specific way), to achieve another feasible forest $\calF'$ satisfying $\deg_{\calF'}(u) \leq k-1$.
This key property is summarized in \Cref{lem:degree-reduction}.

\begin{lemma}[Degree Reduction]\label{lem:degree-reduction}
    There is a {\bf degree-reduction} subroutine, that given $(\calF, C, u)$, where $\calF$ is a feasible forest, $C$ is a sub-tree of $\calF$, and $u \in C$ is reducible w.r.t.~$C$, returns another feasible forest $\calF^+$ satisfying $\deg_{\calF^+}(u) \leq k-1$ by inserting/deleting some edges $e \in E(C)$ internally in $C$ in a total of $\tilde{O}\left( \sum_{x \in V(C)} \deg_G(c) \right)$ time.
\end{lemma}

\noindent
\textbf{The Algorithm of \cite{FR92}.}
The algorithm runs in $n-1$ rounds.
At the start of round $t \in [1, n-1]$, $\calF^t$ is a feasible forest with $n-t+1$ connected components (initially, $\calF^1$ is an empty forest, i.e., $E(\calF^1) = \emptyset$).
The goal is to achieve $\calF^{t+1}$ by modifying $\calF^t$ with one fewer connected component.
Eventually, $\calF^n$ would be a spanning tree of $G$ with desired maximum degree.

\medskip
\noindent
\textbf{Round $t$.}
We find all the reducible and non-reducible nodes w.r.t.~each connected component of $\calF^t$.
With a simple counting argument (as in \cite{BFW26}), it is possible to show that there always exists an edge $(u,v) \in E(G)$ between two reducible nodes $u$ and $v$ in to two different connected components of $\calF^t$.
We find such an edge $(u,v)$, run the degree reduction subroutine (\Cref{lem:degree-reduction}) on $u$ and $v$ separately, and add $(u, v)$ to the forest, obtaining $\calF^{t+1}$.

\subsection{The Algorithm of \cite{BFW26}}\label{sec:BFW26}

Recently, \cite{BFW26} improved the \cite{FR92} algorithm and provided a plus one additive approximate MDST algorithm in $\tilde{O}\!\left(mn^{3/4}\right)$ time.
The algorithm starts with a feasible forest $\calF$ with $f$ components, and considers a decomposition of the forest by chopping up each component of $\calF$ into small pieces of size $O(n/f)$.
For each piece $M$ (called a molecule), the \cite{FR92} algorithm is executed to find the reducible and non-reducible vertices w.r.t.~$M$.
They introduced the notion of an \textit{augmenting chain} of length $h \geq 0$, which is used to reduce the number of components of the forest by one, while destroying only $O(h)$ many molecules.

\medskip
\noindent
\textbf{Main Subroutine.}
\cite{BFW26} shows how to find a blocking set of independent augmenting chains (exploiting different molecules), each of length $O(n/f)$, in $O(mn/f)$ time.
With a deliberate counting argument, they showed that the size of this blocking set is $\Omega(f^3/n^2)$.
So, it is possible to reduce the number of components of the forest from $f$ to $f - \Omega(f^3/n^2)$ in $\tilde{O}(mn/f)$ time. 

\medskip
\noindent
\textbf{The Algorithm.}
As long as $f \geq n^{3/4}$, \cite{BFW26} iteratively runs the main subroutine.
Eventually, when $f \leq n^{3/4}$, they run \cite{FR92} algorithm iteratively to reduce the number of components of the forest one by one, as described in \Cref{sec:O(mn)-time}.
This leads to their result.

\subsubsection{Barriers of \cite{BFW26} Algorithm}\label{sec:BFW-barriers}

When the number of connected components of the forest becomes small, the \cite{BFW26} algorithm does not have any advantage over \cite{FR92}.
For instance, when $f = O(\sqrt{n})$, the size of the molecules considered by the algorithm becomes too large, $\Omega(n/f) = \Omega(\sqrt{n})$.
Hence, the procedure of \cite{BFW26} can destroy a molecule consisting of a large number of nodes, and the algorithm finds only one augmenting chain.
This issue remains \textit{even if we relax the approximation ratio to multiplicative} instead of additive.
In the next section (\ref{appendix:hard-instance}), we provide a precise instance that shows the \cite{BFW26} algorithm needs at least $\Omega(mn^{1/4})$ time to build a spanning tree of maximum degree $2\Delta^\star$.
This is not necessarily the hardest instance for the algorithm, but it shows that the running time of their algorithm is far from optimal even for $O(\Delta^\star)$ approximation. 

\subsection{A Hard Instance for the \cite{BFW26} Algorithm}\label{appendix:hard-instance}

The algorithm of \cite{BFW26} defines the parameter $\theta = 20n/f$ and if a connected component of the forest has less than $2\theta$ nodes, it considers this component as a singleton molecule.\footnote{Even if we change the constant $20$ in the algorithm, we can still construct a hard instance by following the same rules.}
Now, assume $n = (a^2 + 5a + 6)/2$ for some $a \in \mathbb{N}$.
The forest consists of $a+1$ connected components.
There is a connected component which is a star graph with center $y$ and $2a+2$ edges $(y, z_i)$ for each $i \in [1, 2a+2]$.
For each $j \in [1, a]$, there is a path denoted by $P^{(j)}$ as a connected component of the forest which has $j$ nodes, $P^{(j)} := \left(x_1^{(j)},x_2^{(j)},\ldots,x_{j}^{(j)}\right)$.
This shows that the number of nodes in the graph is $(2a+ 3) + \sum_{j=1}^a j = n$.
Assume that there are non-forest edges $\left(x_1^{(i)}, y\right)$ for $i \in [1, a]$ and $(z_{i}, z_{i+1})$ for $i \in [1, 2a+1]$.
\Cref{fig:hard-instance} illustrates the structure of this instance.
It is straightforward to see that $\Delta^\star = a + 1$ in this instance.

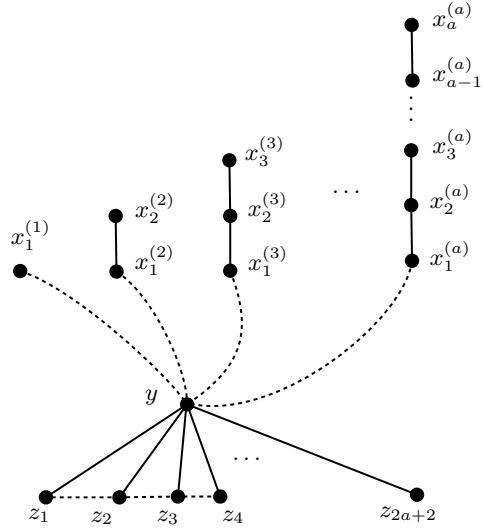
\begin{figure}[ht!]
\caption{A hard instance for the \cite{BFW26} algorithm. Solid edges are forest edges, and dashed edges are non-forest edges.}
    \label{fig:hard-instance}
\centering

\tikzset{every picture/.style={line width=0.75pt}} %

\begin{tikzpicture}[x=0.75pt,y=0.75pt,yscale=-1.25,xscale=1.25]

\draw  [dash pattern={on 1.5pt off 1.5pt}]  (203.33,119.33) .. controls (200.33,143) and (150.81,183.33) .. (115.47,177) ;
\draw    (130.33,101.33) -- (130,79) ;
\draw    (130.33,101.33) -- (130.33,123.33) ;
\draw  [fill={rgb, 255:red, 0; green, 0; blue, 0 }  ,fill opacity=1 ] (127.86,123.33) .. controls (127.86,121.97) and (128.97,120.86) .. (130.33,120.86) .. controls (131.7,120.86) and (132.81,121.97) .. (132.81,123.33) .. controls (132.81,124.7) and (131.7,125.81) .. (130.33,125.81) .. controls (128.97,125.81) and (127.86,124.7) .. (127.86,123.33) -- cycle ;
\draw  [fill={rgb, 255:red, 0; green, 0; blue, 0 }  ,fill opacity=1 ] (127.53,79) .. controls (127.53,77.63) and (128.63,76.53) .. (130,76.53) .. controls (131.37,76.53) and (132.47,77.63) .. (132.47,79) .. controls (132.47,80.37) and (131.37,81.47) .. (130,81.47) .. controls (128.63,81.47) and (127.53,80.37) .. (127.53,79) -- cycle ;
\draw  [fill={rgb, 255:red, 0; green, 0; blue, 0 }  ,fill opacity=1 ] (127.86,101.33) .. controls (127.86,99.97) and (128.97,98.86) .. (130.33,98.86) .. controls (131.7,98.86) and (132.81,99.97) .. (132.81,101.33) .. controls (132.81,102.7) and (131.7,103.81) .. (130.33,103.81) .. controls (128.97,103.81) and (127.86,102.7) .. (127.86,101.33) -- cycle ;
\draw    (84.67,123.67) -- (84.33,101.33) ;
\draw    (113,177) -- (56.33,214.33) ;
\draw  [fill={rgb, 255:red, 0; green, 0; blue, 0 }  ,fill opacity=1 ] (82.19,123.67) .. controls (82.19,122.3) and (83.3,121.19) .. (84.67,121.19) .. controls (86.03,121.19) and (87.14,122.3) .. (87.14,123.67) .. controls (87.14,125.03) and (86.03,126.14) .. (84.67,126.14) .. controls (83.3,126.14) and (82.19,125.03) .. (82.19,123.67) -- cycle ;
\draw  [fill={rgb, 255:red, 0; green, 0; blue, 0 }  ,fill opacity=1 ] (43.53,123.47) .. controls (43.53,122.11) and (44.63,121) .. (46,121) .. controls (47.37,121) and (48.47,122.11) .. (48.47,123.47) .. controls (48.47,124.84) and (47.37,125.95) .. (46,125.95) .. controls (44.63,125.95) and (43.53,124.84) .. (43.53,123.47) -- cycle ;
\draw  [fill={rgb, 255:red, 0; green, 0; blue, 0 }  ,fill opacity=1 ] (81.86,101.33) .. controls (81.86,99.97) and (82.97,98.86) .. (84.33,98.86) .. controls (85.7,98.86) and (86.81,99.97) .. (86.81,101.33) .. controls (86.81,102.7) and (85.7,103.81) .. (84.33,103.81) .. controls (82.97,103.81) and (81.86,102.7) .. (81.86,101.33) -- cycle ;
\draw    (203.33,119.33) -- (203,97) ;
\draw    (203,75) -- (203,97) ;
\draw  [fill={rgb, 255:red, 0; green, 0; blue, 0 }  ,fill opacity=1 ] (200.86,119.33) .. controls (200.86,117.97) and (201.97,116.86) .. (203.33,116.86) .. controls (204.7,116.86) and (205.81,117.97) .. (205.81,119.33) .. controls (205.81,120.7) and (204.7,121.81) .. (203.33,121.81) .. controls (201.97,121.81) and (200.86,120.7) .. (200.86,119.33) -- cycle ;
\draw  [fill={rgb, 255:red, 0; green, 0; blue, 0 }  ,fill opacity=1 ] (200.53,97) .. controls (200.53,95.63) and (201.63,94.53) .. (203,94.53) .. controls (204.37,94.53) and (205.47,95.63) .. (205.47,97) .. controls (205.47,98.37) and (204.37,99.47) .. (203,99.47) .. controls (201.63,99.47) and (200.53,98.37) .. (200.53,97) -- cycle ;
\draw    (203.5,46.97) -- (203.17,24.64) ;
\draw  [fill={rgb, 255:red, 0; green, 0; blue, 0 }  ,fill opacity=1 ] (200.69,24.64) .. controls (200.69,23.27) and (201.8,22.17) .. (203.17,22.17) .. controls (204.53,22.17) and (205.64,23.27) .. (205.64,24.64) .. controls (205.64,26.01) and (204.53,27.11) .. (203.17,27.11) .. controls (201.8,27.11) and (200.69,26.01) .. (200.69,24.64) -- cycle ;
\draw  [fill={rgb, 255:red, 0; green, 0; blue, 0 }  ,fill opacity=1 ] (110.53,177) .. controls (110.53,175.63) and (111.63,174.53) .. (113,174.53) .. controls (114.37,174.53) and (115.47,175.63) .. (115.47,177) .. controls (115.47,178.37) and (114.37,179.47) .. (113,179.47) .. controls (111.63,179.47) and (110.53,178.37) .. (110.53,177) -- cycle ;
\draw    (113,177) -- (126.33,214) ;
\draw    (113,177) -- (85.86,214.33) ;
\draw    (113,177) -- (205.33,213.28) ;
\draw  [fill={rgb, 255:red, 0; green, 0; blue, 0 }  ,fill opacity=1 ] (53.86,214.33) .. controls (53.86,212.97) and (54.97,211.86) .. (56.33,211.86) .. controls (57.7,211.86) and (58.81,212.97) .. (58.81,214.33) .. controls (58.81,215.7) and (57.7,216.81) .. (56.33,216.81) .. controls (54.97,216.81) and (53.86,215.7) .. (53.86,214.33) -- cycle ;
\draw  [fill={rgb, 255:red, 0; green, 0; blue, 0 }  ,fill opacity=1 ] (83.39,214.33) .. controls (83.39,212.97) and (84.49,211.86) .. (85.86,211.86) .. controls (87.23,211.86) and (88.33,212.97) .. (88.33,214.33) .. controls (88.33,215.7) and (87.23,216.81) .. (85.86,216.81) .. controls (84.49,216.81) and (83.39,215.7) .. (83.39,214.33) -- cycle ;
\draw  [fill={rgb, 255:red, 0; green, 0; blue, 0 }  ,fill opacity=1 ] (123.86,214) .. controls (123.86,212.63) and (124.97,211.53) .. (126.33,211.53) .. controls (127.7,211.53) and (128.81,212.63) .. (128.81,214) .. controls (128.81,215.37) and (127.7,216.47) .. (126.33,216.47) .. controls (124.97,216.47) and (123.86,215.37) .. (123.86,214) -- cycle ;
\draw  [fill={rgb, 255:red, 0; green, 0; blue, 0 }  ,fill opacity=1 ] (202.86,213.28) .. controls (202.86,211.91) and (203.97,210.81) .. (205.33,210.81) .. controls (206.7,210.81) and (207.81,211.91) .. (207.81,213.28) .. controls (207.81,214.65) and (206.7,215.75) .. (205.33,215.75) .. controls (203.97,215.75) and (202.86,214.65) .. (202.86,213.28) -- cycle ;
\draw  [dash pattern={on 1.5pt off 1.5pt}]  (130.33,123.33) .. controls (139.33,147) and (136.33,161.33) .. (113,177) ;
\draw  [dash pattern={on 1.5pt off 1.5pt}]  (84.67,123.67) .. controls (99.33,133) and (113.33,163.33) .. (113,177) ;
\draw  [dash pattern={on 1.5pt off 1.5pt}]  (46,123.47) .. controls (67.33,130) and (105.33,165) .. (113,177) ;
\draw  [dash pattern={on 1.5pt off 1.5pt}]  (109.33,214) -- (85.86,214.33) ;
\draw  [dash pattern={on 1.5pt off 1.5pt}]  (85.86,214.33) -- (56.33,214.33) ;
\draw  [fill={rgb, 255:red, 0; green, 0; blue, 0 }  ,fill opacity=1 ] (106.86,214) .. controls (106.86,212.63) and (107.97,211.53) .. (109.33,211.53) .. controls (110.7,211.53) and (111.81,212.63) .. (111.81,214) .. controls (111.81,215.37) and (110.7,216.47) .. (109.33,216.47) .. controls (107.97,216.47) and (106.86,215.37) .. (106.86,214) -- cycle ;
\draw  [dash pattern={on 1.5pt off 1.5pt}]  (108.33,214) -- (126.33,214) ;
\draw  [fill={rgb, 255:red, 0; green, 0; blue, 0 }  ,fill opacity=1 ] (201.03,46.97) .. controls (201.03,45.61) and (202.13,44.5) .. (203.5,44.5) .. controls (204.87,44.5) and (205.97,45.61) .. (205.97,46.97) .. controls (205.97,48.34) and (204.87,49.45) .. (203.5,49.45) .. controls (202.13,49.45) and (201.03,48.34) .. (201.03,46.97) -- cycle ;
\draw  [fill={rgb, 255:red, 0; green, 0; blue, 0 }  ,fill opacity=1 ] (200.53,75) .. controls (200.53,73.63) and (201.63,72.53) .. (203,72.53) .. controls (204.37,72.53) and (205.47,73.63) .. (205.47,75) .. controls (205.47,76.37) and (204.37,77.47) .. (203,77.47) .. controls (201.63,77.47) and (200.53,76.37) .. (200.53,75) -- cycle ;
\draw    (113,174.53) -- (109.33,214) ;

\draw (48,217.87) node [anchor=north west][inner sep=0.75pt]  [font=\footnotesize]  {$z_{1}$};
\draw (170,89.87) node [anchor=north west][inner sep=0.75pt]  [font=\footnotesize]  {$\dotsc $};
\draw (200.5,45.9) node [anchor=north west][inner sep=0.75pt]  [font=\footnotesize]  {$\vdots $};
\draw (130,196.93) node [anchor=north west][inner sep=0.75pt]  [font=\footnotesize]  {$\dotsc $};
\draw (95,169.87) node [anchor=north west][inner sep=0.75pt]  [font=\footnotesize]  {$y$};
\draw (73.1,218.73) node [anchor=north west][inner sep=0.75pt]  [font=\footnotesize]  {$z_{2}$};
\draw (99.6,218.57) node [anchor=north west][inner sep=0.75pt]  [font=\footnotesize]  {$z_{3}$};
\draw (189,216.87) node [anchor=north west][inner sep=0.75pt]  [font=\footnotesize]  {$z_{2a+2}$};
\draw (125.86,218.4) node [anchor=north west][inner sep=0.75pt]  [font=\footnotesize]  {$z_{4}$};
\draw (41,101.87) node [anchor=north west][inner sep=0.75pt]  [font=\footnotesize]  {$x_{1}^{( 1)}$};
\draw (90.81,111.73) node [anchor=north west][inner sep=0.75pt]  [font=\footnotesize]  {$x_{1}^{( 2)}$};
\draw (91,90.87) node [anchor=north west][inner sep=0.75pt]  [font=\footnotesize]  {$x_{2}^{( 2)}$};
\draw (136.33,112.21) node [anchor=north west][inner sep=0.75pt]  [font=\footnotesize]  {$x_{1}^{( 3)}$};
\draw (136.17,90.57) node [anchor=north west][inner sep=0.75pt]  [font=\footnotesize]  {$x_{2}^{( 3)}$};
\draw (135,67.87) node [anchor=north west][inner sep=0.75pt]  [font=\footnotesize]  {$x_{3}^{( 3)}$};
\draw (209.33,110.21) node [anchor=north west][inner sep=0.75pt]  [font=\footnotesize]  {$x_{1}^{( a)}$};
\draw (209.33,87.21) node [anchor=north west][inner sep=0.75pt]  [font=\footnotesize]  {$x_{2}^{( a)}$};
\draw (210.33,65.21) node [anchor=north west][inner sep=0.75pt]  [font=\footnotesize]  {$x_{3}^{( a)}$};
\draw (210.33,14.21) node [anchor=north west][inner sep=0.75pt]  [font=\footnotesize]  {$x_{a}^{( a)}$};
\draw (210.33,36.4) node [anchor=north west][inner sep=0.75pt]  [font=\footnotesize]  {$x_{a-1}^{( a)}$};

\end{tikzpicture}
  
\end{figure}

We show that the algorithm of \cite{BFW26} takes at least $\Omega\!\left(mn^{1/4}\right)$ time in order to build a spanning tree starting from the provided forest.

\medskip
\noindent
\textbf{First Iteration.}
In the first iteration, the algorithm of \cite{BFW26} considers each connected component of the forest as one singleton molecule.
It will then find one augmenting chain of length $1$ that inserts $(z_1,z_2)$ and $\left(x_1^{(1)},y\right)$ into the forest, and removes $(z_2,y)$ from the forest.
Hence, the molecule containing $y$ is already destroyed, and it can not do any further improvement since the degree of $y$ is $2\Delta^\star$ and all of the edges between different connected components of the forest are adjacent to $y$.
Hence, in $\Omega(m)$ time, it can only reduce the number of components by $1$.

\medskip
\noindent
\textbf{Iteration $r \in [2, \sqrt{a}]$.}
In general, assume that after $r \leq \sqrt{a}$ iterations, the algorithm has inserted the edges $(z_1,z_2),(z_2,z_3),\ldots,(z_{r},z_{r+1}), \left(x_1^{(1)},y\right),\ldots,\left(x_1^{(r)},y\right)$ into the forest and has removed the edges $(y,z_2),\ldots,(y,z_{r+1})$ from the forest.
In the next iteration, the algorithm defines $\theta = 20n/f = 10(2a^2 + 5a + 6)/(a+1-r)$.
The size of the connected component containing $y,z_1,z_2,\ldots,z_{2a+2}$, is $ (2a + 3) + r(r+1)/2$.
According to our assumption $r \leq \sqrt{a}$, we have that 
$$ 2a + 3 + r(r+1)/2 <  20(2a^2 + 5a + 6)/(a+1-r) = 2\theta $$
As a result, the same issue remains, and the entire connected component of the forest containing $y$ is considered as one large molecule.
The algorithm finds an augmenting chain of length $1$, and reduces the number of components of $\calF$ by only one unit after inserting $(z_{r+1},z_{r+2})$ and $\left(x_{1}^{(r+1)},y\right)$ into the forest and removing $(y,z_{r+2})$ from the forest.
Again, the molecule containing $y$ is destroyed, and the algorithm can not do any further improvement since the degree of $y$ is $2\Delta^\star$ and all of the edges between different connected components of the forest are adjacent to $y$.
Hence, in $\Omega(m)$ time, it can only reduce the number of components by $1$.

\medskip
\noindent
\textbf{Conclusion.}
This example shows that the algorithm of \cite{BFW26} needs at least $\Omega(\sqrt{a}) = \Omega\!\left(n^{1/4}\right)$ many iterations, where each iteration only reduces the number of components by one and takes $\Omega(m)$ time.
As a result, the algorithm needs at least $\Omega\!\left(mn^{1/4}\right)$ time to build the final spanning tree.
This is far from being near-linear $\Tilde{O}(m)$.

\subsection{Our Techniques}\label{sec:techniques-overview}

Recall the degree reduction subroutine (\Cref{lem:degree-reduction}) as described in \Cref{sec:O(mn)-time}.
If the degree reduction subroutine is applied on a node $u$ w.r.t.~region $C$, we refer to $C$ as a \textit{touched} region.
For a touched region $C$, the algorithm can no longer guarantee the degree reduction subroutine for any $v \in C$.
Now, recall Steps I and II of the \cite{FR92} algorithm in \Cref{sec:O(mn)-time}.
In the following, we highlight the main ideas to improve \cite{FR92} and \cite{BFW26}.

\subsubsection{Potential Improvements Over the Main Steps of \cite{FR92}}

\medskip
\noindent
\textbf{Discovering Step.}
In the discovering step, an edge $(x,y)$ between two regions can provide a potential improvement for the algorithm (by growing regions in Step II and discovering more reducible nodes).
A simple observation is that, for a touched region $C$, any node $v \in C$ whose degree is already at most $k-1$ does not need any degree reduction.
Hence, we consider a larger set of potential edges in the discovering step.
We refer to these edges as \textit{effective} edges (formally defined in \Cref{full:def:effective-edge}).
A similar idea has been exploited in \cite{BFW26} algorithm as well.

\medskip
\noindent
\textbf{Merging Step.}
In the merging step (after discovering an effective edge $(x,y)$), the algorithm merges all regions in the boundary of $P^\calF_{(x,y)}$.
This step is quite naive for the following reason;
Consider a scenario where $x$ belongs to a region $C_x^\old$ of small size (say $|C_x^\old| = O(1)$).
If the algorithm keeps merging the regions naively, after some iterations, the new region containing $x$ (denoted by $C_x^\new$) can be very large (even $|C_x^\new| = \Omega(n)$).
Now, assume that the algorithm applies the degree reduction subroutine on $x$.
Although the degree of $x$ can be reduced only by updating $C_x^\old$, the construction of $C_x^\new$ required that all nodes inside $C_x^\old$ were reducible.
So, after this degree reduction on $x$, the entire $C_x^\new$ gets touched.

As a result, if we can decide whether a degree reduction subroutine for $x$ is useful before growing the region containing $x$, we can do the degree reduction on $x$ at that time which touches only a small number of nodes. 
To utilize the above idea, we introduce a critical object called an augmenting chain (formally defined in \Cref{full:def:aug-chain}).
We describe this object intuitively in \Cref{sec:aug-chain-overview}, how to use them in \Cref{sec:use-aug-chain}, and eventually how to find them in \Cref{sec:find-chain-overview}.

\subsubsection{Augmenting Chains}\label{sec:aug-chain-overview}

We consider the same name `augmenting chain' as in \cite{BFW26} since our notion of an augmenting chain is a non-trivial extension of the same object in \cite{BFW26}.
A high-level intuition of this object is provided below.
Assume that $(x, y)$ is an effective edge.
There are two cases for $(x,y)$;

\medskip
\noindent
\textbf{Case I:}
$x$ and $y$ are in different connected components of $\calF$.
We can just apply the degree reduction subroutine on $x$ and $y$ separately, and then add $(x,y)$ to $\calF$ in order to reduce the number of components of the forest.
In this case, we call $(x,y)$ an augmenting chain of length $1$.

\medskip
\noindent
\textbf{Case II:} 
$x$ and $y$ are in the same connected component of $\calF$.
Let $z \in P^\calF_{x,y}$ be a node of degree $k$.
We can do the following: 1) apply the degree reduction subroutine on $x$ and $y$ separately, 2) add $(x,y)$ to the forest, and 3) remove an edge incident on $z$ on the path $P^\calF_{x,y}$.
If we do such an operation, we have a feasible forest, where the degree of $z$ is reduced to $k-1$.
Now, it is possible that a new edge $(w, z)$ becomes effective, in which case, we can repeat the above process on $(w,z)$.
Eventually, we might end up in case I, which reduces the number of components of the forest.

\medskip
In essence, an augmenting chain is a sequence of non-forest edges ($\notin E(\calF)$) derived from the above process.
We refer to the length of this sequence as the \textit{length} of the augmenting chain.
Moreover, the above process that updates the forest according to an augmenting chain is called the process of \textit{applying an augmenting chain}.
In \Cref{sec:use-aug-chain}, we link how this object resolve the issue of \cite{FR92}.

\medskip
\noindent
\textbf{Comparison to \cite{BFW26}.}
The augmenting chains introduced in \cite{BFW26} depend on the forest decomposition in their algorithm.
As explained in \Cref{sec:BFW-barriers}, this decomposition of the forest introduces a hard barrier towards achieving fast algorithms (see Appendix \ref{appendix:hard-instance} for a concrete example), which means that we need a stronger object independent of such decomposition.
For this reason, our definition of an augmenting chain described above is a non-trivial refinement of \cite{BFW26} to avoid creating cycles in the forest.

\subsubsection{Incorporating Augmenting Chains In the Algorithm}\label{sec:use-aug-chain}

Using our terminology, the \cite{FR92} algorithm as described in \Cref{sec:O(mn)-time}, simply keeps applying the merging steps until an augmenting chain of length 1 appears.
Thus, regions are often too large when they become touched.
Instead, our algorithm will prioritize finding longer augmenting chains over merging. This encourages the algorithm not to merge regions too aggressively.
On the other hand, an augmenting chain of length $h$, touches $\Theta(h)$ many regions (see the intuitive description in \Cref{sec:aug-chain-overview}).
Hence, we need to deliberately choose the length $h$ of the augmenting chains that we apply since 1) if $h$ is large, the \textit{number} of touched regions becomes large, and 2) if $h$ is small, the \textit{size} of the touched regions might be large.
Both cases lead to little progress in the algorithm.
To address this, we run our algorithm in rounds, where the goal of round $t$ is to find a large set of augmenting chains, each of length $\leq t$.

\medskip
\noindent
\textbf{Final Length of Augmenting Chains.}
Eventually, to achieve multiplicative $(1+\varepsilon)$-approximate MDST, it is sufficient to run $H = \Theta(\log_{1+\varepsilon} n)$ many rounds of our algorithm to reduce the number of components of the forest by $\Omega(f/H)$.
To achieve a plus one additive approximate MDST, it is sufficient to run $H = \Theta(n/f)$ rounds of our algorithm, to reduce the number of components of the forest by $\Omega(f/H)$.
These are aligned with the guarantees of \Cref{lem:main} as desired.

\subsubsection{Finding Augmenting Chains}\label{sec:find-chain-overview}

Augmenting chains are complicated objects, where applying an augmenting chain can destroy or create new augmenting chains since the forest is changing in a non-trivial way.
As a result, searching for augmenting chains is a non-trivial task due to the limited running time of the algorithm.

\medskip
\noindent
\textbf{Levels of the Edges.}
We maintain an integer for each edge of the graph called the \textit{level} of the edge.
If the level of an effective edge $(u,v)$ is $r$, it intuitively means that the length of the minimum augmenting chain whose last edge is $(u,v)$ should be at least $r$.

\medskip
\noindent
\textbf{Searching for Augmenting Chains.}
We provide a deliberate subroutine $\findChain(h, (u,v))$ that receives an integer $h$ and an effective edge $(u,v)$ and tries to find an augmenting chain of length at most $h$.
This subroutine searches for a sequence of effective edges whose levels are \textit{strictly increasing}, and upon any failure, it increases the levels of the corresponding edges.

In round $t$ of the algorithm, we bound the level of the edges by $t$ to make sure that $\findChain$ only finds augmenting chains of length $\leq t$.
Moreover, the total time spent for finding augmenting chains is proportional to the summation of the levels of all edges, which is $O(m \cdot H)$, aligned with the guarantee of \Cref{lem:main}.

\medskip
\noindent
\textbf{Comparison to \cite{BFW26}.}
The algorithm of \cite{BFW26} searches for augmenting chains as follows: Assuming that no augmenting chain of length $\leq h-1$ exists, it first constructs a layering of the nodes (analogous to the notion of levels in our algorithm), and then searches for a blocking set of augmenting chains of length $=h$ by a DFS-type algorithm going down the levels.
There are strong barriers to using their algorithm here, since 1) this layering of the nodes depends completely on [their decomposition of the forest] and [the assumption that no augmenting chain of length $\leq h-1$ exists], 2) their algorithm is static in the sense that they search for a blocking set of augmenting chains independently and then apply all of them at the end.

On the other hand, our notion of an augmenting chain is independent of any decomposition, and these objects are super dynamic in the sense that applying an augmenting chain might introduce a new augmenting chain, even of length $1$.
As a result, we need an intricate subroutine that finds a large number of augmenting chains in a dynamic way and with limited running time.

\end{document}